\documentclass{llncs}
\usepackage[dvipsnames,svgnames,x11names,hyperref]{xcolor}
\usepackage{amsmath}
\usepackage{xspace}
\usepackage{float}
\usepackage[T1]{fontenc}
\usepackage{mdframed}
\usepackage{mathtools}
\usepackage{enumitem}

\usepackage[linesnumbered,ruled]{algorithm2e}

\usepackage{algpseudocode}

\usepackage{thmtools}
\usepackage[colorlinks]{hyperref}
\hypersetup{colorlinks={true},linkcolor={blue},citecolor=magenta}

\usepackage{amssymb,amsfonts,amsmath,
	amscd,dsfont,mathrsfs}
\usepackage{complexity}
\usepackage{tikz}
\usetikzlibrary{decorations.pathreplacing,backgrounds}
\usepackage{multirow}
\usepackage{braket}
\usepackage{todonotes}
\usepackage{comment}
\usepackage{quantikz}

\usepackage{fullpage}

\usepackage{url}
\usepackage{bbm}
\usepackage[capitalize]{cleveref}
\hypersetup{breaklinks=true}

\renewcommand{\set}[1]{\left\{ #1 \right\}}
\newcommand{\abs}[1]{\left| #1 \right|}

\newcommand{\Id}{\mathbbm{1}}

\renewcommand{\proj}[1]{\ensuremath{|#1\rangle \langle #1|}}

\newcommand{\norm}[1]{\left\Vert {#1} \right\Vert}
 
\newcommand{\Tr}{\mathrm{Tr}}

\newcommand{\bit}{\{0,1\}}

\newcommand{\algo}{\mathcal}

\newcommand{\dbs}{{D^1D^2D^3}}
\newcommand{\keys}{{K_1K_2K_3}}
\newcommand{\dbswithc}{{D^1,D^2,D^3}}
\newcommand{\ppar}[1]{{^{(#1)}}}
\newcommand{\vf}{{V_\mathsf{fll}}}

\newcommand{\vr}{{V_\mathsf{rmv}}}
\newcommand{\dom}[1]{{\mathsf{Dom}(#1)}}
\newcommand{\im}[1]{{\mathsf{Im}(#1)}}
\newcommand{\cco}{\mathsf{cO}} 
\newcommand{\cp}{\mathsf{cpO}} 
\newcommand{\fc}{\mathsf{fC}} 
\newcommand{\pc}{\mathsf{pC}} 
\newcommand{\pu}{\mathsf{P}}
\newcommand{\fp}{\mathsf{F}}
\newcommand{\sw}{\mathsf{S}}
\newcommand{\dm}{\mathsf{M}}
\newcommand{\w}{\mathsf{W}}

\newcommand{\pkac}{E^{P_1, P_2}_{k_1,k_2,k_3}}
\newcommand{\vp}{{V_\mathsf{p}}}

\newcommand{\vpc}{{V_\mathsf{pc}}}
\newcommand{\vfpc}{{V_\mathsf{pc,fll}}}
\newcommand{\vrpc}{{V_\mathsf{pc,rmv}}}

\newcommand{\Pmiss}[3]{\Pi^{#1}_{#2,#3}}   

\newif\ifShowComments
\ShowCommentstrue
\newif\ifShowAuthors
\ShowAuthorstrue

\ifShowComments
\newcommand{\joe}[1]{{\noindent \textcolor{teal}{\emph{(Joe:  #1)}}}{}}
\newcommand{\cm}[1]{{\noindent \textcolor{red}{\emph{(Chris:  #1)}}}{}}
\newcommand{\st}[1]{{\noindent \textcolor{Orchid}{\emph{(Saliha:  #1)}}}{}}
\newcommand{\ga}[1]{{\noindent \textcolor{orange}{\emph{(Gorjan:  #1)}}}{}}
\else
\newcommand{\joe}[1]{}
\newcommand{\cm}[1]{}
\newcommand{\st}[1]{}
\newcommand{\ga}[1]{}
\fi

\everymath{\displaystyle}
\allowdisplaybreaks

\title{On The Simplest Quantum-Secure Block Cipher}
\ifShowAuthors
\author{Gorjan Alagic\inst{1,2} \and Joseph Carolan\inst{1} \and Christian Majenz\inst{3} \and Saliha Tokat\inst{3}}
\institute{University of Maryland \and National Institute of Standards and Technology \and Technical University of Denmark}
\else
\author{}
\institute{}
\fi

\begin{document}

\maketitle

\begin{abstract}
Pseudorandom permutations are ubiquitous in theoretical and applied cryptography. PRPs that offer security even against adversaries making quantum queries are of increasing interest, and used in applications ranging from constructing pseudorandom unitaries to separating SZK from BQP.

A successful framework for constructing classically-secure PRPs is the key-alternating Even-Mansour approach, which interleaves applications of public permutations with additions of round keys. The single-round construction is already classically secure in the ideal permutation model (IPM), with added rounds offering improved concrete security. However, in the quantum-query setting, the status of this framework is presently unclear. A simple quantum-query attack based on Simon's algorithm breaks the one-round cipher. For two or more rounds, security is only known against non-adaptive adversaries who must prepare all queries in advance.

In this work, we show that the two-round Even-Mansour cipher is information theoretically secure in the IPM against adversaries making polynomially-many adaptive forward and inverse quantum queries to all available oracles. Our proof uses compressed permutation oracles and a specially crafted isometry relating the ideal and real experiments. We also show that this construction is minimal, in the sense that essentially any cipher constructed via a single call to a public permutation is quantumly insecure. 

\end{abstract}

\section{Introduction}

A Key-Alternating Cipher (KAC) constructs a block cipher by alternating public permutation layers with secret key additions. The prototypical example is the $t$-round Even-Mansour construction~\cite{EM97,BKLSST12}, defined by 
\begin{equation}\label{eq:intro:em-multi}
    E_{k}(x) = P_t( \cdots P_2(P_1(x \oplus k_1) \oplus k_2) \cdots) \oplus k_{t+1}
\end{equation}
where $k = (k_1, k_2, \cdots, k_{t+1})$. Here each $P_j$ is a permutation of $\{0,1\}^n$, and each $k_j$ is sampled uniformly at random from $\{0,1\}^n$. Practical constructions such as AES follow a similar pattern to \eqref{eq:intro:em-multi}, with concrete public permutations and key additions performed according to a certain key schedule. 

\paragraph{Theoretical security and construction minimalism.} Studying the theoretical, provable security of concrete constructions such as AES has turned out to be difficult. A standard approach is thus to work in an appropriate idealized model~\cite{BR93,EM97,Shannon49,BRS02}. In the case of \eqref{eq:intro:em-multi}, the public permutations $P_j$ are then sampled uniformly at random, and access to $E_k$ and the $P_j$ is via (bidirectional) oracles only; this is the so-called \emph{ideal permutation model} (IPM). The first-order task is then to ascertain the number of queries a computationally unbounded adversary needs to distinguish this ``real world'' from an ``ideal world'' in which $E_k$ is replaced with an independently sampled permutation $R$. In this context, a natural question is how complex a KAC needs to be in order to achieve security, with \emph{round count} arguably being the most important metric. In the classical security setting, this question is settled: the one-round construction 
\begin{equation}\label{eq:intro:em}
    E_{k}(x) = P(x \oplus k_1) \oplus k_2
\end{equation}
is already secure, up to $O(2^{n/2})$ queries~\cite{EM97,DKS12}. 
This is not just an important theoretical result: the core idea of \eqref{eq:intro:em} is used in numerous practical applications, ranging from XTS-AES disk encryption to the Chaskey MAC~\cite{Rogaway04,Dworkin10XTS,MouhaEtAl14}.

\paragraph{The quantum-query setting.} The transition to quantum-resistant cryptography requires assessing the security of KAC constructions against quantum adversaries. An early surprise in this setting was that the standard one-round construction \eqref{eq:intro:em} has a quantum vulnerability: Simon's algorithm extracts the key using only $O(n)$ quantum queries to $E_k \oplus P$~\cite{KM12,KLLNP16}. Importantly, this does not yield a quantum-computational attack on real-world symmetric-key constructions like XTS-AES. Indeed, in those settings the key $k$ lives on a classical machine. The appropriate interface model should thus only grant \emph{classical} access to $E_k$, while retaining quantum access to $P$. Simon's algorithm is not applicable in this ``Q1 model'', and subsequent work showed that one-round Even-Mansour is secure in Q1, up to $O(2^{n/3})$ queries~\cite{ABKM22}. Moreover, the major real-world applications of \eqref{eq:intro:em} survive in the appropriate post-quantum setting (see, e.g.,~\cite{ABMS26}). 

At the same time, the security of $t$-round Even-Mansour for $t > 1$ in the setting where all oracles can be queried quantumly (called the ``Q2 model'') is unknown\footnote{Bai, Esmaili, and Mantri showed Q2 security for $t \geq 2$ against non-adaptive adversaries who must prepare all oracle queries prior to receiving any responses \cite{BEM25}.}. This leaves open the possibility that, in the quantum setting, the idea of constructing pseudorandom permutations via key alternation is \emph{structurally unsound}---regardless of the number of rounds---and only grants security when the Q1 interface is forced. This would mean that the KAC framework is unsuitable for the numerous applications requiring Q2-secure PRPs (or qPRPs), such as:
\begin{itemize}
    \item constructions of quantum pseudorandomness (e.g., unitaries, states, entanglement) and quantum money~\cite{MPSY24,MH25,FLMNW26,JMW24,BBSS23,JLS18}.
    \item classical encryption and authentication secure against superposition attacks~\cite{GHS16,AJOP20,AMRS20}.
    \item multi-copy public-key encryption of quantum data~\cite{AGKL24}.
    \item two-prover protocols for classical verification of quantum
    depth~\cite{CH22}.
    \item Post-quantum security proofs for hashing and signatures~\cite{Zhandry21,BHRV21}.
    \item efficient oracle separation between $\mathsf{SZK}$ and $\mathsf{BQP}$~\cite{AC17}.
\end{itemize}

\paragraph{Our results.} In this work, we show that the security of the key-alternating approach in the quantum setting does not rest fundamentally on the choice of interface model: two rounds suffice. Specifically, we show that the two-round Even-Mansour construction (i.e., \eqref{eq:intro:em-multi} for $t=2$) is a strong qPRP: it is indistinguishable from a truly random permutation in the Q2 model to any polynomial-query adversary. We also show that this is the best one can hope for, by giving a generic polynomial-quantum-query attack against any IPM block cipher construction that queries $P$ only once\footnote{Specifically, any construction which is itself a permutation of the same size as $P$.}. Arguably, two-round Even-Mansour is also quite simple when compared to known non-KAC qPRP constructions, such as the seven-round Feistel construction~\cite{Carolan26}.

\subsection{Technical summary of results}

We now give a more technical summary of our results and methods. As discussed above, our goal is to understand information-theoretically-secure block cipher (i.e., qPRP) constructions in the public permutation model. In our Q2 setting, all oracles are quantumly accessible. 

Specifically, let $E$ be a keyed cipher construction from a collection of public permutations $\{P_j\}$. The adversary gets oracle access to either (i.) $\{P_j\}$ and $E_k$ for uniformly random $k$ (the real world), or (ii.) $\{P_j\}$ and an independently sampled permutation $R$ (the ideal world). The adversary's task is to distinguish these two cases with non-negligible advantage. Access to all oracles is given via the standard quantum black-box oracle in both the forward and inverse directions. So, for example, access to $P_1$ means black-box access to the unitaries
\begin{equation}\label{eq:intro:quantum-oracle}
    \ket{x}\ket{y} \mapsto \ket{x}\ket{y \oplus P_1(x)}
    \qquad \text{and} \qquad
    \ket{x}\ket{y} \mapsto \ket{x}\ket{y \oplus P_1^{-1}(x)}\,.
\end{equation}
The adversary is granted unbounded computational power, and we are interested only in measuring their success probability as a function of the total number of queries made to the three oracles.
\subsubsection{All one-query constructions are insecure.}

We begin by showing that any construction that uses only a single query to the public permutation $P$ is insecure, provided the constructed cipher has a sufficiently large domain. This should be contrasted with the classical case, where one-round Even-Mansour provides a secure construction~\cite{EM97}. As discussed above, this is easily broken in Q2 using Simon's algorithm~\cite{KM12}. However, this attack takes crucial advantage of the structure of Even-Mansour~\cite{AR17}, and leaves open the possibility of an alternative one-query construction that is in fact secure.

A deterministic one-query cipher $Q:[D]\to[D]$ constructed from a public permutation $P:[N]\to[N]$ and secret key $k$, using one forward query to $P$, is without loss of generality of the form $Q_k^P(x)=b_k(x,P(a_k(x))),$
for publicly known keyed functions $a_k:[D]\to[N]$ and $b_k:[D]\times[N]\to[D]$ independent of $P$.
We require that $Q_k^P$ is a permutation of $[D]$ for every $P\in S_N$ and every key $k$. In particular, the query input $a_k(x)$ need not depend on the entirety of $x$, and the block sizes of $P$ and $Q$ need not agree. We take $K$ to represent the universe of keys, meaning $k\in K$, and write
\[
    N=2^n,\qquad D=2^m,\qquad
    \kappa=\lceil\log_2|K|\rceil=\operatorname{poly}(n).
\]
We assume that the domain of the cipher is not of polynomial size, as in this case the key may itself contain enough entropy to sample a random cipher without using $P$. In particular, we assume
\begin{equation}
    \kappa+1+\frac nm=o(D).
    \label{eq:one-query-size-regime}
\end{equation}
The adversary has quantum access to both directions of $P$ and $Q$, as above, although the attack below uses only forward queries. Our distinguisher uses $2t$ queries to $P$ and $Q$ to prepare $t$ copies of each of their Choi states, i.e.,
\begin{align}
    \ket{\psi_{P,Q}}
    &\coloneqq
    \left(
        \frac{1}{\sqrt N}\sum_{u\in[N]}\ket{u,P(u)}
    \right)^{\otimes t}
    \otimes
    \left(
        \frac{1}{\sqrt D}\sum_{x\in[D]}\ket{x,Q(x)}
    \right)^{\otimes t}.
    \label{eqn:choi-state}
\end{align}

We show that this is a valid distinguisher by showing that the ideal and real distributions on these states are information-theoretically distinguishable when the number of copies $t$ is $O(\kappa+1+n/m)$. In the ideal experiment, $P\sim S_N$ and $Q\sim S_D$ are independent.

\begin{lemma}
    Under \eqref{eq:one-query-size-regime}, the quantum distinguisher described above distinguishes any single-query block cipher construction from random with constant advantage using $O(\kappa+1+n/m)$ oracle queries in total. In particular, $O(\kappa+1)$ queries suffice when $m=\Omega(n)$.
    \label{lem:impossibility-result}
\end{lemma}

We prove the above in \Cref{app:one-query}. Note that this distinguisher is not in general computationally efficient. However, this is inherent if computationally-secure qPRPs exist: the equal-size case includes $Q_k^P=P\circ\sigma_k$ where $\sigma_k$ is such a PRP~\cite{Zhandry25PRP}. Such a construction can be broken via offline key search and is thus not information-theoretically secure. The result still requires that the constructed $Q_k^P$ is actually a permutation rather than indistinguishable from one. We leave removing this technical limitation to future work. The domain-size assumption is also necessary in some form: for a one-bit cipher and a uniform key $k\in[N]$, the one-query construction $Q_k^P(x)=x\oplus\operatorname{lsb}(P(k))$ is exactly a uniform permutation of two points independent of $P$.

\subsubsection{Two-round Even-Mansour is a qPRP.}

Our main result is that the two-round Even-Mansour construction is a qPRP in the ideal permutation model. The construction is as follows. Sample uniformly random permutations $P_1$ and $P_2$ of $\{0,1\}^n$ and uniformly random keys $k_1, k_2, k_3 \in \{0,1\}^n$, and define
\begin{equation}\label{eq:intro:em-2}
    \pkac(x) = P_2(P_1(x\oplus k_1)\oplus k_2)\oplus k_3.
\end{equation}
When the choices of permutation and key have been fixed, we will write $E_k$ in place of $\pkac$ to simplify notation. As discussed above, the adversary gets oracles for $P_1$ and $P_2$, and either $E_k$ (in the real world), or an independently sampled permutation $R$ (in the ideal world). Access to all oracles is quantum and bidirectional. Our main result is then the following.

\begin{theorem}\label{thm:intro:main} Let $\mathcal{A}$ be a computationally unbounded quantum algorithm making at most $q \leq \sqrt{N/2}$ queries to its oracles. Then
\begin{align*}
   \abs{\Pr[\algo A^{P_1, P_2, P_3}()=1]-\Pr[\algo A^{P_1, P_2, \pkac}()=1] } \leq O\!\left(\sqrt{q^3/N}\right)+ \varepsilon_{\textsf{CPO}}\,,
\end{align*} 
where $\varepsilon_{CPO} \leq O(q/\sqrt{N})$ is the soundness of the compressed permutation oracle simulation. 
\end{theorem}

The best-known quantum-query attacks on two-round Even-Mansour have a query complexity of $\tilde O(2^{n/2})$. The simplest of these can be viewed as a Grover search for $k_3$; this leverages the fact that, given $k_3$ and inverse access to $P_2$, recovering the remainder of the key can be done using the Simon-based one-round attack~\cite{CGL22,Zhang24}.

We remark that our bound is the sum of two terms, one of which is the soundness error upper bound for Carolan's compressed permutation oracle~\cite{Carolan26,CM26}.

\paragraph{Technical challenges.}

We now briefly outline some technical challenges in establishing \Cref{thm:intro:main}. We first remark that translating proof techniques designed for the classical or Q1 setting does not seem promising. Indeed, the classical and Q1 security of two-round Even-Mansour (against poly-query adversaries) follows immediately from the respective security of the one-round construction---which is insecure in the Q2 setting. This transition from fundamentally insecure to fundamentally secure with the addition of a second round is thus a new phenomenon to contend with. Moreover, some variants of the two-round construction, such as setting $P_1 = P_2$ and $k_1=k_2=k_3$ are actually \emph{still vulnerable} to a Simon-style $O(n)$-query key-recovery attack~\cite{KLLNP16}. This also appears to be a uniquely Q2 phenomenon: whether keys and permutations are independent or correlated across rounds can affect the concrete classical security bound, but generally does not turn secure constructions into insecure ones.

As with most quantum-query idealized-model lower bounds, translating classical proofs to our setting necessitates overcoming the challenge of recording and inspecting queries. While the recent development of compressed permutation oracles~\cite{Carolan26} provides an important tool in this context, many challenges remain.

\paragraph{High-level proof ideas.}

To prove that the two-round Even-Mansour construction is indistinguishable from a random permutation, we consider an adversary with oracle access to one of the two, together with oracle access to the inner permutations $P_1$ and $P_2$. Since permutations are invertible, the adversary may query all of these in both directions. We instantiate every uniformly random permutation using the compressed permutation oracle (CPO) technique, which records the queries in a database.

We connect the two worlds by a hybrid argument in \Cref{thm:hybrid-diff}. The \emph{ideal world} is the one in which the adversary is given a random permutation $P_3$ for the whole experiment; the \emph{real world} is the one in which it is given the two-round Even-Mansour construction throughout. To interpolate between them, we define an isometry from ideal world states to real world states, and a sequence of hybrids in which the ideal-world oracle is replaced by the real one, one query at a time, after applying the isometry.

The isometry $V=\vr\fp_{D^2}\vf\fp_{D^2}$ defined in \Cref{def:isometry-k1-k2-k3} keeps the databases recording queries to the inner permutations and re-expresses the $P_3$ queries in terms of the databases of $P_1$ and $P_2$ together with the keys (by $\vf$). Concretely, each entry in the $P_3$ database induces the entries in the $P_1$ and $P_2$ databases that a two-round Even-Mansour evaluation of that query would have produced. This is done in a single step, using the flipped database (by $\fp_{D^2}$) of $P_2$. In a second step, the isometry uncomputes the corresponding $P_3$ entry (by $\vr$), since in the following hybrid that query is answered by the construction and therefore no longer appears in the ideal database. 

We bound the difference between consecutive hybrids for forward queries only, having first reduced inverse queries to forward ones in \Cref{thm:inv-to-fwd-hybrid}. The reduction uses an operator $\dm$, which swaps and flips the database registers so as to prepare the databases in the form on which a forward query performs the work of an
inverse one. For this reduction we must also bound the commutator of $\dm$ with the isometry; this commutator vanishes on \emph{good} databases as in \Cref{def:pi-g}, i.e., those for which no output of $P_1$ differs from an input of $P_2$ by the middle key $k_2$. The reason is that the isometry is defined symmetrically for forward and inverse queries: applying it to a database in standard orientation agrees with applying it to the $\dm$-image of that database.

To bound the hybrid difference for forward queries, we treat each of the three oracles separately in \Cref{thm:p1-bound,thm:p2-bound}, \ref{thm:p3-bound}. In addition to the database projectors interleaved with the queries, we apply further projectors --- also controlled on the query register $X$ --- immediately before the differing oracles, in order to force the two hybrids to return the same answer. For $P_1$ and $P_2$ queries these projectors ensure that the isometry leaves the database register holding the current query untouched, and does not link it to the other registers; consequently the behaviour of the isometry elsewhere does not depend on what that register holds. For $P_3$ queries the situation is reversed: the projectors ensure that the isometry \emph{carries} the value at the ideal world register to the associated real world registers, since this is precisely the query for which the ideal and real databases differ.

For $P_3$ queries, we split the isometry into two parts controlled on the query register as in \Cref{def:vp-vpc}, discarding the complement since it plays no role in propagating the queried database register output to the $P_1$ and $P_2$ databases. This splitting is only possible for $P_3$. The first part of the isometry is controlled on the queried database register, so the subsequent process depends on what happened to the database value after the query — notably, uncomputation or a trivial (non-defining) query are possibilities. We therefore perform a case analysis with respect to the change on this register, employing the Hadamard basis representation, in which computation and uncomputation are clearly distinguished.

Our choice of good databases is particularly useful in orthogonality-preserving arguments. Good databases enforce a directional structure on the isometry: in the first part (filling the $P_1$ and $P_2$ databases), they force the isometry to either fill an entry or act trivially, ruling out uncomputation; in the second part (removing the ideal world queries), they force it to remove an entry rather than define a new one. This directional structure is precisely what ensures orthogonality preservation: on good databases, the isometry generates only positive coefficient branches.

Finally, a remark on our use of the CPO given in \Cref{sec:cpo}. Throughout, we replace the
compression and decompression operators of the permutation oracle by their function-oracle counterparts, and maintain injectivity of the databases by means of database projectors instead. This is possible because the two pairs of operators differ only slightly, as shown in \Cref{lem:diff-full-and-partial-decomp}. The resulting loss does not worsen the bound obtained from the hybrid difference.

\section{Preliminaries}

We use $[N]=\set{0,\dots,N-1}\cong\set{0,1}^n$ where $N=2^n$. We employ the ordered product as
\begin{equation}
\prod_{x\in [N]} A\ppar{x} \;=\; A\ppar{N-1} \cdots A\ppar{1} A\ppar{0}.
\end{equation}
We use $S \oplus k$ to denote the subset $S$ of $n$-bit strings shifted by $k \in \bit^n$, i.e., $S \oplus k \coloneqq \{\, s \oplus k : s \in S \,\}$ for $S \subseteq \bit^n$ and $k \in \bit^n.$ Throughout, $\im{D}$ and $\dom{D}$ refer to the image and domain of a database viewed as a partial function. For $x\in[N]$, $\bot\oplus x$ is interpreted as $\bot$, hence it is neither in the domain nor the image of any database. We use $\mathsf{Tup}(J,S)$ to denote the set of tuples of size $|S|$, indexed by $S$, with elements in $J$. The state $\ket{+^n}$ is the uniform superposition over $[N]$.

\begin{lemma}{\cite[Lemma~2.2]{Carolan26}}
    Let $\ket{\psi} \in \algo H_{ A} \otimes \algo H_{ B}$ and $\ket{\phi} \in \algo H_{ A} \otimes \algo H_{ C}$ be normalized quantum states, and let $V : \algo H_{ C} \rightarrow \algo H_{ B}$ be an isometry. Let $\sigma_A=\Tr_B[\proj{\psi}_{AB}]$ and $\tau_{ A}=\Tr_{ C}[\proj{\phi}_{AC}]$. Then we have \begin{align}
        \frac{1}{2}\norm{\sigma_{ A} - \tau_{ A}}_1 &\leq \norm{\ket{\psi}_{{AB}} - V_{ C} \ket{\phi}_{{AC}}}.
    \end{align}
    \label{lem:approx-uhlman}
\end{lemma}

\begin{definition}[Database]
    We call $D$ a database if it is a partial function from $[N]$ to $[N]$, equivalently, a total function from $[N]$ to $[N] \cup \{\bot\}$. We write $\mathbf D$ for the set of all databases, $\dom{D}=\{x\in[N]: D(x)\neq \bot \}$, and $\im{D}=\{D(x)\in[N]: D(x)\neq \bot \}$. The size of the database is $\abs{D}\coloneqq\abs{\dom{D}}$.
\end{definition}
A single database projector written without a register subscript is applied to all
database registers, i.e., $\Pi^{a} \coloneqq \Pi^{a}_{D^1}\Pi^{a}_{D^2}\Pi^{a}_{D^3}$.
For a projector $\Pi^{a}$ we write $\mathbf{D}_{a}$ for its range, and for
several projectors we list their labels, so that $\mathbf{D}_{a,b} \coloneqq \mathbf{D}_{a} \cap \mathbf{D}_{b}$. 
\begin{definition}
    The flip operator and swap operator are defined as
\begin{align}
    \fp\ket{D} \coloneqq \begin{cases}
        \ket{D^{-1}} & \text{if $D$ is injective},\\
        \ket{D} & \text{otherwise},
    \end{cases}
    \qquad
    \sw\ket{a,b}_{AB} \coloneqq \ket{b,a}_{AB}.
\end{align}
\end{definition}

\subsection{Compressed Permutation Oracle}\label{sec:cpo}

We recall the compressed permutation oracle of
\cite{Carolan26}. Let $N=2^n$ and let $\varphi$ be
uniformly distributed over $S_N$. Its bidirectional quantum oracle is
\[
    \mathcal O_\varphi\ket{b,x,y}_{BXY}
    =\ket{b,x,y\oplus\varphi^{\,1-2b}(x)}_{BXY}.
\]
The compressed oracle maintains a private register $D$ whose basis
states are injective partial functions $D:[N]\hookrightarrow [N]$, where $D$ is thought of as a truth table with undefined entries represented by a $\bot$ symbol. The database is spanned by such truth tables, and is initialized to $\ket{\boldsymbol\bot}=\ket{\bot}^{\otimes N}$.

For an injective database $D$ with $x\notin\dom{D}$, write
$D[x\rightarrow y]$ for its extension by $(x,y)$ and define
\[
    \ket{+_{x,D}}
    \coloneqq
    \frac{1}{\sqrt{N-|D|}}
    \sum_{y\notin\im{D}}\ket{D[x\rightarrow y]}.
\]
On the span of $\ket{D}$ and its extensions at $x$, define
\begin{align}
    \pc_{x,D}
    &\coloneqq \operatorname{Exc}(\ket{D}, \ket{+_{x,D}}), & \operatorname{Exc}(\ket{A}, \ket{B}) &=
    \Id-\ket{A}\bra{A}
       -\ket{B}\bra{B}
       +\ket{A}\bra{B}
       +\ket{B}\bra{A}. \nonumber
\end{align}
Taking the direct sum over such databases gives $\pc_x$;
set $\pc_{XD}=\sum_x\ket{x}\bra{x}_X\otimes\pc_x$.
Define the evaluation and inversion operators by
\[
    \pu\ket{x,y}\ket{D}
    \coloneqq
    \begin{cases}
        \ket{x,y\oplus D(x)}\ket{D}, & x\in\dom{D},\\
        \ket{x,y}\ket{D},           & x\notin\dom{D}.
    \end{cases}
\]
The compressed permutation oracle is formally
\[
    \cp
    \coloneqq
    \sum_{b\in\{0,1\}}\ket{b}\bra{b}_B
    \otimes
    \fp_D^{\,b}\pc_{XD}\pu_{XYD}\pc_{XD}^{\dagger}
    (\fp_D^{\,b})^\dagger.
\]
\begin{theorem}[Soundness {\cite[Corollary~4.3]{CM26}}]\label{thm:soundness}
For any quantum algorithm $\mathcal A$ making at most $q$
bidirectional queries, let $\rho_A^{\mathrm{perm}}$ be its final
state averaged over $\varphi\leftarrow S_N$, and let
$\rho_A^{\mathrm{comp}}$ be its final state when querying $\cp$,
after tracing out $D$. Here $A$ includes all algorithm registers.
Then
\[
    \frac12
    \left\|\rho_A^{\mathrm{perm}}
              -\rho_A^{\mathrm{comp}}\right\|_1
    \leq O\!\left(\frac{q}{\sqrt{N}}\right).
\]
Therefore, distinguishing the compressed permutation oracle from a random permutation oracle requires $\Omega(\sqrt{N})$ queries.
\end{theorem}

A useful, general fact about compression (for both permutations and functions) is that one can condition on a low-probability event not occurring under compression without sacrificing soundness. In particular, we have the following lemma.

\begin{lemma}\label{lem:diff-full-and-partial-decomp}
    Let $\algo H(\{\bot\} \cup [N])$ denote a database register, and $\algo H(\mathbf D)$ for some set $\mathbf D$ denote an auxiliary register of arbitrary dimension. Let $\mathcal{S}_D \subset [N]$ denote a family of sets indexed by $D \in \mathbf D$, $\ket{\mathcal{S}_D}$ the uniform superposition over elements in $\mathcal{S}_D$, and $\ket{+^n}$ the uniform superposition over $[N]$. Define operators \begin{align}
        \fc &= \operatorname{Exc}(\ket{\bot}, \ket{+^n}) \otimes I, & \widetilde \fc &= \sum_{D \in \mathbf D} \operatorname{Exc}(\ket{\bot}, \ket{\mathcal{S}_D}) \otimes \proj{D},
    \end{align}
    where the former is the standard function compression operator.
    Then we have \begin{align}
        \norm{\fc - \widetilde \fc} &= O\left(\max_{D \in \mathbf D}\sqrt{\frac{N-|\mathcal{S}_D|}{N}}\right).
    \end{align}
    \label{lem:sanitized-comp}
\end{lemma}

\begin{proof}
Both operators are block diagonal in $D$. For each block, expanding
the exchange operators and applying the triangle inequality bounds
their difference by $4\norm{\ket{+^n}-\ket{\mathcal S_D}}$.
Since $\norm{\ket{+^n}-\ket{\mathcal S_D}}^2
=2(1-\sqrt{|\mathcal S_D|/N})\leq 2(N-|\mathcal S_D|)/N$,
taking the maximum over $D$ proves the claim.
\end{proof}

\begin{corollary}\label{cor:part-decomp}
    For $0\leq t<N$, permutation and function compression satisfy
    \begin{align}
        \norm{(\pc_{XD}-\fc_{XD})\Pi^{inj}_D\Pi^t_D}
        \leq O\left(\sqrt{\frac{t}{N}}\right).
    \end{align}
    \label{cor:perm-function-comp}
\end{corollary}
\begin{proof}
    See \Cref{app:partial-dec}.
\end{proof}

\section{Quantum Security of Two-round Even-Mansour}

\paragraph{Real World.} Functions $P_1, P_2$ are sampled uniformly and independently at random from the set of all permutations on $[N]$ where $N=2^n$, i.e., $S_N$. Keys $k_1, k_2,k_3$ are sampled uniformly and independently at random from the set of $n$-bit strings. An algorithm $\algo A$, making at most $q$ queries, gets both forward and inverse oracle access to $P_1, P_2$ and \begin{align*}
    E^{P_1,P_2}_{k_1,k_2,k_3}(x) \coloneqq P_2(P_1(x \oplus k_1) \oplus k_2)\oplus k_3. 
\end{align*}

\paragraph{Ideal World.} It is identical to the real world, except that oracle access to $\pkac$ is replaced by oracle access to $P_{3}$, where $P_{3}$ is sampled uniformly at random (and independently of $P_{1}$ and $P_{2}$) from the set of permutations on $[N]$. 
\paragraph{Main Theorem.}
\begin{theorem} \label{thm:main-perm} Let $\mathcal{A}$ be a quantum algorithm making at most $q$ queries to its oracles in both forward and inverse directions. Then
\begin{align}
   \abs{\Pr[\algo A^{P_1, P_2, P_3}()=1]-\Pr[\algo A^{P_1, P_2, \pkac }()=1] } \leq O\left( \sqrt{q^3/N}  \right) + \varepsilon_{\textsf{CPO}}
\end{align} The result holds in the regime $q^2\leq N/2$.
\end{theorem}
\begin{proof}
We reduce the distinguishing advantage to the difference of hybrids as
    \begin{align}
        &\abs{\Pr[\algo A^{P_1, P_2, P_3}()=1]-\Pr[\algo A^{P_1, P_2, \pkac}()=1] } \nonumber\\
        &\quad \leq  \sum_{t=1}^q \Bigl\lVert  \Bigl( \cco^{R}V- V \Pi^\star   \cco^{I} \Bigr) \Pi^\star  \Pi^{t-1}\Pi^{k_1,k_3} \Bigr\rVert+ O(\sqrt{q^3/N}) + \varepsilon_{\textsf{CPO}}\\
        &\quad \leq  \sum_{t=1}^q O(\sqrt{t/N})+ O(\sqrt{q^3/N})  + \varepsilon_{\textsf{CPO}}\\
        &\quad \leq O(\sqrt{q^3/N}) + \varepsilon_{\textsf{CPO}}\,.
    \end{align}
    In the above, the first line follows by \Cref{thm:hybrid-diff}. The second follows by \Cref{thm:p1-bound}, \Cref{thm:p2-bound}, and \Cref{thm:p3-bound}, each of which handles queries to one of the three oracles.
\end{proof}

\subsection{Query operators.}
We introduce the query operators both for real and ideal worlds. Forward queries in both worlds are handled directly, while inverse queries are implemented with two helper operators, the flip and the swap.

\begin{definition}[Database inversion]
    The database inversion operator is \begin{align}
\dm\coloneqq\sw_{K^1K^3}\cdot\sw_{D^1D^2}\cdot\fp^{\otimes 3}_{D^1D^2D^3}. \label{eq:SSF} 
\end{align} We conjugate the forward query operators by this to answer the inverse queries.
\end{definition}
We implement all inverse queries by conjugation with $\dm$.
Since $\dm$ swaps $D^1$ and $D^2$, inverse queries to $P_1$
and $P_2$ use the forward operators for $P_2$ and $P_1$,
respectively. This uniform convention reduces inverse-query
bounds to forward-query bounds and a commutator bound for
$\dm$ and $V$.

We define the query operators for both worlds.

\begin{definition}[Ideal world queries]\label{def:ideal-queries}
Let the register $B$ hold the direction of the query, with $b=0$ for forward and $b=1$ for inverse, and $I$ hold the index of the queried permutation. The ideal world compressed permutation oracle call is
\begin{align*}
   \cco^{I}\ket{b}_B\ket{i}_I \coloneqq \ket{b}\ket{i}\otimes \begin{cases}
        \fc_{XD^i}\cdot \pu_{XYD^i}\cdot \fc_{XD^i}^\dagger &\text{ if } b=0,i\in[3]\\
        \dm\cdot  \fc_{XD^2}\cdot \pu_{XYD^2}\cdot  \fc_{XD^2}^\dagger \cdot \dm^\dagger &\text{ if } b=1,i=1\\
     \dm\cdot  \fc_{XD^1}\cdot \pu_{XYD^1}\cdot  \fc_{XD^1}^\dagger \cdot \dm^\dagger &\text{ if } b=1,i=2\\
    \dm\cdot\fc_{XD^i}\cdot \pu_{XYD^i}\cdot  \fc_{XD^i}^\dagger\cdot \dm^\dagger &\text{ if } b=1,i=3.
    \end{cases}
\end{align*}
\end{definition}

\begin{definition}[Real world queries]\label{def:real-queries}
Let the register $B$ hold the direction of the query, with $b=0$ for forward and $b=1$ for inverse, and $I$ hold the index of the queried permutation, with $i=3$ denoting the construction $\pkac$. The real world compressed permutation oracle call is 
    \begin{align*}
       \cco^{R}\ket{b}_B\ket{i}_I \coloneqq \ket{b}\ket{i} \otimes \begin{cases}
  \fc_{XD^i} \cdot \pu_{XYD^i} \cdot \fc_{XD^i}^\dagger
    & \text{if } b=0,\ i\in [3] \\
\cco\ppar{1}\cdot \cco\ppar{2} \cdot \cco\ppar{1}
    & \text{if } b=0,\ i=3 \\
  \dm\cdot  \fc_{XD^2}\cdot \pu_{XYD^2}\cdot  \fc_{XD^2}^\dagger \cdot \dm^\dagger &\text{if } b=1,i=1\\
     \dm\cdot  \fc_{XD^1}\cdot \pu_{XYD^1}\cdot  \fc_{XD^1}^\dagger \cdot \dm^\dagger &\text{if } b=1,i=2\\
  \dm \cdot  \cco\ppar{1}\cdot \cco\ppar{2} \cdot (\cco\ppar{1})^\dagger \cdot \dm^\dagger
    & \text{if } b=1,\ i=3
\end{cases}
    \end{align*}
 where $\cco\ppar{1}\coloneqq \fc\ppar{1} \cdot \pu\ppar{1} \cdot \fc\ppar{1}$ and $\cco\ppar{2}\coloneqq \fc\ppar{2} \cdot \pu\ppar{em} \cdot \fc\ppar{2},$  and \begin{align*}
& \fc\ppar{1}\ket{x,k_1}\coloneqq\ket{x,k_1} \otimes\fc_{D^1_{x\oplus k_1}}, \quad \pu\ppar{1}\ket{x,k_1,y_1,D^1}\coloneqq \ket{x,k_1,y_1\oplus D^1(x\oplus k_1),D^1} \\
& \fc\ppar{2}\ket{y_1,k_2}\coloneqq\ket{y_1,k_2} \otimes\fc_{D^2_{y_1\oplus k_2}}, \quad \pu\ppar{em}\ket{y_1,k_2,k_3,y,D^2}\coloneqq \ket{y_1,k_2,k_3,y \oplus D^2(y_1\oplus k_2)\oplus k_3,D^2} 
\end{align*} for computational basis input states. The value $y_1$ is held in an ancilla register $Y_1$ initialized to $\ket{0}$. Each construction query uncomputes it exactly, so we omit $Y_1$ from the oracle's workspace. When $Y_1$ holds $\bot$, we do not call $\cco\ppar{2}$; it acts as the identity instead.\footnote{This occurs when the decompression $\fc\ppar{1}$ returns $\bot$, which signals a failure of the compressed oracle simulation. Such failures are accounted for by the soundness bound of the compressed permutation technique in the final distinguishing advantage. Several conventions are possible for continuing the simulation after a failure; ours is one of them. Also, composing several compressed oracle calls adds up the soundness bounds, which is irrelevant in big-$O$ for a constant number of calls.}
\end{definition}
We omit the index and direction registers whenever clear from context. Also, since all $\dm,\fc,\fc\ppar{1},\fc\ppar{2},\pu,\pu\ppar{em}$ are self-adjoint, we drop the adjoint symbol in the rest. The factors defining $\dm$ commute pairwise.
\begin{remark}\label{rem:construction-query}
    Since $\fc\ppar{1}$ commutes with $\cco\ppar{2}$, the two inner copies cancel, i.e.,
    \begin{align*}
       \cco\ppar{1}\cdot \cco\ppar{2} \cdot \cco\ppar{1} = \fc\ppar{1}\cdot \pu\ppar{1}\cdot \cco\ppar{2} \cdot \pu\ppar{1} \cdot \fc\ppar{1}.
    \end{align*}
    If $D^1_{x\oplus k_1}$ is not $\bot$ after the decompression $\fc\ppar{1}$, the two copies of $\pu\ppar{1}$ can also be removed, by defining $\cco\ppar{2}$ controlled on $D^1_{x\oplus k_1}$ instead of on the ancilla $Y_1$, giving $\fc\ppar{1}\cdot \cco\ppar{2} \cdot \fc\ppar{1}$.
\end{remark}

\subsection{Setup}
Consider a $q$-query adversary with workspace $A$, query register $X$, and output register $Y$. Parameterize the adversary by unitaries $U^0, U^1, \dots, U^q$.

The oracle's state consists of registers $D^1, D^2, D^3, K_1, K_2, K_3$: each $D^i$ stores the queries to $P_i$, and $D^1, D^2$ together store the queries to $\pkac$, controlled on $K_1, K_2, K_3$. Let the initial state be\begin{align}
    \ket{\psi^0}\coloneqq \ket{0}_A \ket{0}_X \ket{0}_Y \ket{\bot}^{\otimes N}_{D^1}\ket{\bot}^{\otimes N}_{D^2}\ket{\bot}^{\otimes N}_{D^3}\ket{+^n}_{K_1} \ket{+^n}_{K_2}\ket{+^n}_{K_3}.
\end{align} 

In the real world, the adversary interacts with oracles for permutations $P_1,P_2,$ and $\pkac$. The oracle state includes the database $D^{3}$ for the permutation $P_3$ in the real world, but it is never accessed by the adversary and is present because our isometry from ideal to real world does not empty $D^3$ entirely.
Define \begin{align}
    R_{a:b}\coloneqq \prod_{j=b}^a (U^j\cco^R) \text{ and } R_{a:1}\coloneqq \prod_{j=1}^a (U^j\cco^R) U^0 \text{ for } a,b\in\set{2,\dots, q}.
\end{align}
Then the state after the real world experiment can be written as \begin{align} \label{real-exp-st}
     \ket{\psi^q_R} \coloneqq R_{q:1}\ket{\psi^0}.
\end{align} 

In the ideal world, the adversary interacts with oracles for random permutations $P_1,P_2,$ and $P_3$. In the ideal world the oracle's state also contains the key registers $K_1, K_2, K_3$. No ideal-world oracle acts on them, so the adversary's view is unaffected. They are included because we project out some keys before applying the isometry, so the keys must already be present in the ideal-world state. Given a projector $\Pi$, define\begin{align}
   I_{a:1}\coloneqq \prod_{j=1}^a (U^j\cco^I) U^0 \text{ and } I^\Pi_{a:1}\coloneqq \prod_{j=1}^a (\Pi U^j\cco^I) \Pi U^0\text{ for } a\in\set{1,\dots, q}\label{eq:def-I-pi}\,.
\end{align}
The state after the ideal experiment can then be written as 
\begin{align} \label{ideal-exp-st}
     \ket{\psi^q_I} \coloneqq I_{q:1}\ket{\psi^0}.
\end{align}

\begin{definition}[Size $t$ or less database]
\begin{align*} 
    \Pi^t \coloneqq \sum_{S\subseteq [N]: |S|\leq t} \Pi^S_D \text{ where } \Pi^S_D = \Bigl(\bigotimes_{i\in S^c}\proj{\bot}_{D_i}\Bigr)\Bigl(\bigotimes_{i\in S}(\Id -\proj{\bot})_{D_i}\Bigr)
\end{align*}    
\end{definition}
 
\begin{definition}[Injective database]
\begin{align*}
  \Pi^{inj} \coloneqq
    \sum_{D \in \mathbf{D}: \abs{\dom{D}}=\abs{\im{D}}}
    \proj{D}_{D}
\end{align*}
\end{definition}
We define a \textit{good} database notion such that $D^1$ and $D^2$ databases do not contain related entries by $k_2$, which will be useful later when reducing inverse queries to forward ones by ensuring $\vr$ does not introduce collisions, i.e., while bounding the commutator $[\dm,V]$. Moreover, $\Pi^g$ ensures $\vf$ either acts as the identity or changes the target registers from $\bot$ to non$-\bot$.
\begin{definition}[Good database]\label{def:pi-g} 
 \begin{align*}
    \kappa_{D^1,D^2} &= \{k_2\in[N] \, : \, \dom{D^2} \cap (\im{D^1} \oplus k_2) = \emptyset\} \\
    \ket{G_{D^1,D^2}} &= \frac{1}{\sqrt{|\kappa_{D^1,D^2}|}} \sum_{k_2 \in \kappa_{D^1,D^2}} \ket{k_2}_{K_2} \otimes \ket{D^1}_{D_1} \otimes \ket{D^2}_{D_2} \\
    \Pi^g &\coloneqq \sum_{D^1, D^2 \in \mathbf D} \proj{G_{D^1,D^2}}.
\end{align*}
The same notion applies to a restriction $D^1|_S$ with $S \subseteq [N]$: we say that $D^1|_S$ and $D^2$ with $k_2$ are good if \Cref{def:pi-g} holds with $D^1$ replaced by $D^1|_S$, and analogously for $D^2|_S$.

For $D^1, D^2 \in \mathbf{D}_t$, we have $|\kappa_{D^1,D^2}| \geq N - t^2$, since each pair of elements eliminates at most one value of $k_2$ from $[N]$.
\end{definition}
\begin{definition}[Uniform superposition keys]
    \begin{align*}
        \Pi^{k_1}=\proj{+^n}_{K_1}, \quad \Pi^{k_3}=\proj{+^n}_{K_3}
    \end{align*}We use $\Pi^{k_1,k_3}$ for their product when they are used together.
\end{definition}
We mostly work on both good and injective databases. Note that the two relevant projectors commute. We define
\begin{align}
  \Pi^\star \coloneqq \Pi^g \Pi^{inj}.  
\end{align}  

\begin{lemma}\label{lem:pi-t-violation}
    \begin{gather}
        \norm{(\Pi^{t+1})^\perp \fc \Pi^t}=0
    \end{gather}
    Consequently, $\lVert(\Pi^{t+1})^\perp \pu\fc \Pi^t\rVert=\lVert(\Pi^{t+1})^\perp \cco^I \Pi^t\rVert=0$.
\end{lemma}
\begin{proof}
    See \Cref{app:proj-bounds}.
\end{proof}

\begin{lemma}\label{lem:pi-star-violation}
	Let $t\geq 0$ be an integer with $t^2\leq N/2$. In
	\Cref{def:pi-g}, omit any summand for which
	$\kappa_{D^1,D^2}=\emptyset$.
	For each $Q\in\{\Pi^{inj},\Pi^g,\Pi^\star\}$ and each
	$i\in\{1,2,3\}$, we have
	\begin{align*}
		\norm{Q^\perp\fc_{XD^i}Q\Pi^t}
		&\leq O\bigl(\sqrt{t/N}\bigr),\\
		\norm{Q^\perp\pu_{XYD^i}\fc_{XD^i}Q\Pi^t}
		&\leq O\bigl(\sqrt{t/N}\bigr),\\
		\norm{Q^\perp\cco^I Q\Pi^t}
		&\leq O\bigl(\sqrt{t/N}\bigr).
	\end{align*}
	The last bound holds for forward and inverse ideal-world queries,
	including superpositions of oracle indices and directions.
	All norms are operator norms, and the implicit constants are absolute.

\end{lemma}

\subsubsection{Isometry from ideal to real world.}

We define an isometry mapping ideal-world states to real-world states, so as to reduce the difference between successive hybrid terms. The database $D^3$ may remain nonempty after the isometry; this is acceptable as long as it does not significantly affect that difference.

We first give an informal description of the isometry. In the first step, it applies the flip operator to the database $D^2$; when the databases under consideration are injective, the flip acts nontrivially, mapping $D^2$ to $(D^2)^{-1}$. Next, for a construction query recorded in the ideal world, i.e., a non-$\bot$ entry of $D^3$, say $D^3(x) = y$, the isometry fills the related entries of $D^1$ and $D^2$: the related $D^1$ entry is at input $x \oplus k_1$, and the related $D^2$ entry is at output $y \oplus k_3$. These are filled with a uniform superposition over all admissible intermediate values $u$ and $u\oplus k_2$, linking $D^1$ to $D^2$ via
$u \mapsto u \oplus k_2$ and chosen to preserve injectivity. As a
result, the ideal world construction query can be recovered from $D^1$
and $D^2$ exactly as in the real world. The step is concluded by flipping $D^2$ back. Finally, the isometry uncomputes the entries of $D^3$ whenever possible.

\begin{definition}[Isometry ideal-to-real] \label{def:isometry-k1-k2-k3} The isometry $V: \mathcal{H}_{D^1D^2D^3K_1K_2K_3} \rightarrow \mathcal{H}_{D^1D^2D^3K_1K_2K_3}$ is defined as a product of isometries 
\begin{equation}
    V\coloneqq\vr \fp_{D^2}\vf \fp_{D^2}
\end{equation}
by its action on computational basis states $\ket{D^1,D^2,D^3,k_1,k_2,k_3}$ as
\begin{align*}
    \vf &\coloneqq \prod_{x\in [N]} \vf\ppar{x}, \quad \vr \coloneqq  \prod_{x \in [N]} \vr\ppar{x}\text{ where }\\
    \vf\ppar{x}&\coloneqq \proj{\bot}_{D^3_{x\oplus k_1}}\otimes \Id + \sum_{y\in [N]} \proj{y}_{D^3_{x\oplus k_1}} \otimes \sum_{(D^1_{x})^c\in\mathbf{D}}\proj{(D^1_{x})^c}_{(D^1_{x})^c} \\
    &\quad \otimes \sum_{(D^2_{y\oplus k_3})^c\in\mathbf{D}}\proj{(D^2_{y\oplus k_3})^c}_{(D^2_{y\oplus k_3})^c}  \otimes T\ppar{\psi_{k_2}}_{D^1_{x}D^2_{y\oplus k_3}}, \text{ and }\\
     \vr\ppar{x} &= \proj{\bot}_{D^1_{x}}\otimes \Id+ \sum_{y' \in [N]} \proj{y'}_{D^1_x} \otimes \proj{\bot}_{D^2_{y'\oplus k_2}}\otimes \Id \\
    &\quad+\sum_{y' \in [N]} \proj{y'}_{D^1_x} \otimes \sum_{y \in [N]} \proj{y}_{D^2_{y' \oplus k_2}} \otimes T\ppar{y\oplus k_3}_{D^3_{x \oplus k_1}}.
\end{align*}
\begin{gather*}
   \text{Here, } T\ppar{\psi_{k_2}} \coloneqq \ket{\bot}\ket{\bot}\bra{\psi_{k_2}} + \ket{\psi_{k_2}}\bra{\bot}\bra{\bot} + (\Id - \ket{\bot}\ket{\bot}\bra{\bot}\bra{\bot} - \proj{\psi_{k_2}})\\
    \text{ with }\ket{\psi_{k_2}}\coloneqq \frac{1}{\sqrt{\abs{J}}} \sum_{u\in J} \ket{u}_{D^1_{x}}\ket{u\oplus k_2}_{D^2_{y\oplus k_3}} \text{ and }\\
    J \coloneqq \set{ u \in [N] :
  u \notin \im{(D^1_x)^c}, u \oplus k_2\not\in \im{(D^2_{y\oplus k_3})^c}},\\
    T\ppar{y} \coloneqq \ket{\bot}\bra{y} + \ket{y}\bra{\bot} + (\Id - \proj{\bot} - \proj{y}) \text{ for } y\in[N].
\end{gather*} 
We say $\vf\ppar{x}$ is active if $D^3_{x\oplus k_1}\neq\bot$. Clearly, $\vf= \prod_{x\in [N]:x \text{ is active}}\vf\ppar{x}$.

\end{definition}

Note that the action of $\vf$ is controlled on the entire database $D^1$ and $D^2$ other than the target registers so that it does not introduce any collisions to the databases $D^1$ and $D^2$.

For databases in standard (forward) orientation, we introduce two sets:
in \Cref{def:set-S}, the set of superscripts $x$ such that $\vf\ppar{x}$ acts nontrivially, changing $D^1_x$ from $\bot$ to non-$\bot$; and in \Cref{def:J}, the set of values that such an action may introduce into the image of $D^1$.
\begin{definition} \label{def:set-S}
For states with registers $K_1, K_3, D^1, D^2, D^3$ fixed to
$k_1, k_3, D^1, D^2, D^3$ in the computational basis, we define the set of superscripts at factors of $\vf$ may perform the fill $\ket{\bot}\ket{\bot} \mapsto \ket{\psi_{k_2}}$:
\begin{align*}
    S(K_1, K_3, D^1, D^2, D^3) \coloneqq \set{x\in[N]: D^3({x\oplus k_1})\neq\bot, D^1(x)=\bot, D^3(x\oplus k_1)\oplus k_3\not\in\im{D^2}}.
\end{align*}
If the values $D^3(x\oplus k_1)$, $x \in S$, are pairwise distinct, then $S$ is the exact set that $\vf$ performs the fill $\ket{\bot}\ket{\bot} \mapsto \ket{\psi_{k_2}}$. Otherwise, for $x\in S$, $\vf\ppar{x}$ might behave as identity depending on the order of the factors.
\end{definition} 

\begin{remark}[Fill or identity] \label{rem:fill-or-identity}
Consider a database state in the computational basis with standard state and register names. In each of the following cases, every factor of $\vf$ either performs the fill $\ket{\bot}\ket{\bot} \mapsto \ket{\psi_{k_2}}$ or acts as the identity. Consequently, the terms created are pairwise orthogonal.
\begin{enumerate}[label=(\arabic*)]
  \item $D^1, D^2$ with $k_2$ forms a good database.
  \item Restricted $D^1|_{A}, D^2$ with $k_2$ forms a good database, where $A \coloneqq \set{x \in [N] : \vf\ppar{x} \text{ is active}}$.\label{item:rem:fill-id-restricted-D1}
\end{enumerate}

\end{remark}

\begin{definition} \label{def:J}
For states with registers $K_2$, $D^1,D^2$ are fixed to $D^1,D^2,k_2$ and satisfy one of the cases in \Cref{rem:fill-or-identity}, the set of values that $\vf$ may introduce into the image of $D^1$ is
\begin{align*}
  J(K_2, D^1, D^2) \coloneqq
    \set{ u \in [N] \;:\; u \notin \im{D^1} \cup (\dom{D^2} \oplus k_2) } .
\end{align*}
\end{definition}

\begin{remark}[Commutativity of factors of $\vf$]\label{rem:comm-vf-factors}
    Note that each factor of $\vf\ppar{x}, x\in[N],$ has distinct targets on $D^1$, if, also, the images $D^3(x'\oplus k_1)$ s.t. $x' \in S$ are pairwise distinct, then $\vf\ppar{x}$'s targets on $D^2$ are also distinct. Moreover, if one of the conditions of \Cref{rem:fill-or-identity} holds, all factors of $\vf$ share the same set $J$ and therefore commute. This holds trivially when only one factor of $\vf$, say $\vf\ppar{x}$, is active.
\end{remark}

\begin{remark}\label{rem:S-and-J-wrt-invdb}
 Consider injective databases $\dbswithc$. Note that for databases in inverse orientation (i.e. $\dm$ applied on the standard ones), the sets in these definitions carry the same information relative to $\dbswithc,k_1,k_2,k_3$ (as states) but the definitions do not give the same sets relative to the registers. We can express the relations between them. Let \begin{align}
\ket{\psi} &= \ket{D^1}_{D^1} \ket{D^2}_{D^2} \ket{D^3}_{D^3}
             \ket{k_1}_{K_1} \ket{k_2}_{K_2} \ket{k_3}_{K_3},\\
\ket{\psi'} &= \dm\ket{\psi}
            = \ket{(D^2)^{-1}}_{D^1} \ket{(D^1)^{-1}}_{D^2} \ket{(D^3)^{-1}}_{D^3}
              \ket{k_3}_{K_1} \ket{k_2}_{K_2} \ket{k_1}_{K_3}.
 \end{align} If $S$ is the set defined in \Cref{def:set-S} where the registers are fixed as in $\ket{\psi}$ and $S'$ is the set for $\ket{\psi'}$, then \begin{align}
     x\in S \iff y_x\oplus k_3 \in S' \text{ where } y_x=D^3(x\oplus k_1).
 \end{align}Similarly, if $J$ is the set defined in \Cref{def:J} where the registers are fixed as in $\ket{\psi}$ and $J'$ is the set for $\ket{\psi'}$, then \begin{align}
     u\in J \iff u\oplus k_2 \in J'.
 \end{align} Note that due to the database swap, the elements introduced to the image of $D^1$ for both $\ket{\psi}$ and $\ket{\psi'}$ remains to be the same. 

\end{remark}

The set $J(K_2, D^1,D^2)$ is defined independently of $D^3$. When $ S(K_1, D^1, D^2,D^3)\neq\emptyset$ (hence dependence on $D^3$), they can alternatively be defined as the set $J$ corresponding to the first non-trivially acting $\vf\ppar{x}$, by changing $\bot$ to a non-$\bot$ state in ${D^1_x}$, in \Cref{def:isometry-k1-k2-k3} — that is the smallest element of $S$. In this nontrivial action, since the target registers are in $\bot$ and the control values are determined directly by the initial databases $D^1, D^2$ and the key $k_2$, we can define $J$ with respect to the whole database rather than the complement of the target.

\subsection{Hybrids difference}

We write the hybrid difference where ideal queries are interleaved by $\Pi^g$ and $ \Pi^{inj}$ to avoid loosening the bound, noting again $\Pi^\star=\Pi^g\Pi^{inj}$. We write $I^{\star }_{q:1}$ instead of $I^{\Pi^\star }_{q:1}$ which is defined by \eqref{eq:def-I-pi} as \begin{align}
     I^{\star }_{q:1}=\prod_{j=1}^a (\Pi^\star U^j\cco^I) \Pi^\star U^0\text{ for } a\in\set{1,\dots, q}.
\end{align}
\begin{lemma}\label{lem:q-star-interleaved-vs-not}
For every $q$ with $q^2 \leq N/2$,
\begin{align}
    \norm{(\Pi^\star)^\perp I_{q:1}\ket{\psi^0}} \leq O\big(\sqrt{q^3/N}\big)
    \quad\text{and}\quad
    \norm{(I^{\star}_{q:1} - \Pi^\star I_{q:1})\ket{\psi^0}} \leq O\big(\sqrt{q^3/N}\big).
\end{align}
\end{lemma}
\begin{proof}
See \Cref{app:hybdrids-diff}.
    
\end{proof}
\begin{theorem}\label{thm:hybrid-diff}
Let $\algo A$ be a quantum algorithm making at most $q$ queries, with $q^2 \leq N/2$. Then
\begin{align}
    &\abs{\Pr[\algo A^{P_1, P_2, P_3}()=1]-\Pr[\algo A^{P_1, P_2, \pkac}()=1]} \nonumber\\
    &\quad \leq \sum_{t=1}^q \Bigl\lVert \Bigl( \cco^{R}V - V \Pi^\star \cco^{I} \Bigr) \Pi^\star \Pi^{t-1}\Pi^{k_1,k_3} \Bigr\rVert + O\big(\sqrt{q^3/N}\big) + \varepsilon_{\textsf{CPO}}. \label{eq:lem:hybrids-with-keys}
\end{align}
We call the $t$-th summand the \emph{hybrid difference at the $t$-th query}.
\end{theorem}

\begin{proof}
    See \Cref{app:hybdrids-diff}.
\end{proof}

\subsection{Reducing inverse queries to forward ones}

As a general proof strategy for bounding the real and ideal world hybrids, we reduce inverse queries to forward ones, accounting for additional terms arising during the reduction, namely the commutator $[\dm, V]$. 

When bounding this commutator, the isometry $V$ (defined in \Cref{def:isometry-k1-k2-k3}) behaves well in the sense that it acts symmetrically for forward and inverse queries. However, the critical issue is that $\fp$ in $\dm$ behaves differently depending on whether the database is injective or not. The factor $\fp_{D^2}\vf\fp_{D^2}$ of $V$ preserves injectivity by definition but $\vr$ may not. For this reason, we use the good database projector defined in \Cref{def:pi-g} so that $D^1$ and $D^2$ databases do not contain related entries by $k_2$. On good databases, $\vr$ does not introduce any new element on $D^3$, hence preserves injectivity, and helps to bound the commutator. 

We look into the commutation relation of $\dm$ and the database projectors.
\begin{lemma}\label{lem:comm-dm-and-proj}
    \begin{align}
    [\dm, \Pi^{t}]=[\dm, \Pi^{inj}]=[\dm, \Pi^{g}]\Pi^{inj}=[\dm, \Pi^{k_1,k_3}]=0.
    \end{align}
    These projectors commute pairwise. Moreover, $M$ commutes with their products on injective databases.
\end{lemma}
\begin{proof}
See \Cref{app:reduce-inv-to-fwd}.
\end{proof}

We look into the commutation relations of $\dm$ and $V$.

\begin{lemma} \label{lem:v-dm-comm}
$\norm{[\dm, V] \Pi^g \Pi^{inj} \Pi^t}=0$.
\end{lemma}
\begin{proof}
    See \Cref{app:reduce-inv-to-fwd}.
\end{proof}

We are ready to reduce the hybrid difference for inverse queries to that for forward ones. For now, fix the index register $I$ to $i=3$. Let $\cco^{R,fwd}$ and  $\cco^{I,fwd}$ be $\cco^{R}$ and $\cco^{I}$ restricted to $b=0, i=3$, respectively. Similarly, let $\cco^{R,inv}$ and  $\cco^{I,inv}$ be $\cco^{R}$ and $\cco^{I}$ restricted to $b=1, i=3$, respectively. Then one sees that \begin{align}
    \cco^{R,inv}= \dm  \cdot\cco^{R,fwd} \cdot\dm^\dagger \text{ and } \cco^{I,inv}= \dm\cdot \cco^{I,fwd} \cdot\dm^\dagger.\label{eq:inv-to-fwd}
\end{align} 

\begin{theorem}\label{thm:inv-to-fwd-hybrid} Using the notation of \eqref{eq:inv-to-fwd} for the inverse query operator, for every integer $t \geq 0$, the hybrid difference for an inverse query ($b=1$) with $i=3$ equals that for the forward one:
    \begin{align}
      & \lVert (\dm  \cco^{R,fwd} \dm V - V \Pi^{g}  \Pi^{inj}  \dm \cco^{I,fwd}  \dm) \Pi^{g}  \Pi^{inj}  \Pi^t  \Pi^{k_1,k_3}\rVert\\
  &\quad= \lVert  ( \cco^{R,fwd} V   -  V \Pi^{g}  \Pi^{inj}  \cco^{I,fwd}  )  \Pi^{g}  \Pi^{inj}  \Pi^t \Pi^{k_1,k_3}\rVert
    \end{align}The same holds for an inverse query with $i \in \{1,2\}$ except that an inverse query with $i=1$ is reduced to forward one with $i=2$, and vice versa by definition.
\end{theorem}
\begin{proof}
  See \Cref{app:reduce-inv-to-fwd}.  
\end{proof}

 \subsection{Bounding the forward query hybrids}

We define two projectors controlled on the query register, both ensuring that the real and ideal world oracles return the same value during the $(t+1)$-st query, so that the hybrid difference is small.

For $P_1$ and $P_2$ queries the two oracles are the same, and $\Gamma_1$ ensures that the isometry leaves the database register holding the current query untouched, and does not link it to the other registers; consequently the behaviour of the isometry elsewhere does not closely depend on what that register holds. Keeping the queries database register the same in both worlds is not sufficient though since isometry might alter the other registers differently in both terms.

For $P_3/\pkac$-queries the oracles differ: the ideal world answers from $D^3$, the real world from $D^1$ and $D^2$. Here the isometry must act nontrivially For query on $x$, this requires both target registers of its factor $\vf\ppar{x\oplus k_1}$ to be in $\bot$; $\Gamma_1$ and $\Gamma_2$ enforce this for the
relevant $D^1$ and $D^2$ registers, respectively. Once $\vf\ppar{x\oplus k_1}$ is active, then $\vr\ppar{x\oplus k_1}$ automatically does the rest of the work of uncomputing $D^3_x$.

\begin{definition}[Good query]\label{def:gamma-p}
For a point $z \in [N]$, a key register $K$ and a database register $D$, let
\[
  \Pmiss{z}{K}{D} \coloneqq
  \sum_{k\in[N]} \proj{k}_{K} \otimes
  \sum_{D:\, z\notin\im{D}\oplus k} \proj{D}_{D}
\]
be the projector onto states in which no entry of $D$, shifted by the key in $K$, equals $z$. Define
\begin{align*}
  \Gamma_1 \coloneqq \sum_{x\in[N]} \proj{x}_X \otimes
  \begin{cases}
    \sum_{k_1\in[N]} \proj{k_1}_{K_1} \otimes \proj{\bot}_{D^3_{x\oplus k_1}}
      & \text{for a $P_1$-query,} 
      \\
    \Pmiss{x}{K_2}{D^1} \otimes \Bigl(
      \proj{\bot}_{D^2_x} \otimes \Id_{K_3 D^3}
      + \sum_{y\in[N]} \proj{y}_{D^2_x} \otimes \Pmiss{y}{K_3}{D^3}
    \Bigr)
      & \text{for a $P_2$-query,} 
      \\
    \sum_{k_1\in[N]} \proj{k_1}_{K_1} \otimes \proj{\bot}_{D^1_{x\oplus k_1}}
      & \text{for a $\pkac/P_3$-query},
  \end{cases}
\end{align*}
    \begin{align*}
       \Gamma_2 &\coloneqq \sum_{x\in[N]}\proj{x}_X\otimes \proj{\bot}_{D^3_x}\otimes \Id+ \sum_{x,y,k_3\in [N]} \proj{x}_X\otimes \proj{y}_{D^3_x} \otimes\proj{k_3}_{K_3}\otimes (\fp_{D^2}\proj{\bot}_{D^2_{y\oplus k_3}}\fp_{D^2})
    \end{align*}for a $\pkac/P_3$-query,
\end{definition}

In words, on a $P_1$-query $\Gamma_1$ ensures that $\vf\ppar{x}$ is not active on the given state, whereas on a $P_3$-query it ensures that the database $D^1$ does not prevent $\vf\ppar{x \oplus k_1}$ from being active.

On $P_2$-query, $\Gamma_1$ behaviour is more involved because the database $D^2$ may be pointed to by both $D^1$ and $D^3$ during the application of $\fp_{D^2}\vf\fp_{D^2}$, hence $\Gamma_1$ ensures two things. First, it ensures $x \notin \im{D^1} \oplus k_2$, hence no register of $D^1$ points to $x$. This is useful by enabling us to argue that a database is good independent of what $D^2_x$ holds. Second, if the queried database register $D^2_x$ is not $\bot$, then it ensures that the value does not lie in $\im{D^3} \oplus k_3$. The second condition is needed later, when we compare the images under $\fp_{D^2}\vf\fp_{D^2}$ of two terms that agree everywhere except at $D^2_x$; if $D^2(x)=D^3(i\oplus k_1)\oplus k_3 \in \im{D^3} \oplus k_3$ and $D^1_i=\bot$ initially, the image of $D^1(i)$ under $\fp_{D^2}\vf\fp_{D^2}$ is $\bot$ in every term. However, if $D^2(x)\not\in \im{D^3} \oplus k_3$ and $D^1_i=\bot$, then the image of $D^1(i)$ under $\fp_{D^2}\vf\fp_{D^2}$ is not necessarily in $\bot$.

\begin{lemma}\label{lem:gamma1-perp-bound}
For every query with $b=0$ and $i \in [3]$, and every integer $t \geq 1$ with $t^2 \leq N/2$,
\begin{align}
    \lVert \Gamma_1^\perp \Pi^\star \Pi^{t} \Pi^{k_1,k_3} \rVert \leq O\big(\sqrt{t/N}\big).
\end{align}
\end{lemma}

\begin{proof}
    See \Cref{app:proj-bounds}.
\end{proof}

\begin{lemma}\label{lem:gamma2-perp-bound}
States holding $\ket{+^n}$ at register $K_3$ and lying in the image of
$\Pi^t$ have only a small component violating $\Gamma_2$, more precisely,
\begin{align}
    \norm{\Gamma_2^\perp\,\Pi^t\,\Pi^{k_3}} \leq O(\sqrt{t/N}).
\end{align}
\end{lemma}
\begin{proof}
    See \Cref{app:proj-bounds}.
\end{proof}

\subsubsection{$P_1$-query}\label{sec:P_1-query}
Using \Cref{thm:inv-to-fwd-hybrid}, it is enough to consider the difference between the hybrids for forward queries, i.e., consider the real and ideal world oracle queries for $b=0$ and $i=1$.

\begin{lemma}\label{lem:p1-psi-y-d}
    Let $\ket{\psi}$ be a database state in the computational basis
\begin{align}
  \ket{\psi} \coloneqq \ket{y}_{D^1_x}\ket{D}_{D}\ket{\bot}_{D^3_{x\oplus k_1}}
\ket{(D^3_{x\oplus k_1})^c}_{(D^3_{x\oplus k_1})^c}\ket{k_1}_{K_1}\ket{k_2}_{K_2}\ket{k_3}_{K_3} \label{eq:lem-p1-psi}
\end{align}
where $D = (D^1_x)^c D^2, y \in [N]\cup\set{\bot}$, $D$ with $k_2$ is a good database, and each database $D^i$ is of size at most $t$ and is injective.  Since $\fp_{D^2}\vf\fp_{D^2}$ acts as the identity
on the remaining registers, we may define $\ket{\psi^y_D}$ by
\begin{align}
  \fp_{D^2}\vf\fp_{D^2}\ket{\psi} = \ket{y}_{D^1_x}\ket{\psi^y_D}_{D}
   \ket{\bot}_{D^3_{x\oplus k_1}}
\ket{(D^3_{x\oplus k_1})^c}_{(D^3_{x\oplus k_1})^c}
\ket{k_1}_{K_1}\ket{k_2}_{K_2}\ket{k_3}_{K_3}. \label{eq:p1-psi-y-d}
\end{align}Then \begin{enumerate}[label=(\arabic*)]
        \item For databases such that $\ket{D'}\perp\ket{D}$, we have $\ket{\psi^y_D}\perp \ket{\psi^y_{D'}}$, as a consequence, we get $( \ket{\psi^y_D}- \ket{\psi^{y'}_D})\perp (\ket{\psi^y_{D'}}- \ket{\psi^{y'}_{D'}}),$ 
        \item We get $\norm{\ket{\psi^y_D} - \ket{\psi^{y'}_D}} \leq  O(\sqrt{t/N})$
    \end{enumerate} for all $y,y'\in[N]\cup\set{\bot}$ and $D,D'\in\mathbf{D}_t$ for which the relative terms are defined. For the bounds, we use the assumption $3q < N/2$ which is trivially satisfied by the regime for $q$.
\end{lemma}

\begin{proof}
    See \Cref{app:p1}.
\end{proof}

\begin{theorem}[$P_1$-query]\label{thm:p1-bound}
Let the registers $B$ and $I$ hold $b=0$ and $i=1$, i.e., a forward query to $P_1$ in both worlds. Then for every integer $t \geq 1$ with $t^2 \leq N/2$,
\begin{align}
    \lVert (\cco^R V - V \Pi^\star \cco^I) \Pi^\star \Pi^t \Pi^{k_1,k_3} \rVert \leq O\big(\sqrt{t/N}\big).
\end{align}
\end{theorem}
\begin{proof}
Note that oracle query to $P_1$ is handled identically in the real and ideal worlds, i.e., $\cco^{R} = \cco^{I}$. Accordingly, we simply write $\cco$ for this operator. We use the uniformity of $k_1$ in \begin{align}
     &\lVert  ( \cco^R V   -  V \Pi^\star  \cco^I  )  \Pi^\star  \Pi^t \Pi^{k_1,k_3} \rVert   \\
     &\quad \leq  \lVert  ( \cco^R V   -  V \Pi^\star  \cco^I  ) \Gamma_1 \Pi^\star  \Pi^t \Pi^{k_1,k_3} \rVert + O(\sqrt{t/N}) \quad \text{(By \Cref{lem:gamma1-perp-bound})}\\
     &\quad \leq  \lVert  ( \cco V   -  V \Pi^\star  \cco  )  \Gamma_1 \Pi^\star  \Pi^t \rVert + O(\sqrt{t/N})\quad \text{(By submultiplicativity)}\\
     &\quad \leq  \lVert  ( \cco V   -  V  \cco  )  \Gamma_1 \Pi^\star  \Pi^t \rVert  + O(\sqrt{t/N}) \label{eq:some-term-p1}
\end{align} where in the last line we use the identity $\Pi^\star=\Id - (\Pi^\star)^\perp$ with triangle inequality and \Cref{lem:pi-star-violation}. For the remainder of the proof, it suffices to bound the first term of \Cref{eq:some-term-p1} on states in the image of $\Gamma_1 \Pi^\star  \Pi^t$. 

We continue by decomposing $\cco$ as \begin{align}
    &\lVert  ( \cco V   -  V  \cco  )  \Gamma_1 \Pi^\star  \Pi^t \rVert \\
    &\quad \leq \norm{[V,\fc]\Gamma_1 \Pi^\star  \Pi^t} + \norm{[V,\pu] \fc\Gamma_1  \Pi^\star  \Pi^t} + \norm{[V,\fc] \pu \fc \Gamma_1  \Pi^\star  \Pi^t}\\
    &\quad \leq \norm{[V,\fc]\Gamma_1 \Pi^\star  \Pi^t}  + 0 + \norm{[V,\fc] \pu \fc \Gamma_1  \Pi^\star  \Pi^t} \label{eq:p1-2sum-v-fc}
\end{align} because, for $x, k_1$ in the registers $X, K_1$, respectively, $\pu$ is
controlled on $D^1_x$, and $D^1_x$ is preserved by $V$. The latter happens due to $\Gamma_1$ since it ensures that $\vf\ppar{x}$ is not active as $D^3_{x\oplus k_1}=\bot$, and $\vr\ppar{x}$ is controlled on $D^1_x$ anyway.

Observe that \begin{align}
          \norm{[V,\fc] \pu \fc \Gamma_1  \Pi^\star \Pi^t}
          &= \norm{[V,\fc]\Gamma_1 \pu  \fc   \Pi^\star\Pi^t}\\
          &\leq  \norm{[V,\fc]\Gamma_1 \Pi^\star  \Pi^{t+1}\pu  \fc   \Pi^\star\Pi^t}  + O(\sqrt{t/N}) \\
          &\leq  \norm{[V,\fc]\Gamma_1 \Pi^\star  \Pi^{t+1}} + O(\sqrt{t/N}) 
      \end{align} where the second line follows by \Cref{lem:pi-t-violation} and \Cref{lem:pi-star-violation}, and last line by the submultiplicativity of the norm. Hence, we upper bound \Cref{eq:p1-2sum-v-fc} as 
      \begin{align}
          &\norm{[V,\fc]\Gamma_1 \Pi^\star  \Pi^t}  + \norm{[V,\fc] \pu \fc \Gamma_1  \Pi^\star  \Pi^t}  \\
          &\leq \norm{[V,\fc]\Gamma_1 \Pi^\star  \Pi^{t}} + \norm{[V,\fc]\Gamma_1 \Pi^\star  \Pi^{t+1}} + O(\sqrt{t/N}) 
      \end{align} by \Cref{lem:pi-t-violation},\ref{lem:pi-star-violation}, and since the two terms are indifferent to big-$O$ in the final bound, we bound only the first. We do this by decomposing $V$ as\begin{align}
          \norm{[V,\fc]\Gamma_1 \Pi^\star  \Pi^t}\leq \norm{[\fp_{D^2}\vf\fp_{D^2},\fc]\Gamma_1 \Pi^\star  \Pi^t} + \norm{[\vr,\fc]\fp_{D^2}\vf\fp_{D^2}\Gamma_1 \Pi^\star  \Pi^t}.  \label{eq:comm-v-fc}
      \end{align}
      In the rest, we fix the registers $X,K_1$ in the computational basis $x,k_1$, respectively, such that they are fixed by $\Gamma_1$ , i.e., $D^3_{x\oplus k_1}=\bot$.
      \begin{enumerate}[leftmargin=0pt, labelindent=0pt, itemindent=*, label=(\arabic*)]
          \item Consider $D^1_x=\bot$.  \begin{enumerate}
        \item First, we bound $\norm{[\fp_{D^2}\vf\fp_{D^2},\fc]\Gamma_1 \Pi^\star  \Pi^t}$. Observe that $\fc$ acts only on $D^1_x$, and $\fp_{D^2}\vf\fp_{D^2}$ preserves $D^1_x$ in the computational basis and $\bot$, while its action on $(D^1_x)^c$ and $D^2$ depends on what $D^1_x$ register holds in order to preserve injectivity of the databases. Since orthogonality of all the registers other than $(D^1_x)^c$ and $D^2$ is preserved by this commutator, we fix all other registers in the computational basis and write a general state\begin{align}
\ket{\psi}\coloneqq \sum_{D\in\mathbf{D}_\mathrm{adm}}\gamma_D\ket{\bot}_{D^1_x}\ket{D}_{D}\ket{\bot}_{D^3_x}
\ket{(D^3_x)^c}_{(D^3_x)^c}\ket{k_1}_{K_1}\ket{k_2}_{K_2}\ket{k_3}_{K_3} 
              \end{align} where $(D^3_x)^c\in\mathbf{D}_{t}$ and\begin{align}
\mathbf{D}_\mathrm{adm}\coloneqq\set{((D^1_x)^c, D^2)\in \mathbf{D}_t \times \mathbf{D}_t: (D^1_x)^c, D^2 \text{ with } k_2 \text{ is good}}.
              \end{align} 
Then we write the terms of the commutator, while omitting the registers that are untouched, \begin{align}
    \fc \fp_{D^2}\vf\fp_{D^2} \ket{\psi} &= \fc \sum_{D\in\mathbf{D}_\mathrm{adm}}\gamma_D\ket{\bot}_{D^1_x}\ket{\psi^\bot_D}_{D}\\
    &=  \frac{1}{\sqrt{N}}\sum_{y\in[N],D\in\mathbf{D}_\mathrm{adm}}\gamma_D\ket{y}_{D^1_x}\ket{\psi^\bot_D}_{D} 
\end{align} and \begin{align}
 \fp_{D^2}\vf\fp_{D^2} \fc\ket{\psi} &=  \frac{1}{\sqrt{N}}\sum_{y\in[N],D\in\mathbf{D}_\mathrm{adm}}\gamma_D\ket{y}_{D^1_x} \ket{D}_D\\
 &=\frac{1}{\sqrt{N}}\sum_{y\in[N],D\in\mathbf{D}_\mathrm{adm}}\gamma_D\ket{y}_{D^1_x} \ket{\psi^y_D}_{D} 
\end{align} where $\ket{\psi^\bot_D}_{D}$ and $\ket{\psi^y_D}_{D}$ are as defined in \Cref{lem:p1-psi-y-d}. So, the difference is \begin{align}
&\norm{[\fc,\fp_{D^2}\vf\fp_{D^2}]\ket{\psi}}\\
&=\Big\lVert\frac{1}{\sqrt{N}}\sum_{y\in[N],D\in\mathbf{D}_\mathrm{adm}}\gamma_D\ket{y}_{D^1_x}\otimes\Big( \ket{\psi^\bot_D}-\ket{\psi^y_D}\Big)_D\Big\rVert\\
    &=\sqrt{ \frac{1}{N} \sum_{y\in[N],D\in\mathbf{D}_\mathrm{adm}}\abs{\gamma_D}^2 \lVert \ket{y}_{D^1_x}\otimes( \ket{\psi^\bot_D}-\ket{\psi^y_D})_D \rVert^2 }\\
    &\leq \sqrt{ \frac{1}{N}\sum_{y\in[N],D\in\mathbf{D}_\mathrm{adm}} \abs{\gamma_D}^2 O\Big(\frac{t}{N}\Big) }\\
    & \leq O(\sqrt{t/N})
\end{align} where in the second line we used the orthogonality of $\ket{y}$'s and $\ket{\psi^\bot_D}-\ket{\psi^y_D}$'s as given in \Cref{lem:p1-psi-y-d}. In the third line, we used the bound in \Cref{lem:p1-psi-y-d}.
              \item Next, we bound $\norm{[\vr,\fc]\fp_{D^2}\vf\fp_{D^2}\Gamma_1 \Pi^\star  \Pi^t}$. Observe that $\fp_{D^2}\vf\fp_{D^2}$ preserves $D^1_x=\bot$ since $\vf\ppar{x}$ is not active because of $\Gamma_1$, and at most doubles the size of $D^2$. Moreover, $\vr\ppar{x}$ is the only component of $\vr$ that acts on $D^1_x$ non-trivially, and the action is in a controlled way. The action of $\vr\ppar{x}$ on the state becomes trivial when the value $D^1(x)\oplus k_2 \not \in \dom{D^2}$. So, when we project out the images of $x$ such that $D^1(x)\oplus k_2 \in \dom{D^2}$, the operators commute. 
              
              The commutator of the operators preserves orthogonality of $(D^1_x)^c$, $D^2$ and $(D^3_{x\oplus k_1})^c$. So, we fix these in the computational basis to databases of size at most $t$, and $D^3_{x\oplus k_1}=\bot$ (due to $\Gamma_1$), and $D^2$ is injective, as well as, fix the remaining key registers, i.e., write a general state \begin{align*}
                  \ket{\psi}\coloneqq \ket{\bot}_{D^1_x} \ket{(D^1_x)^c}_{(D^1_x)^c} \ket{D^2} \ket{\bot}_{D^3_{x\oplus k_1}}\ket{(D^3_{x\oplus k_1})^c}_{(D^3_{x\oplus k_1})^c}\ket{k_1}_{K_1}\ket{k_2}_{K_2}\ket{k_3}_{K_3}.
              \end{align*} Then write the image of $\fp_{D^2}\vf\fp_{D^2}$ on $\ket{\psi}$ as \begin{align}
\fp_{D^2}\vf\fp_{D^2}\ket{\psi}=\sum_{\substack{D\in \mathbf{D}_{2t, inj}}}\gamma_D \ket{\bot}_{D^1_x}\ket{D}_{D^2}\ket{\psi_D}
              \end{align} by omitting the irrelevant registers for the commutator, then
              \begin{align}
&\norm{[\vr,\fc]\fp_{D^2}\vf\fp_{D^2}\Gamma_1 \Pi^\star  \Pi^t} \\
&=\norm{[\vr,\fc]\fp_{D^2}\vf\fp_{D^2}\ket{\psi}}\\
            &\leq \norm{\sum_{\substack{D\in \mathbf{D}_{2t, inj}}}\gamma_D \frac{1}{\sqrt{N}}\sum_{y\in[N]:y\oplus k_2\in \dom{D}}\ket{y}_{D^1_x}\ket{D}_{D^2}} \\
                 &\leq \sqrt{\sum_{D\in \mathbf{D}_{2t,inj}}\abs{\gamma_D}^2\sum_{\substack{y\in[N]:y\oplus k_2\in \dom{D}}}\frac{1}{N}}\\
                 &\leq O(\sqrt{t/N}).
              \end{align}
          \end{enumerate}
          \item Consider $D^1_x\neq\bot$. 
          \begin{enumerate}
              \item \label{it:d1x-non-bot-vf-fc}
              We bound $\norm{[\fp_{D^2}\vf\fp_{D^2},\fc]\Gamma_1 \Pi^\star  \Pi^t}$. Similar to the previous case, observe that $\fc$ acts only on $D^1_x$, and $\fp_{D^2}\vf\fp_{D^2}$ preserves $D^1_x$ in the computational basis, while its action on $(D^1_x)^c$ and $D^2$ depends on what $D^1_x$ register holds in order to preserve injectivity of the databases. Since orthogonality of all the registers other than $(D^1_x)^c$ and $D^2$ is preserved by this commutator, we fix all other registers in the computational basis and write a general state\begin{align*}
\ket{\psi}\coloneqq \sum_{y\in[N],D\in\mathbf{D}_\mathrm{adm}}\gamma_{y,D}\ket{y}_{D^1_x}\ket{D}_{D}\ket{\bot}_{D^3_x}
\ket{(D^3_x)^c}_{(D^3_x)^c}\ket{k_1}_{K_1}\ket{k_2}_{K_2}\ket{k_3}_{K_3} 
              \end{align*} where $(D^3_x)^c\in\mathbf{D}_{t}$ and\begin{align}
    \mathbf{D}_\mathrm{adm}\coloneqq\set{((D^1_x)^c, D^2)\in \mathbf{D}_{t-1} \times \mathbf{D}_t: (D^1_x)^c, D^2 \text{ with } k_2 \text{ is good}}.
              \end{align} 
Then we write the terms of the commutator, while omitting the registers that are untouched, \begin{align}
 &   \fc \fp_{D^2}\vf\fp_{D^2} \ket{\psi} = \fc\Biggl( \sum_{y\in [N], D\in \mathbf{D}_\mathrm{adm}} \gamma_{y,D}\ket{y}_{D^1_x} \ket{\psi^y_D}_D\Biggr)\\
&= \sum_{y\in [N], D\in \mathbf{D}_\mathrm{adm} } \gamma_{y,D}\Biggl(\ket{y}-\frac{\ket{+^n}}{\sqrt{N}}+\frac{\ket{\bot}}{\sqrt{N}}\Biggr)_{D^1_x} \ket{\psi^y_D}_D
\end{align} and \begin{align}
 &\fp_{D^2}\vf\fp_{D^2}\fc\ket{\psi} = \fp_{D^2}\vf\fp_{D^2} \sum_{y\in [N], D\in \mathbf{D}_\mathrm{adm}} \gamma_{y,D}\Biggl(\ket{y}-\frac{\ket{+^n}}{\sqrt{N}}+\frac{\ket{\bot}}{\sqrt{N}}\Biggr)_{D^1_x}\ket{D}_D\\
 &= \sum_{y\in [N], D\in \mathbf{D}_\mathrm{adm}} \gamma_{y,D}\Biggl(\ket{y}_{D^1_x}\ket{\psi^y_D}_D -\sum_{y'\in[N]}\frac{\ket{y'}_{D^1_x}}{N}\ket{\psi^{y'}_D}_D+ \frac{\ket{\bot}_{D^1_x}}{\sqrt{N}}\ket{\psi^{\bot}_D}_D\Biggr)  
\end{align} where $\ket{\psi^\bot_D}_{D}$ and $\ket{\psi^y_D}_{D}$ are as defined in \Cref{lem:p1-psi-y-d}. So, the difference is \begin{align}
&\norm{[\fc,\fp_{D^2}\vf\fp_{D^2}]\ket{\psi}}\\
&\leq\norm{\sum_{y\in[N],D\in\mathbf{D}_\mathrm{adm}}\gamma_{y,D}\frac{1}{N}\sum_{y'\in[N]}\ket{y'}_{D^1_x}\otimes\Big( \ket{\psi^y_D}-\ket{\psi^{y'}_D}\Big)_D}\nonumber\\
&\quad + \norm{\sum_{y\in[N],D\in\mathbf{D}_\mathrm{adm}}\gamma_{y,D}\frac{1}{\sqrt{N}}\ket{\bot}_{D^1_x}\otimes\Big( \ket{\psi^y_D}-\ket{\psi^\bot_D}\Big)_D}\\
&\quad \text{(By triangle inequality)}\nonumber\\
&\leq\sum_{y\in[N]}\norm{\sum_{D\in\mathbf{D}_\mathrm{adm}}\gamma_{y,D}\frac{1}{N}\sum_{y'\in[N]}\ket{y'}_{D^1_x}\otimes\Big( \ket{\psi^y_D}-\ket{\psi^{y'}_D}\Big)_D}\nonumber\\
&\quad + \sum_{y\in[N]} \norm{\sum_{D\in\mathbf{D}_\mathrm{adm}}\gamma_{y,D}\frac{1}{\sqrt{N}}\ket{\bot}_{D^1_x}\otimes\Big( \ket{\psi^y_D}-\ket{\psi^\bot_D}\Big)_D}\\
&\quad \text{(By triangle inequality over $y$)}\nonumber\\
&=\sum_{y\in[N]}\sqrt{\sum_{D\in\mathbf{D}_\mathrm{adm}}\abs{\gamma_{y,D}}^2\frac{1}{N^2}\sum_{y'\in[N]}\lVert\ket{y'}_{D^1_x}\otimes\Big( \ket{\psi^y_D}-\ket{\psi^{y'}_D}\Big)_D\rVert^2}\nonumber\\
&\quad + \sum_{y\in[N]}\sqrt{\sum_{D\in\mathbf{D}_\mathrm{adm}}\abs{\gamma_{y,D}}^2\frac{1}{N}\lVert\ket{\bot}_{D^1_x}\otimes\Big( \ket{\psi^y_D}-\ket{\psi^\bot_D}\Big)_D\rVert^2}\\
&\leq\sum_{y\in[N]}\sqrt{\sum_{D\in\mathbf{D}_\mathrm{adm}}\abs{\gamma_{y,D}}^2\frac{1}{N^2}\sum_{y'\in[N]}O\Big(\frac{t}{N}\Big)} + \sum_{y\in[N]}\sqrt{\sum_{D\in\mathbf{D}_\mathrm{adm}}\abs{\gamma_{y,D}}^2\frac{1}{N}O\Big(\frac{t}{N}\Big)} \label{eq:something}\\
&=O(\sqrt{t}/N)\sum_{y\in[N]}\sqrt{\sum_{D\in\mathbf{D}_\mathrm{adm}}\abs{\gamma_{y,D}}^2} \\
&\leq O(\sqrt{t}/N)\sqrt{\sum_{y\in[N]}\sum_{D\in\mathbf{D}_\mathrm{adm}}\abs{\gamma_{y,D}}^2}\sqrt{\sum_{y\in[N]}1} \quad \text{(By Cauchy-Schwarz)}\\
&\leq O(\sqrt{t/N})
\end{align} where in \eqref{eq:something} we used the orthogonality of $\ket{y'}$'s and $\ket{\psi^y_D}-\ket{\psi^{y'}_D}$'s as given in \Cref{lem:p1-psi-y-d}. 
  
 \item We bound $\norm{[\vr,\fc]\fp_{D^2}\vf\fp_{D^2}\Gamma_1 \Pi^\star  \Pi^t}$. Observe that $\Gamma_1$ commutes with $\fp_{D^2}\vf\fp_{D^2}$, and $\fp_{D^2}\vf\fp_{D^2}$ at most doubles the database $D^1$ and $D^2$ while preserving the injectivity, then \begin{align}
\norm{[\vr,\fc]\fp_{D^2}\vf\fp_{D^2}\Gamma_1 \Pi^\star  \Pi^t} &\leq \norm{[\vr,\fc]\Gamma_1\Pi^{inj}\Pi^{2t}\fp_{D^2}\vf\fp_{D^2}\Pi^\star  \Pi^t}\\
&\leq \norm{[\vr,\fc]\Gamma_1\Pi^{inj}\Pi^{2t}}
 \end{align} by submultiplicativity of the norm. Due to $\Gamma_1$, we can fix $\D^3_{x\oplus k_1}=\bot$, and since orthogonality of all registers other than $D^1_x$ are preserved by the commutator, we fix them in the computational basis and write a general state \begin{align*}
     \ket{\psi}\coloneqq \sum_{y\in[N]}\alpha_y\ket{y}_{D^1_x} \ket{(D^1_x)^c}_{(D^1_x)^c} \ket{D^2}_{D^2} \ket{\bot}_{D^3_{x\oplus k_1}}\ket{(D^3_{x\oplus k_1})^c}_{(D^3_{x\oplus k_1})^c}\ket{k_1}_{K_1}\ket{k_2}_{K_2}\ket{k_3}_{K_3}
 \end{align*} where $D^2\in\mathbf{D}_{2t,inj}$. Next observe that $\vr\ppar{x}$ is the only component of $\vr$ that acts on $D^1_x$ non-trivially, and the action is in a controlled way. The action of $\vr\ppar{x}$ on the state becomes trivial when the value $D^1(x)\oplus k_2 \not \in \dom{D^2}$. So, when we project out the images of $x$ in $\fc\ket{\psi}$ such that $D^1(x)\oplus k_2 \in \dom{D^2}$, the operators commute. We write \begin{align}
   & \fc\vr\ket{\psi}=\fc\Bigg( \sum_{y\in[N]:y\oplus k_2\in\dom{D^2}}\alpha_y\ket{y}_{D^1_x}\ket{D^2}_{D^2}\ket{k_2}_{K_2}\otimes T_{D^3_{x\oplus k_1}}^{(D^2(y\oplus k_2))} \nonumber\\
    &\quad+\sum_{y\in[N]:y\oplus k_2\not\in\dom{D^2}}\alpha_y\ket{y}_{D^1_x}\ket{D^2}_{D^2}\ket{k_2}_{K_2}\otimes \Id_{D^3_{x\oplus k_1}}\Biggl)\\
    &=\sum_{y\in[N]:y\oplus k_2\in\dom{D^2}}\alpha_y\Big(\ket{y}-\frac{\ket{+^n}}{\sqrt{N}} + \frac{\ket{\bot}}{\sqrt{N}}\Big)_{D^1_x}\ket{D^2}_{D^2}\ket{k_2}_{K_2}\otimes T_{D^3_{x\oplus k_1}}^{(D^2(y\oplus k_2))}\nonumber\\
    &\quad+\sum_{y\in[N]:y\oplus k_2\not\in\dom{D^2}}\alpha_y\Big(\ket{y}-\frac{\ket{+^n}}{\sqrt{N}} + \frac{\ket{\bot}}{\sqrt{N}}\Big)_{D^1_x}\ket{D^2}_{D^2}\ket{k_2}_{K_2}\otimes \Id_{D^3_{x\oplus k_1}}
\end{align} and 
\begin{align}
    &\vr\fc\ket{\psi}\\
    &=\vr  \Bigg( \sum_{y\in[N]}\alpha_y\Big(\ket{y}-\frac{\ket{+^n}}{\sqrt{N}} + \frac{\ket{\bot}}{\sqrt{N}}\Big)_{D^1_x}\ket{D^2}_{D^2}\ket{k_2}_{K_2}\Bigg)\\
    &=\sum_{y\in[N]:y\oplus k_2\in\dom{D^2}}\alpha_y\ket{y}_{D^1_x}\ket{D^2}_{D^2}\ket{k_2}_{K_2}\otimes T_{D^3_{x\oplus k_1}}^{(D^2(y\oplus k_2))}\nonumber\\
    &\quad+\sum_{y\in[N]:y\oplus k_2\not\in\dom{D^2}}\alpha_y\ket{y}_{D^1_x}\ket{D^2}_{D^2}\ket{k_2}_{K_2}\otimes \Id_{D^3_{x\oplus k_1}}\nonumber\\
    &\quad - \sum_{y\in[N]:y\oplus k_2\in\dom{D^2}}\alpha_y\frac{1}{N}\sum_{y'\in[N]:y'\oplus k_2 \in\dom{D^2}}\ket{y'}_{D^1_x}\ket{D^2}_{D^2}\ket{k_2}_{K_2}\otimes T_{D^3_{x\oplus k_1}}^{(D^2(y'\oplus k_2))}\nonumber\\
    &\quad - \sum_{y\in[N]:y\oplus k_2\in\dom{D^2}}\alpha_y\frac{1}{N}\sum_{y'\in[N]:y'\oplus k_2 \not \in\dom{D^2}}\ket{y'}_{D^1_x}\ket{D^2}_{D^2}\ket{k_2}_{K_2}\otimes \Id_{D^3_{x\oplus k_1}}\nonumber\\
    &\quad - \sum_{y\in[N]:y\oplus k_2\not\in\dom{D^2}}\alpha_y\frac{1}{N}\sum_{y'\in[N]:y'\oplus k_2 \in\dom{D^2}}\ket{y'}_{D^1_x}\ket{D^2}_{D^2}\ket{k_2}_{K_2}\otimes T_{D^3_{x\oplus k_1}}^{(D^2(y'\oplus k_2))}\nonumber\\
    &\quad - \sum_{y\in[N]:y\oplus k_2\not\in\dom{D^2}}\alpha_y\frac{1}{N}\sum_{y'\in[N]:y'\oplus k_2 \notin\dom{D^2}}\ket{y'}_{D^1_x}\ket{D^2}_{D^2}\ket{k_2}_{K_2}\otimes \Id_{D^3_{x\oplus k_1}}\nonumber\\
    &\quad + \sum_{y\in[N]:y\oplus k_2\in\dom{D^2}}\alpha_y\frac{1}{\sqrt{N}}\ket{\bot}_{D^1_x}\ket{D^2}_{D^2}\ket{k_2}_{K_2}\otimes \Id_{D^3_{x\oplus k_1}}\nonumber\\
    &\quad + \sum_{y\in[N]:y\oplus k_2\not\in\dom{D^2}}\alpha_y\frac{1}{\sqrt{N}}\ket{\bot}_{D^1_x}\ket{D^2}_{D^2}\ket{k_2}_{K_2}\otimes \Id_{D^3_{x\oplus k_1}}
\end{align}
then we bound the commutator by looking at the difference when there is a $T$ map with different superscripts and when one has $T$ map while other has $\Id$, namely the third, fourth, fifth and seventh terms on the last equation, as 
\begin{align}
&\norm{[\vr,\fc]\Gamma_1 \Pi^{inj} \Pi^{2t}}\\
&=\norm{[\vr,\fc]\ket{\psi}}\\
&\leq \norm{\sum_{y\in[N]:y\oplus k_2\in\dom{D^2}}\frac{\alpha_y}{N}\sum_{y'\in[N]:y'\oplus k_2 \in\dom{D^2}}\ket{y'}_{D^1_x}\ket{D^2}_{D^2}\ket{k_2}_{K_2}\otimes (T_{D^3_{x\oplus k_1}}^{(D^2(y\oplus k_2))}-T_{D^3_{x\oplus k_1}}^{(D^2(y'\oplus k_2))})}\nonumber\\
& \quad+ \norm{\sum_{y\in[N]:y\oplus k_2\in\dom{D^2}}\frac{\alpha_y}{N}\sum_{y'\in[N]:y'\oplus k_2 \not \in\dom{D^2}}\ket{y'}_{D^1_x}\ket{D^2}_{D^2}\ket{k_2}_{K_2}\otimes (T_{D^3_{x\oplus k_1}}^{(D^2(y\oplus k_2))}- \Id_{D^3_{x\oplus k_1}})}\nonumber\\
&\quad+ \norm{\sum_{y\in[N]:y\oplus k_2\not\in\dom{D^2}}\frac{\alpha_y}{N}\sum_{y'\in[N]:y'\oplus k_2 \in\dom{D^2}}\ket{y'}_{D^1_x}\ket{D^2}_{D^2}\ket{k_2}_{K_2}\otimes (\Id_{D^3_{x\oplus k_1}}- T_{D^3_{x\oplus k_1}}^{(D^2(y'\oplus k_2))})}\nonumber\\
&\quad +\norm{\sum_{y\in[N]:y\oplus k_2\in\dom{D^2}}\frac{\alpha_y}{\sqrt{N}}\ket{\bot}_{D^1_x}\ket{D^2}_{D^2}\ket{k_2}_{K_2}\otimes (T_{D^3_{x\oplus k_1}}^{(D^2(y\oplus k_2))}-\Id_{D^3_{x\oplus k_1}})} \\
&\leq \norm{\sum_{y\in[N]:y\oplus k_2\in\dom{D^2}}\frac{\alpha_y}{N}\sum_{y'\in[N]:y'\oplus k_2 \in\dom{D^2}}\ket{y'}_{D^1_x}\ket{D^2}_{D^2}\ket{k_2}_{K_2}\otimes T_{D^3_{x\oplus k_1}}^{(D^2(y\oplus k_2))}}\nonumber\\
&\quad +\norm{\sum_{y\in[N]:y\oplus k_2\in\dom{D^2}}\frac{\alpha_y}{N}\sum_{y'\in[N]:y'\oplus k_2 \in\dom{D^2}}\ket{y'}_{D^1_x}\ket{D^2}_{D^2}\ket{k_2}_{K_2}\otimes T_{D^3_{x\oplus k_1}}^{(D^2(y'\oplus k_2))}} \nonumber\\
&\quad + \norm{\sum_{y\in[N]:y\oplus k_2\in\dom{D^2}}\frac{\alpha_y}{N}\sum_{y'\in[N]:y'\oplus k_2 \not \in\dom{D^2}}\ket{y'}_{D^1_x}\ket{D^2}_{D^2}\ket{k_2}_{K_2}\otimes T_{D^3_{x\oplus k_1}}^{(D^2(y\oplus k_2))}} \nonumber\\
&\quad + \norm{\sum_{y\in[N]:y\oplus k_2\in\dom{D^2}}\frac{\alpha_y}{N}\sum_{y'\in[N]:y'\oplus k_2 \not \in\dom{D^2}}\ket{y'}_{D^1_x}\ket{D^2}_{D^2}\ket{k_2}_{K_2}\otimes \Id_{D^3_{x\oplus k_1}}} \nonumber\\
&\quad + \norm{\sum_{y\in[N]:y\oplus k_2\not\in\dom{D^2}}\frac{\alpha_y}{N}\sum_{y'\in[N]:y'\oplus k_2 \in\dom{D^2}}\ket{y'}_{D^1_x}\ket{D^2}_{D^2}\ket{k_2}_{K_2}\otimes \Id_{D^3_{x\oplus k_1}}} \nonumber\\
&\quad + \norm{\sum_{y\in[N]:y\oplus k_2\not\in\dom{D^2}}\frac{\alpha_y}{N}\sum_{y'\in[N]:y'\oplus k_2 \in\dom{D^2}}\ket{y'}_{D^1_x}\ket{D^2}_{D^2}\ket{k_2}_{K_2}\otimes T_{D^3_{x\oplus k_1}}^{(D^2(y'\oplus k_2))}} \nonumber\\
&\quad +\norm{\sum_{y\in[N]:y\oplus k_2\in\dom{D^2}}\frac{\alpha_y}{\sqrt{N}}\ket{\bot}_{D^1_x}\ket{D^2}_{D^2}\ket{k_2}_{K_2}\otimes \Id_{D^3_{x\oplus k_1}}} \nonumber\\
&\quad+ \norm{\sum_{y\in[N]:y\oplus k_2\in\dom{D^2}}\frac{\alpha_y}{\sqrt{N}}\ket{\bot}_{D^1_x}\ket{D^2}_{D^2}\ket{k_2}_{K_2}\otimes T_{D^3_{x\oplus k_1}}^{(D^2(y\oplus k_2))}} \\
&\leq \sqrt{\sum_{y\in[N]:y\oplus k_2\in\dom{D^2}}\frac{\abs{\alpha_y}^2}{N^2}\sum_{y'\in[N]:y'\oplus k_2 \in\dom{D^2}}\lVert\ket{y'}_{D^1_x}\ket{D^2}_{D^2}\ket{k_2}_{K_2}\otimes \ket{D^2(y\oplus k_2)}_{D^3_{x\oplus k_1}}\rVert^2 }\nonumber\\
&\quad+ \sqrt{\frac{|\sum_{y\in[N]:y\oplus k_2\in\dom{D^2}}\alpha_y|^2}{N^2}\sum_{y'\in[N]:y'\oplus k_2 \in\dom{D^2}}\lVert\ket{y'}_{D^1_x}\ket{D^2}_{D^2}\ket{k_2}_{K_2}\otimes \ket{D^2(y'\oplus k_2)}_{D^3_{x\oplus k_1}}\rVert^2}\nonumber\\
&\quad + \sqrt{\sum_{y\in[N]:y\oplus k_2\in\dom{D^2}}\frac{\abs{\alpha_y}^2}{N^2}\sum_{y'\in[N]:y'\oplus k_2 \not \in\dom{D^2}}\lVert\ket{y'}_{D^1_x}\ket{D^2}_{D^2}\ket{k_2}_{K_2}\otimes \ket{D^2(y\oplus k_2)}_{D^3_{x\oplus k_1}}\rVert^2}\nonumber\\
&\quad+ \sqrt{\frac{|\sum_{y\in[N]:y\oplus k_2\in\dom{D^2}}\alpha_y|^2}{N^2}\sum_{y'\in[N]:y'\oplus k_2 \not \in\dom{D^2}}\lVert\ket{y'}_{D^1_x}\ket{D^2}_{D^2}\ket{k_2}_{K_2}\otimes \ket{\bot}_{D^3_{x\oplus k_1}}\rVert^2}\nonumber\\
&\quad +\sqrt{\frac{|\sum_{y\in[N]:y\oplus k_2\not\in\dom{D^2}}\alpha_y|^2}{N^2}\sum_{y'\in[N]:y'\oplus k_2 \in\dom{D^2}}\lVert\ket{y'}_{D^1_x}\ket{D^2}_{D^2}\ket{k_2}_{K_2}\otimes \ket{\bot}_{D^3_{x\oplus k_1}}\rVert^2}\nonumber\\
&\quad + \sqrt{\frac{|\sum_{y\in[N]:y\oplus k_2\not\in\dom{D^2}}\alpha_y|^2}{N^2}\sum_{y'\in[N]:y'\oplus k_2 \in\dom{D^2}}\lVert\ket{y'}_{D^1_x}\ket{D^2}_{D^2}\ket{k_2}_{K_2}\otimes \ket{D^2(y'\oplus k_2)}_{D^3_{x\oplus k_1}}\rVert^2 }\nonumber\\
&\quad + \sqrt{\frac{|\sum_{y\in[N]:y\oplus k_2\in\dom{D^2}}\alpha_y|^2}{N}\lVert\ket{\bot}_{D^1_x}\ket{D^2}_{D^2}\ket{k_2}_{K_2}\otimes \ket{\bot}_{D^3_{x\oplus k_1}}\rVert^2}\nonumber\\
&\quad + \sqrt{\sum_{y\in[N]:y\oplus k_2\in\dom{D^2}}\frac{\abs{\alpha_y}^2}{N}\lVert\ket{\bot}_{D^1_x}\ket{D^2}_{D^2}\ket{k_2}_{K_2}\otimes \ket{D^2(y\oplus k_2)}_{D^3_{x\oplus k_1}}\rVert^2}\\
&\leq O(\sqrt{t/N})
          \end{align} by using the injectivity of $D^2$ sometimes to keep use orthogonality of $y$, and the standard Cauchy-Schwarz argument on square of sum, namely \begin{align}
              \Big|\sum_{y\in S}\alpha_y\Big|^2\leq \Bigl(\sum_{y\in S}\abs{\alpha_y}^2\Bigr) \Bigl(\sum_{y\in S} 1\Bigr) \leq |S|,\quad S\subseteq [N].
          \end{align} 
          \end{enumerate}   
      \end{enumerate}
\end{proof}

\subsubsection{$P_2$-query}
Essentially, the $P_2$-query proofs are analogous to the $P_1$-query ones, except with additional complications arising from the fact that $P_2$ sits in the middle and can be pointed to by both $D^1$ and $D^3$. For the proofs, see \Cref{app:p2-query}.
\begin{theorem}[$P_2$-query]\label{thm:p2-bound}
Let the registers $B$ and $I$ hold $b=0$ and $i=2$, i.e., a forward query to $P_2$ in both worlds. Then for every integer $t \geq 1$ with $t^2 \leq N/2$,
\begin{align}
    \lVert (\cco^R V - V \Pi^\star \cco^I) \Pi^\star \Pi^t \Pi^{k_1,k_3} \rVert \leq O\big(\sqrt{t/N}\big).
\end{align}
\end{theorem}
\begin{proof}
    See \Cref{app:p2-query}.
\end{proof}

\subsubsection{$\pkac/P_3$-query} \label{sec:\pkac/P_3-query}

In this query the oracles in the real and ideal worlds differ. To return the same answer to the adversary's output register, the isometry must perform the fill and the removal on the queried register of the database, namely $D^3_x$. This requires the factor $\vf\ppar{x \oplus k_1}$ to be active, that is, to perform the fill $\ket{\bot}\ket{\bot} \mapsto \ket{\psi_{k_2}}$, which in turn requires that both of its target registers be $\bot$; $\Gamma_1$ and $\Gamma_2$ enforce this for the relevant indices of $D^1$ and $D^2$ registers, respectively. Once $\vf\ppar{x \oplus k_1}$ is active, $\vr\ppar{x \oplus k_1}$ automatically does the rest of the work of uncomputing $D^3_x$.

The remaining factors of the isometry do not play a role here. We therefore split the isometry into a product controlled on the query register as in \Cref{def:vp-vpc-main}. After bounding the difference between $\vpc\vp$ and $V$ in \Cref{lem:p-vpvpc-V}, we discard the irrelevant parts, namely $\vpc$, by bounding its commutator with $\cco^R$ and using the unitarity of the norm in \Cref{lem:comm-real-vpc}. Hence, we have removed the order dependence of the factors in $V$. 

 \begin{definition} \label{def:vp-vpc-main}
    \begin{align*}
        \vp&\coloneqq \sum_{x,k_1\in[N]} \proj{x}_X \otimes \proj{k_1}_{K_1} \otimes \vr\ppar{x\oplus k_1}\fp_{D_2} \vf\ppar{x\oplus k_1} \fp_{D_2}\\
        \vpc&\coloneqq \sum_{x,k_1\in[N]} \proj{x}_X \otimes \proj{k_1}_{K_1} \otimes \prod_{\substack{x'\in[N]:\\x'\neq x\oplus k_1}}\vr\ppar{x'}\fp_{D_2}  \prod_{\substack{x'\in[N]:\\x'\neq x\oplus k_1}}\vf\ppar{x'} \fp_{D_2}
    \end{align*}
    For fixed $x, k_1\in[N]$ at registers $X,K_1$, respectively, we use the notation
    \begin{align*}
      \vfpc\coloneqq  \prod_{x'\in[N]:x'\neq x\oplus k_1}\vf\ppar{x'} \text{ and } \vrpc\coloneqq  \prod_{x'\in[N]:x'\neq x\oplus k_1}\vr\ppar{x'}. 
    \end{align*}
\end{definition}

\begin{theorem}[$P_3$-query]\label{thm:p3-bound}
Let the registers $B$ and $I$ hold $b=0$ and $i=3$, i.e., a forward query to the construction in the real world and to $P_3$ in the ideal world. Then for every integer $t \geq 1$ with $t^2 \leq N/2$,
\begin{align}
    \lVert (\cco^R V - V \Pi^\star \cco^I) \Pi^\star \Pi^t \Pi^{k_1} \Pi^{k_3} \rVert \leq O\big(\sqrt{t/N}\big).
\end{align}
\end{theorem}

\begin{proof}
    See \Cref{app:p3-query-lemma}.
\end{proof}

\bibliography{ref,cryptobib/crypto,cryptobib/abbrev3}
\pagebreak
\appendix
\section{Partial Decompression Deferred Proof}\label{app:partial-dec}

\subsection{Proof of \Cref{cor:part-decomp}}.
\begin{proof}
    Fix the query register $X$ to $\ket{x}$ and the complementary
    database register $(D_x)^c$ to a computational basis value $D'$.
    On the support of $\Pi^{inj}_D\Pi^t_D$, the partial database $D'$
    is injective and has size at most $t$. Define
    \[
        \mathcal{S}_{D'}\coloneqq [N]\setminus\im{D'}.
    \]
    Conditioned on $D'$, the operators $\fc$ and $\pc$ act on $D_x$
    as $\operatorname{Exc}(\ket{\bot},\ket{+^n})$ and
    $\operatorname{Exc}(\ket{\bot},\ket{\mathcal{S}_{D'}})$, respectively,
    with the latter restricted to the injective subspace.
    Since
    \[
        N-|\mathcal{S}_{D'}|=|\im{D'}|=|D'|\leq t,
    \]
    \Cref{lem:sanitized-comp}, applied with $(D_x)^c$ as the
    auxiliary register, bounds their difference by
    $O(\sqrt{t/N})$ on these blocks. Restricting the input further
    by $\Pi^{inj}_D\Pi^t_D$ cannot increase the norm.
    Finally, both operators preserve $X$ in the computational
    basis, so taking the maximum over $x$ proves the claim.
\end{proof}
\section{One Query Impossibility Deferred Proof}\label{app:one-query}

We give a proof of the information theoretic bound using repeated swap
tests, an additional consistency test for query preimage sets, and
an elementary argument for combining the tests over all keys.

Write $\kappa=\lceil\log_2|K|\rceil$, $N=2^n$, and $D=2^m$, and assume
\eqref{eq:one-query-size-regime}. For a permutation $F$ of $[L]$,
we will write the Choi state as
\[
    \ket F
    :=\frac{1}{\sqrt L}\sum_{x\in[L]}\ket{x,F(x)}.
\]
Thus the distinguisher prepares
\[
    \ket{\psi_{P,Q}}
    =\ket P^{\otimes t}\otimes\ket Q^{\otimes t}
\]
using $t$ queries to each oracle, or $2t$ queries in total.
We can recall the ideal and real views of this distinguisher as
\begin{align}
    \rho_I
    &=\mathbb E_{P\sim S_N,\,Q\sim S_D}
      \left[\proj{\psi_{P,Q}}\right],
    \\
    \rho_R
    &=\mathbb E_{P\sim S_N,\,k\sim K}
      \left[\proj{\psi_{P,Q_k^P}}\right],
\end{align}
where the choices in each expectation are independent and
$Q_k^P(x)=b_k(x,P(a_k(x)))$ is a permutation of $[D]$ for every $P$ and $k$.
We first prove a helpful lemma identifying the structure imposed by this bijectivity requirement.

\begin{lemma}
\label{lem:one-query-structure}
Fix a key $k$, and define $X_{k,u}:=a_k^{-1}(u)$ for every $u\in[N]$.
There are pairwise disjoint sets
\[
    \{Y_{k,u}\}_{u\in[N]},\qquad
    \{Z_{k,y}\}_{y\in[N]},
\]
whose union is $[D]$, independent of $P$, such that for every $u,y\in[N]$,
\[
    b_k(\,\cdot\,,y):X_{k,u}\longrightarrow Y_{k,u}\sqcup Z_{k,y}
\]
is a bijection. Moreover, there is an integer $d_k\ge0$ such that
\[
    |Z_{k,y}|=d_k,\qquad
    |X_{k,u}|=|Y_{k,u}|+d_k,
    \qquad
    D=\sum_u|Y_{k,u}|+Nd_k.
\]
In particular, $d_k=0$ whenever $a_k$ is not surjective, and hence
whenever $D<N$.
\end{lemma}

\begin{proof}
Fix $k$ and suppress its subscript. For every $u,y\in[N]$, the map
$b(\,\cdot\,,y)$ is injective on $X_u$, since it is the restriction of
a cipher with $P(u)=y$. Define
\[
    S_u(y):=\{b(x,y):x\in X_u\}.
\]
For $u\ne u'$ and $y\ne y'$, there is a permutation $P$ satisfying
$P(u)=y$ and $P(u')=y'$. Bijectivity of $Q^P$ therefore implies
\[
    S_u(y)\cap S_{u'}(y')=\varnothing.
\]
View $S_u(y)$ as an $N\times N$ matrix of sets, with rows indexed by
$u$ and columns by $y$. Any two occurrences of the same output symbol
share a row or a column. Consequently, all occurrences of that symbol
lie in a single row or a single column: if two occurrences have
different columns, they share a row, and every other occurrence must
lie in that row. The case of different rows is symmetric.

Each output symbol therefore occurs in at most $N$ cells. On the other
hand, the total number of occurrences is
\[
    \sum_{u,y}|S_u(y)|=N\sum_u|X_u|=ND.
\]
Since there are $D$ symbols, each occurs in exactly $N$ cells and fills
a complete row or a complete column. Let $Y_u$ contain the symbols
filling row $u$, and let $Z_y$ contain the symbols filling column $y$.
These sets partition $[D]$, and $S_u(y)=Y_u\sqcup Z_y$.
The cardinality identity
\[
    |X_u|=|Y_u|+|Z_y|
\]
shows that all $Z_y$ have the same size $d$. Summing over $u$ gives
the claimed expression for $D$. If some $X_u$ is empty, the same
identity forces $d=0$.
\end{proof}

\begin{lemma}
For some $t=O(\kappa+1+n/m)$,
\[
    \norm{\rho_I-\rho_R}_1=\Omega(1).
\]
\label{lem:views-indist}
\end{lemma}

\begin{proof}
We may assume $D\ge16$ and $t\le D/16$; these conditions hold for all
sufficiently large parameters with our choice of $t$. For each candidate
key $k\in K$, we define a projector $\Pi_k$ with perfect completeness
for that key and a small acceptance probability in the ideal experiment.
Use the sets from \Cref{lem:one-query-structure}, and write
\[
    Z_k:=\bigcup_{y\in[N]}Z_{k,y},\qquad
    \alpha_k:=\frac{|Z_k|}{D}=\frac{Nd_k}{D}.
\]

\paragraph{Many outputs depending on the query answer: swap tests.}
Suppose first that $\alpha_k\ge1/4$. Enumerate
$Z_{k,y}=\{z_{k,y,1},\ldots,z_{k,y,d_k}\}$. For every $u,y,i$, there
is a unique $x_{k,u,y,i}\in X_{k,u}$ satisfying
\[
    b_k(x_{k,u,y,i},y)=z_{k,y,i}.
\]
Let $\mathcal H_k$ be the span of the corresponding basis vectors
$\ket{x_{k,u,y,i},z_{k,y,i}}$. The map
\[
    V_k:\ket{x_{k,u,y,i},z_{k,y,i}}\longmapsto\ket{u,y}\ket i
\]
is an isometry from $\mathcal H_k$ onto
$\mathbb C^N\otimes\mathbb C^N\otimes\mathbb C^{d_k}$, depending
only on the publicly known keyed functions. Let $T_k$ project onto
all basis vectors whose second register lies in $Z_k$, and let $H_k$
project onto $\mathcal H_k$. In particular, $H_k\preceq T_k$.
Changing variables in the definition of the Choi state gives
\begin{equation}
    V_kT_k\ket{Q_k^P}
    =\sqrt{\alpha_k}\ket P\ket{+_{d_k}},
    \qquad
    \ket{+_{d_k}}:=\frac1{\sqrt{d_k}}\sum_{i=1}^{d_k}\ket i.
    \label{eq:one-query-partial-choi}
\end{equation}
Here $T_k\ket{Q_k^P}$ belongs to $\mathcal H_k$.

For each pair consisting of one $P$ copy and one $Q$ copy, accept the
component in $I\otimes(I-T_k)$, reject the component in
$I\otimes(T_k-H_k)$, and on $I\otimes H_k$ apply $V_k$ and swap-test
the decoded pair of registers against the $P$ copy, leaving the $i$
register untouched. The accepting operator is a projector: on the
last subspace it is the conjugate of the symmetric-subspace projector,
and the three subspaces are orthogonal. Let $\Pi_k$ be the tensor
product of these accepting projectors on all $t$ pairs, expressed on
the original registers. Equation~\eqref{eq:one-query-partial-choi}
shows that in the real experiment,
\[
    \Pi_k\ket{\psi_{P,Q_k^P}}=\ket{\psi_{P,Q_k^P}}.
\]

In the ideal experiment, fix $P$ and write
\[
    \ket\phi:=(V_kH_k)\ket Q,\qquad
    \beta:=\norm{\ket\phi}^2,\qquad
    \gamma:=\norm{(\bra P\otimes I)\ket\phi}^2.
\]
Because $Q$ is a permutation, its Choi state has weight $\alpha_k$
on outputs in $Z_k$, so $\beta\le\alpha_k$. The one-pair acceptance
probability is
\begin{equation}
    p_k(P,Q)=1-\alpha_k+\frac{\beta+\gamma}{2}
    \le1-\frac{\alpha_k}{2}+\frac\gamma2.
    \label{eq:repeated-swap-test}
\end{equation}
For $i\in[d_k]$, let $J_i$ count the $u\in[N]$ such that
\[
    Q(x_{k,u,P(u),i})=z_{k,P(u),i}.
\]
The inputs in these conditions are distinct over all pairs $(u,i)$.
Thus, writing $J:=|\{x\in[D]:Q(x)=Q_k^P(x)\}|$, we have
\[
    \gamma=\frac1{ND}\sum_iJ_i^2
    \le\frac1D\sum_iJ_i\le\frac JD.
\]
For independent uniform $Q\sim S_D$, $J$ is distributed as the number
of fixed points of a uniform permutation. For $1\le r\le D$,
\[
    \Pr[J\ge r]\le\mathbb E\binom Jr=\frac1{r!}.
\]
Taking $r=\lceil D/8\rceil$, equation~\eqref{eq:repeated-swap-test}
gives $p_k(P,Q)\le15/16$ whenever $J<r$. Consequently,
\begin{equation}
    \operatorname{Tr}(\Pi_k\rho_I)
    \le\left(\frac{15}{16}\right)^t
       +\frac1{\lceil D/8\rceil!}.
    \label{eq:ideal-swap-bound}
\end{equation}

\paragraph{Few outputs depending on the query answer: consistency tests.}
Suppose now that $\alpha_k<1/4$. Measure the $t$ copies of $\ket Q$
in the computational basis, obtaining $(x_j,z_j)$ for $j\in[t]$.
Accept if
\[
    z_j\in Y_{k,a_k(x_j)}\cup Z_k\qquad\text{for every }j,
\]
and, for every $u\in[N]$, there exists a common $y\in[N]$ such that
\[
    b_k(x_j,y)=z_j\qquad\text{whenever }a_k(x_j)=u.
\]
Let $\Pi_k$ be the diagonal projector for the event that this test
passes, acting as the identity on the $P$ registers. In the real
experiment, all conditions hold for the true key: the common answer
for preimage set $X_{k,u}$ is $y=P(u)$. Thus
\[
    \Pi_k\ket{\psi_{P,Q_k^P}}=\ket{\psi_{P,Q_k^P}}.
\]

We claim that in the ideal experiment,
\begin{equation}
    \operatorname{Tr}(\Pi_k\rho_I)
    \le e^{-t/32}+\frac{N}{(D/2)^{\lfloor t/4\rfloor}}.
    \label{eq:ideal-consistency-bound}
\end{equation}
Write $L:=\max_u|X_{k,u}|$. If $L\le D/2$, expose the measured
pairs sequentially. Conditional on any full previous transcript, the
unexposed part of $Q$ remains a uniform bijection. Let $h\le t\le D/16$
be the number of distinct inputs already exposed. Each prescribed
output set has size at most
\[
    |Y_{k,a_k(x)}|+|Z_k|\le L+\alpha_kD\le3D/4.
\]
The probability that the next pair satisfies its output-set condition
is therefore at most
\[
    \frac1{16}+\frac{3D/4}{15D/16}
    =\frac{69}{80}<\frac78.
\]
Multiplying these conditional bounds shows that the whole test passes
with probability at most $(7/8)^t\le e^{-t/32}$.

If $L>D/2$, fix a largest preimage set $X_{k,u}$. The measured inputs
are independent uniform elements of $[D]$, independent of $Q$. Let
$G$ be the number of distinct measured inputs in $X_{k,u}$. At each
step, conditional on all preceding inputs, the probability of adding
a new such input is at least
\[
    \frac{L-t}{D}\ge\frac7{16}.
\]
Thus $G$ stochastically dominates a binomial random variable with
parameters $t$ and $7/16$. The multiplicative Chernoff bound gives
\[
    \Pr[G<t/4]\le e^{-9t/224}\le e^{-t/32}.
\]
Conditional on any input sequence with $G\ge r:=\lfloor t/4\rfloor$,
choose its first $r$ distinct inputs in $X_{k,u}$. For any fixed $y$,
the common-answer condition prescribes distinct outputs at those
inputs, since $b_k(\,\cdot\,,y)$ is injective on $X_{k,u}$. A uniform
$Q\sim S_D$ satisfies these conditions with probability $1/(D)_r$,
where $(D)_r=D(D-1)\cdots(D-r+1)$. Taking a union bound over $y\in[N]$
and using $t\le D/16$ gives
\[
    \operatorname{Tr}(\Pi_k\rho_I)
    \le e^{-t/32}+\frac{N}{(D)_r}
    \le e^{-t/32}+\frac{N}{(D/2)^r}.
\]
This proves \eqref{eq:ideal-consistency-bound} in either case.

\paragraph{Combining the key tests.}
For every candidate key, equations~\eqref{eq:ideal-swap-bound}
and~\eqref{eq:ideal-consistency-bound} give
\begin{equation}
    \operatorname{Tr}(\Pi_k\rho_I)
    \le e^{-t/32}
       +\frac{N}{(D/2)^{\lfloor t/4\rfloor}}
       +\frac1{\lceil D/8\rceil!}
    =:\varepsilon_t.
    \label{eq:ideal-key-test-bound}
\end{equation}
We now combine the key tests. Rather than measuring them
sequentially, define
\[
    A:=\sum_{k\in K}\Pi_k,
\]
and let $B$ project onto the eigenvectors of $A$ with eigenvalues
strictly below $1/2$. The distinguisher measures $\{B,I-B\}$ and
outputs ``real'' on outcome $I-B$.

For a real state with true key $k$, we have $A\succeq\Pi_k$, so
\[
    B\Pi_k B
    \preceq BAB
    \preceq\frac12 B.
\]
Since $\Pi_k\ket{\psi_{P,Q}}=\ket{\psi_{P,Q}}$, it follows that
\[
    \norm{B\ket{\psi_{P,Q}}}^2
    =\norm{B\Pi_k\ket{\psi_{P,Q}}}^2
    \le\norm{B\Pi_k B}_\infty
    \le\frac12.
\]
Thus the real experiment is accepted with probability at least $1/2$.
For the ideal experiment, the definition of $B$ gives $I-B\preceq2A$.
Hence
\[
    \operatorname{Tr}((I-B)\rho_I)
    \le2\operatorname{Tr}(A\rho_I)
    \le2|K|\varepsilon_t.
\]
Choosing
\[
    t=\max\left\{
        \left\lceil32\ln(32|K|)\right\rceil,
        \;4\left\lceil\frac{\ln(32|K|N)}{\ln(D/2)}\right\rceil
    \right\}
\]
makes this at most
\[
    \frac18+\frac{2|K|}{\lceil D/8\rceil!}
    =\frac18+o(1).
\]
Indeed, $t=O(\kappa+1+n/m)$, so \eqref{eq:one-query-size-regime}
ensures $t\le D/16$ for all sufficiently large parameters. It also
implies $\kappa=o(D)$ and $D\to\infty$, which give
$|K|/\lceil D/8\rceil!=o(1)$. The acceptance probabilities therefore
differ by at least $3/8-o(1)$, proving the lemma.

All operations after preparing the Choi states depend only on the
publicly known keyed functions. They require no further oracle queries,
although they need not be computationally efficient. Inverse queries
are unnecessary.
\end{proof}

\section{Projector Violation Bound Deferred Proofs}\label{app:proj-bounds}
We collect the proofs of the single-query violation bounds for the database projectors.

\subsection{Proof of \Cref{lem:pi-t-violation}}
\begin{proof}[Proof of \Cref{lem:pi-t-violation}]
It suffices to show that $\fc$ maps $\Pi^t\mathcal{H}$ into $\Pi^{t+1}\mathcal{H}$, i.e., that $(\Pi^{t+1})^\perp\fc\Pi^t = 0$. Any basis state in $\Pi^t\mathcal{H}$ is of the form
\begin{align}
    \ket{\psi} = \ket{y}_{D_x}\bigotimes_{i\neq x}\ket{z_i}_{D_i},
\end{align}
where $|\{i: z_i\neq\bot\}|\leq t$. We consider two cases.\begin{enumerate}[leftmargin=0pt, labelindent=0pt, itemindent=*, label=(\arabic*)]

\item Let $y\neq\bot$, then $x\in\dom{D}$ already, and $\fc$ acts as the identity on $D_x$, so $|\dom{\fc(D)}|=|\dom{D}|\leq t\leq t+1$. Hence $\fc\ket{\phi}\in\Pi^{t+1}\mathcal{H}$.

\item Let $y=\bot$, then $x\notin\dom{D}$, and $\fc$ replaces $\ket{\bot}_{D_x}$ with $\ket{+_N}_{D_x} = \frac{1}{\sqrt{N}}\sum_{y'\in[N]}\ket{y'}_{D_x}$, each term of which has $y'\neq\bot$, so $|\dom{\fc(D)}|=|\dom{D}|+1\leq t+1$. Hence $\fc\ket{\phi}\in\Pi^{t+1}\mathcal{H}$.

In both cases $\fc\ket{\phi}\in\Pi^{t+1}\mathcal{H}$, so $(\Pi^{t+1})^\perp\fc\Pi^t = 0$, and therefore
\begin{align}
    \norm{(\Pi^{t+1})^\perp\fc\Pi^t} = 0. 
\end{align}
\end{enumerate}
\end{proof}

\subsection{Proof of \Cref{lem:pi-star-violation}}
\begin{proof}
	Write $C=\fc_{XD^i}$ and $P=\pu_{XYD^i}$ when the oracle index
	$i$ is fixed. The evaluation operator $P$ preserves the database in
	the computational basis and does not act on the key registers.
	Consequently it commutes with each of
	$\Pi^{inj},\Pi^g,\Pi^\star$.
	We first bound the commutators of these projectors with $C$.
	
	The operators under consideration are block diagonal in the query
	input $x$ and in all database entries outside the queried register
	$R=D^i_x$. We may therefore fix these registers in the computational
	basis, while leaving $R$, the key registers, the response register,
	and any auxiliary registers arbitrary. Let $E$ denote the database
	$D^i[x\mapsto\bot]$. It suffices to consider blocks
	with $|E|\leq t$ and each other database of size at most $t$.
	Within each such block we allow all values of $R$, even when this
	makes $|D^i|=t+1$. These enlarged blocks are invariant under $C$,
	$P$, and all three projectors, and contain the support of $\Pi^t$.
	In particular, we will not need to increase $t$ between the two
	compression operators of a query.
	
	\medskip
	\noindent\emph{Injectivity.}
	If one of the fixed databases, including $E$, is noninjective,
	then $\Pi^{inj}$ is zero throughout this block and its commutator
	with $C$ vanishes. Otherwise put $S=\im{E}$ and
	\[
	B_S=\sum_{y\in S}\proj{y}_R,
	\qquad J=\Id_R-B_S.
	\]
	On this block $\Pi^{inj}=J$, with identities on the remaining
	registers. Since
	\[
	C=\Id_R-\ket{v}\bra{v},
	\qquad \ket{v}=\ket{\bot}-\ket{+^n},
	\]
	we obtain
	\[
	\norm{B_S C J}
	=\norm{B_S\ket{v}}\,\norm{J\ket{v}}
	=\sqrt{\frac{|S|}{N}\left(2-\frac{|S|}{N}\right)}
	\leq\sqrt{\frac{2t}{N}}.
	\]
	For a self-adjoint operator $C$ and an orthogonal projector $J$,
	$\norm{[C,J]}=\norm{J^\perp C J}$: relative to $J\oplus J^\perp$,
	the two off-diagonal blocks of the commutator are adjoints up to
	sign. Thus, on every block under consideration,
	\begin{equation}\label{eq:pi-star-proof-inj-comm}
		\norm{[C,\Pi^{inj}]}\leq\sqrt{2t/N}.
	\end{equation}
	
	\medskip
	\noindent\emph{The good-key projector.}
	There is nothing to prove for $i=3$, since $\Pi^g$ does not act on
	$D^3$. For $i\in\{1,2\}$, let $\kappa_z$ be the admissible-key set
	when $R$ holds $z\in[N]\cup\{\bot\}$, with all the fixed database
	entries as above. Set
	\[
	G_z=\proj{g_z},\qquad
	\ket{g_z}=\frac{1}{\sqrt{|\kappa_z|}}
	\sum_{k\in\kappa_z}\ket{k}_{K_2}
	\]
	when $\kappa_z\neq\emptyset$, and set $G_z=0$ otherwise.
	In this block,
	\[
	\Pi^g=\sum_{z\in[N]\cup\{\bot\}}\proj{z}_R\otimes G_z.
	\]
	Adding one entry to $D^1$ adds at most one value to its image,
	whereas adding one entry to $D^2$ adds at most one point to its
	domain. In either case this excludes at most $t$ additional keys.
	Consequently,
	\[
	\kappa_z\subseteq\kappa_\bot,
	\qquad |\kappa_\bot\setminus\kappa_z|\leq t,
	\qquad |\kappa_\bot|\geq N-t^2>0.
	\]
	The overlap of the normalized uniform superpositions on two
	nested sets gives
	\[
	\norm{G_z-G_\bot}
	=\sqrt{1-\frac{|\kappa_z|}{|\kappa_\bot|}}
	\leq\delta_t,
	\qquad
	\delta_t\coloneqq\sqrt{\frac{t}{N-t^2}}.
	\]
	The same identity holds when $\kappa_z$ is empty, under the zero
	projector convention. The reference projector
	$G_0=\Id_R\otimes G_\bot$ commutes with $C$, and hence
	\begin{equation}\label{eq:pi-star-proof-good-comm}
		\norm{[C,\Pi^g]}
		=\norm{[C,\Pi^g-G_0]}
		\leq 2\norm{\Pi^g-G_0}
		\leq 2\delta_t.
	\end{equation}
	Notice that $K_2$ was not fixed: the estimate accounts for the
	coherent dependence of its state on the database.
	
	Using $\Pi^\star=\Pi^g\Pi^{inj}$ and the product rule for
	commutators now yields
	\[
	\norm{[C,\Pi^\star]}
	\leq 2\delta_t+\sqrt{2t/N}.
	\]
	Since $\delta_t\leq\sqrt{2t/N}$, all three commutator bounds are
	$O(\sqrt{t/N})$. Moreover,
	\[
	Q^\perp C Q=Q^\perp[C,Q]Q,
	\qquad
	Q^\perp P C Q=P Q^\perp C Q.
	\]
	Taking the direct sum of the block estimates proves the first
	two bounds of the lemma, including arbitrary auxiliary systems.
	
	\medskip
	\noindent\emph{Forward queries.}
	Let $U_i=CPC$ be the forward ideal-world query to $D^i$.
	Within the same invariant blocks, $[P,Q]=0$ implies
	\[
	[U_i,Q]=CP[C,Q]+[C,Q]PC,
	\qquad
	\norm{[U_i,Q]}\leq 2\norm{[C,Q]}.
	\]
	It follows that
	\begin{equation}\label{eq:pi-star-proof-forward}
		\norm{Q^\perp U_i Q\Pi^t}=O(\sqrt{t/N})
		\qquad
		(Q\in\{\Pi^{inj},\Pi^g,\Pi^\star\}).
	\end{equation}
	
	\medskip
	\noindent\emph{Inverse queries: injectivity and the joint projector.}
	The operator $\dm$ preserves database sizes and injectivity, so
	it commutes with $\Pi^t$ and $\Pi^{inj}$. On injective databases,
	\[
	\kappa_{D^1,D^2}
	=\kappa_{(D^2)^{-1},(D^1)^{-1}},
	\]
	because both complements consist of the keys
	$a\oplus b$ with $a\in\im{D^1}$ and $b\in\dom{D^2}$.
	Thus $\dm$ maps the normalized good-key superposition to the
	corresponding normalized good-key superposition after inversion
	and swapping. Together with preservation of the injective
	subspace, this proves $[\dm,\Pi^\star]=0$ on the whole space.
	This argument does not require the generally false unrestricted
	identity $[\dm,\Pi^g]=0$.
	
	An inverse query is $\dm U_j\dm$ for the appropriate rerouted
	index $j$. For $Q\in\{\Pi^{inj},\Pi^\star\}$, unitary invariance
	therefore gives
	\[
	\norm{Q^\perp\dm U_j\dm Q\Pi^t}
	=\norm{Q^\perp U_j Q\Pi^t}
	=O(\sqrt{t/N}).
	\]
	
	\medskip
	\noindent\emph{Inverse queries: goodness without an injectivity assumption.}
	We give a separate argument for $Q=\Pi^g$.
	The inverse ideal-world query to $D^i$ can equivalently be written
	$F_i U_i F_i$, where $F_i=\fp_{D^i}$: all the other flips and the
	swaps in the definition using $\dm$ cancel. The case $i=3$
	commutes with $\Pi^g$ exactly, so assume $i\in\{1,2\}$.
	Put
	\[
	Q_i'=F_i\Pi^g F_i.
	\]
	This projector is diagonal in the database basis, with a key
	projector in each database block. In particular, it commutes
	with $P$ and preserves the blocks obtained by fixing $x$, $E$,
	and the other databases as before.
	
	If $E$ is noninjective, every completion of $E$ is noninjective,
	so $F_i$ acts as the identity throughout the block.
	Thus $Q_i'=\Pi^g$ there, and
	\eqref{eq:pi-star-proof-good-comm} applies.
	
	Suppose instead that $E$ is injective. Set $S=\im{E}$,
	$B_S=\sum_{z\in S}\proj{z}_R$, and $J=\Id_R-B_S$.
	Replace $C$ within this block by
	\[
	\widehat C
	=\operatorname{Exc}\left(
	\ket{\bot},
	\frac{1}{\sqrt{N-|S|}}\sum_{z\notin S}\ket{z}
	\right).
	\]
	By \Cref{lem:diff-full-and-partial-decomp},
	\begin{equation}\label{eq:pi-star-proof-inverse-partial}
		\norm{C-\widehat C}=O(\sqrt{t/N}).
	\end{equation}
	The operator $\widehat C$ preserves $J$ and acts as the identity
	on $B_S$. Since $Q_i'$ also preserves these two subspaces, its
	commutator with $\widehat C$ is zero on $B_S$.
	
	On $J$, every database completion is injective. Write $E_z$ for
	the completion with value $z$ at $x$, where $E_\bot=E$.
	The admissible-key set defining the key projector of $Q_i'$ is
	\[
	\lambda_z=
	\begin{cases}
		[N]\setminus\{a\oplus b:
		a\in\dom{E_z},\ b\in\dom{D^2}\}, & i=1,\\[2pt]
		[N]\setminus\{a\oplus b:
		a\in\im{D^1},\ b\in\im{E_z}\}, & i=2.
	\end{cases}
	\]
	Indeed, inversion exchanges the domain and image of the
	injective queried database. For each $z$ allowed by $J$, we have
	\[
	\lambda_z\subseteq\lambda_\bot,
	\qquad |\lambda_\bot\setminus\lambda_z|\leq t,
	\qquad |\lambda_\bot|\geq N-t^2.
	\]
	The nested-set projector estimate used above therefore shows
	that, on $J$, $Q_i'$ differs by at most $\delta_t$ from the
	projector whose key component is the uniform state on
	$\lambda_\bot$, independently of $R$. This reference projector
	commutes with $\widehat C$. Hence
	\[
	\norm{[\widehat C,Q_i']}\leq 2\delta_t
	\]
	on the whole block, since the commutator vanishes on $B_S$.
	Combining this with \eqref{eq:pi-star-proof-inverse-partial} gives
	\[
	\norm{[C,Q_i']}
	\leq 2\norm{C-\widehat C}+2\delta_t
	=O(\sqrt{t/N}).
	\]
	The same estimate holds in the blocks with noninjective $E$.
	As $[P,Q_i']=0$, it follows on all these invariant blocks that
	\[
	\norm{[U_i,Q_i']}\leq 2\norm{[C,Q_i']}
	=O(\sqrt{t/N}).
	\]
	Finally, $F_i$ is a unitary involution and preserves $\Pi^t$, so
	\[
	\norm{(\Pi^g)^\perp F_i U_i F_i\Pi^g\Pi^t}
	=\norm{(Q_i')^\perp U_i Q_i'\Pi^t}
	=O(\sqrt{t/N}).
	\]
	This proves the remaining inverse-query bound. The oracle index
	and direction are preserved control registers, so taking their
	direct sum also proves the bound for coherent superpositions
	of query types.
\end{proof}

\subsection{Proof of \Cref{lem:gamma1-perp-bound}}
\begin{proof}
Observe that for each $P_i$-query, $\Gamma_1^\perp$ selects, for each $x$ and $D^i$, at most $t$ of the (almost) $N$ values of the relevant key; considering that each key brings $1/\sqrt{N}$ coefficient and the normalization, this gives $\sqrt{t/N}$. We compute explicitly in the rest. 
\begin{enumerate}[leftmargin=0pt, labelindent=0pt, itemindent=*, label=(\arabic*)]
\item For $P_1$-query, we take a general state in the range $\Pi^{k_1}\Pi^t$ with subnormalized $\ket{\psi_{x,D^3}}$'s \begin{align*}
    \ket{\psi_1}\coloneqq \sum_{x\in[N],D^3\in\mathbf{D}_t} \gamma_{x,D^3}\ket{x}_X\ket{D^3}_{D^3}\ket{+^n}_{K_1}\ket{\psi_{x,D^3}}.
\end{align*} Then, we write \begin{align}
  \lVert  \Gamma_1^\perp \ket{\psi_1} \rVert&= \norm{\sum_{x\in[N],D^3\in\mathbf{D}_t} \gamma_{x,D^3}\ket{x}_X\ket{D^3}_{D^3}\frac{1}{\sqrt{N}}\sum_{k_1:x\oplus k_1\in\im{D^3}}\ket{k_1}_{K_1}\ket{\psi_{x,D^3}}} \\
  &= \sqrt{\sum_{x\in[N],D^3\in\mathbf{D}_t} \abs{\gamma_{x,D^3}}^2\frac{1}{N}\sum_{k_1:x\oplus k_1\in\im{D^3}}\lVert\ket{\psi_{x,D^3}}\rVert^2}\\
  &\leq O(\sqrt{t/N}).
\end{align}
\item For a $P_2$-query, we treat $\Gamma_1$ as a product of the two projectors such that $\Gamma_1 = A B$ where
\begin{align}
   A &\coloneqq \sum_{x\in[N]}\proj{x}_X \otimes \Pmiss{x}{K_2}{D^1},\\
  \qquad
  B &\coloneqq \sum_{x\in[N]}\proj{x}_X \otimes \Bigl(
    \proj{\bot}_{D^2_x} \otimes \Id_{K_3 D^3}
    + \sum_{y\in[N]} \proj{y}_{D^2_x} \otimes \Pmiss{y}{K_3}{D^3}
  \Bigr) 
\end{align} and use $(AB)^\perp=A^\perp + B^\perp$
 For $A$, we take a general state in the range $\Pi^\star\Pi^t$ with subnormalized $\ket{\psi_{x,D^1, D^2}}$'s \begin{align*}
    \ket{\psi_2}\coloneqq \sum_{x\in[N],D^1,D^2\in\mathbf{D}_t, k_2\in \kappa_{D^1,D^2}} \frac{\gamma_{x,D^1,D^2}}{\sqrt{S_{D_1,D_2}}}\ket{x}_X\ket{D^1}_{D^1}\ket{D^2}_{D^2}\ket{k_2}_{K_2}\ket{\psi_{x,D^1,D^2}}.
\end{align*} we write \begin{align}
  \lVert  A^\perp \ket{\psi_2} \rVert&= \norm{\sum_{D^1,D^2\in\mathbf{D}_t, k_2\in \kappa_{D^1,D^2}} \frac{1}{\sqrt{S_{D_1,D_2}}}\ket{D^1}_{D^1}\ket{D^2}_{D^2}\ket{k_2}_{K_2}\sum_{x\in[N]:x\in\im{D^1}\oplus k_2}\gamma_{x,D^1,D^2}\ket{x}_X\ket{\psi_{x,D^1,D^2}}} \\
  &= \sqrt{\sum_{D^1,D^2\in\mathbf{D}_t, k_2\in \kappa_{D^1,D^2}} \frac{1}{\sqrt{S_{D_1,D_2}}}\sum_{x\in[N]:x\in\im{D^1}\oplus k_2}\abs{\gamma_{x,D^1,D^2}}^2\lVert\ket{\psi_{x,D^1,D^2}}\rVert^2}\\
  &\leq O(\sqrt{t/N-t^2})
\end{align} using the fact that $\abs{ \kappa_{D^1,D^2}}\leq N-t^2$ for $D^1,D^2\mathbf{D}_t$.

For $B$, we take a general state in the range $\Pi^{k_3}\Pi^t$ with subnormalized $\ket{\psi_{x,y, D^3}}$'s  \begin{align*}
    \ket{\psi_3}\coloneqq \sum_{x,y\in[N],D^3\in\mathbf{D}_t} \gamma_{x,y,D^3}\ket{x}_X\ket{y}_{D^2_x}\ket{D^3}_{D^3}\ket{+^n}_{K_3}\ket{\psi_{x,y,D^3}}.
\end{align*} Then, we write \begin{align}
  \lVert  B^\perp \ket{\psi_3} \rVert&= \norm{\sum_{x,y\in[N],D^3\in\mathbf{D}_t} \gamma_{x,y,D^3}\ket{x}_X\ket{y}_{D^2_x}\ket{D^3}_{D^3}\frac{1}{\sqrt{N}}\sum_{k_3:x\oplus k_3\in\im{D^3}}\ket{k_3}_{K_3}\ket{\psi_{x,y,D^3}}} \\
  &= \sqrt{\sum_{x,y\in[N],D^3\in\mathbf{D}_t} \abs{\gamma_{x,y,D^3}}^2\frac{1}{N}\sum_{k_3:x\oplus k_3\in\im{D^3}}\lVert\ket{\psi_{x,y,D^3}}\rVert^2}\\
  &\leq O(\sqrt{t/N}).
\end{align}

\item For $P_3$-query, it is exactly the same as $P_1$ except that $\dom{D^3}$ is replaced by $\dom{D^1}$.
  \end{enumerate} 
\end{proof}
\subsection{Proof of \Cref{lem:gamma2-perp-bound}}
\begin{proof}
 We take a general state in the range $\Pi^{k_3}\Pi^t$ with subnormalized $\ket{\psi_{x,y, D^2}}$'s \begin{align}
    \ket{\psi}\coloneqq \sum_{x\in[N],y\in[N]\cup\set{\bot}, D^2\in\mathbf{D}_t}\gamma_{x,y,D^2} \ket{x}_X\ket{y}_{D^3_x}\ket{D^2}_{D^2}\ket{+^n}_{K_3}\ket{\psi_{x,y,D^2}}.
\end{align} Then, noting that the $y=\bot$ component
satisfies $\Gamma_2$, we write\begin{align}
     \norm{\Gamma_2^\perp\ket{\psi}}&=\norm{ \Gamma_2^\perp\sum_{x\in[N],y\in[N]\cup\set{\bot}, D^2\in\mathbf{D}_t}\gamma_{x,y,D^2} \ket{x}_X\ket{D^2}_{D^2}\ket{+^n}_{K_3}\ket{\psi_{x,y,D^2}}}\\
    & =\norm{ \sum_{x,y\in[N], D^2\in\mathbf{D}_t}\gamma_{x,y,D^2} \ket{x}_X\ket{y}_{D^3_x}\ket{D^2}_{D^2}\ket{\psi_{x,y,D^2}}\frac{1}{\sqrt{N}}\sum_{k_3\in[N]:y\oplus k_3\in\im{D^2}}\ket{k_3}_{K_3} }\\
    &=\sqrt{ \sum_{x,y\in[N], D^2\in\mathbf{D}_t}\abs{\gamma_{x,y,D^2}}^2\lVert \ket{\psi_{x,y,D^2}}\rVert^2\frac{1}{N}\sum_{k_3\in[N]:y\oplus k_3\in\im{D^2}}1 }\\
    &\leq O(\sqrt{t/N}).
\end{align}

\end{proof}
\section{Hybrid Differences Deferred Proofs}\label{app:hybdrids-diff}
 We collect the proof used in reducing the distinguishing advantage to the sum of hybrid differences.
 \subsection{Proof of \Cref{lem:q-star-interleaved-vs-not}}
 \begin{proof}
     \begin{align}
    \norm{(I^{\star }_{q:1}- \Pi^\star I_{q:1})\ket{\psi^0}}&\leq \sum_{t=1}^{q-2} \norm{(\Pi^\star)^\perp\cco^I\Pi^\star\Pi^t} \\
    &\leq \sum_{t=1}^{q-2}O(\sqrt{t/N})\\
    &\leq O\Big(\frac{q\sqrt{q}}{\sqrt{N}}\Big)
\end{align} by \Cref{lem:pi-t-violation}, \ref{lem:pi-star-violation} and submultiplicativity of the norm.

Similarly, \begin{align}
    \norm{(\Pi^\star)^\perp I_{q:1}\ket{\psi^0} }\leq \sum_{t=1}^{q-1} \norm{(\Pi^\star)^\perp\cco^I\Pi^\star\Pi^t} \leq O\Big(\frac{q\sqrt{q}}{\sqrt{N}}\Big).
\end{align}
 \end{proof}

 \subsection{Proof of \Cref{thm:hybrid-diff}}

 \begin{proof}
    \begin{align}
    &\abs{\Pr[\algo A^{P_1, P_2, P_3}()=1]-\Pr[\algo A^{P_1, P_2, \pkac}()=1] } \\
   &\leq \frac{1}{2}\norm{ \Tr_{\dbs \keys}[\proj{\psi^q_R}] - \Tr_{\dbs \keys}[\proj{\psi^q_I}] }_1 \\
   &\leq\norm{ \ket{\psi^q_R } - V \ket{\psi^q_I} } \quad \text{(By \Cref{lem:approx-uhlman})}\\
   &= \Bigl\lVert \Bigl( R_{q:1}V  
   - V I_{q:1} \Bigr) \ket{\psi^0} \Bigr\rVert\\
   &\leq \Bigl\lVert \Bigl(R_{q:1}V 
   - R_{q:1}V \Pi^\star  \Bigr) \ket{\psi^0} \Bigr\rVert  + \Bigl\lVert \Bigl(R_{q:1}V \Pi^\star  - V \Pi^\star  I_{q:1}\Bigr) \ket{\psi^0} \Bigr\rVert + \Bigl\lVert \Bigl(V \Pi^\star  I_{q:1}  - V  I_{q:1}\Bigr) \ket{\psi^0} \Bigr\rVert\\
   &= \norm{\Bigl(R_{q:1}V \Pi^\star - V \Pi^\star  I_{q:1}\Bigr) \ket{\psi^0}} +\norm{ V(\Pi^\star - \Id)I_{q:1} \ket{\psi^0}} \text{ (Since $\Pi^\star \ket{\psi^0}=\ket{\psi^0}$)}\\ 
   &=\norm{\Bigl(R_{q:1}V \Pi^\star - V \Pi^\star  I_{q:1}\Bigr) \ket{\psi^0}} +\norm{ (\Pi^\star )^\perp I_{q:1} \ket{\psi^0} }\\
   &\leq \norm{\Bigl(R_{q:1}V \Pi^\star - V \Pi^\star  I_{q:1}\Bigr) \ket{\psi^0}} + O(\sqrt{q^3/N}) \quad \text{(By \Cref{lem:q-star-interleaved-vs-not})}
   \label{eq:term-moving-projectors} 
    \end{align} where in the one before the last line we used unitarity of norm to remove $V$ and the identity $\Pi^\perp = \Id- \Pi$ to replace projectors.
    
    We bound the first term in \eqref{eq:term-moving-projectors} by inserting \begin{align}
        0= \pm V I^{\star }_{q:1} \sum_{t=1,\dots, q-1} \pm R_{q:(t+1)} V I^{\star }_{t:1} ,\label{eq:zero-term}
    \end{align} while noting that $U^t$ commutes with $V,\Pi^{g},  \Pi^{inj} $ for $t\in \set{0,\dots, q}$, i.e.,
    \begin{align}
        &\norm{0+\Bigl(R_{q:1}V \Pi^\star - V \Pi^\star  I_{q:1}\Bigr) \ket{\psi^0}}\\
        &= \norm{ R_{q:1}V\Pi^\star  - R_{q:2} V I^{\star }_{1:1} + \sum_{t=2}^{q-1} (R_{q:t}V\Pi^\star  - R_{q:(t+1)} V I^{\star }_{t:1} )} + \norm{(V I^{\star }_{q:1}- V \Pi^\star I_{q:1})\ket{\psi^0}} \\
        &= \Bigl\lVert \Bigl( \sum_{t=1}^q U^q \cco^{R} U^{q-1} \cdots \cco^{R} U^t \Bigl( \cco^{R}V- V \Pi^{g}  \Pi^{inj}  \cco^{I}  \Bigr)\Pi^\star    U^{t-1} \cco^{I} \Pi^\star   \cdots U^0  \Pi^\star    \ket{\psi^0} \Bigr\rVert \nonumber\\
        &\quad + \norm{(V I^{\star }_{q:1}- V \Pi^\star I_{q:1})\ket{\psi^0}} \\
        & \leq \sum_{t=1}^q \Bigl\lVert U^q \cco^{R} U^{q-1} \cdots \cco^{R} U^t \Bigl( \cco^{R}V- V \Pi^\star   \cco^{I}  \Bigr)\Pi^\star    U^{t-1} \cco^{I} \Pi^\star   \cdots U^0  \Pi^\star   \ket{\psi^0} \rVert \nonumber\\
        &\quad + \norm{(V I^{\star }_{q:1}- V \Pi^\star I_{q:1})\ket{\psi^0}} \text{ (By triangle inequality)}\\
        & = \sum_{t=1}^q \Bigl\lVert  \Bigl( \cco^{R}V- V \Pi^\star   \cco^{I} \Bigr)  \Pi^\star  \Pi^\star   U^{t-1} \cco^{I} \Pi^\star   \cdots U^0  \Pi^\star    \ket{\psi^0} \Bigr\rVert+ \norm{(I^{\star }_{q:1}- \Pi^\star I_{q:1})\ket{\psi^0}}\\
        &\qquad\text{(By unitarity of the norm)} \nonumber\\
        & \leq \sum_{t=1}^q \Bigl\lVert  \Bigl( \cco^{R}V- V \Pi^\star   \cco^{I} \Bigr) \Pi^\star  \Pi^{t-1}\Pi^{k_1,k_3} \Bigr\rVert+ \norm{(I^{\star }_{q:1}- \Pi^\star I_{q:1})\ket{\psi^0}}\\
        & \leq \sum_{t=1}^q \Bigl\lVert  \Bigl( \cco^{R}V- V \Pi^\star   \cco^{I} \Bigr) \Pi^\star  \Pi^{t-1}\Pi^{k_1,k_3} \Bigr\rVert+ O(\sqrt{q^3/N})\label{eq:hybrids-with-keys} \quad \text{(By \Cref{lem:q-star-interleaved-vs-not})}
    \end{align} where in the one before the last line we used that the state
$\Pi^\star U^{t-1}\cco^{I}\Pi^\star\cdots U^{0}\Pi^\star\ket{\psi^0}$ lies in
the range of $\Pi^\star\Pi^{t-1}\Pi^{k_1,k_3}$: by
\Cref{lem:pi-t-violation} the first $t-1$ queries increase the size of each
database by at most $t-1$, and neither the ideal-world queries nor the
projectors act on the registers $K_1$ and $K_3$, so these retain their
uniform superpositions.
 \end{proof} 
\section{Reducing Inverse Queries to Forward Ones Deferred Proofs}\label{app:reduce-inv-to-fwd}

We collect the proofs for reducing inverse queries to forward ones, in particular comutation relations of $\dm$ with database projectors and the isometry $V$.

\subsection{Proof of \Cref{lem:comm-dm-and-proj}}
\begin{proof}
Recall the action of $\dm$ (omitting the registers it leaves unchanged except for $K_2$): \begin{gather}
    \ket{\psi}\coloneqq \frac{1}{\sqrt{\kappa_{D^1,D^2}}}\sum_{k_2\in \kappa_{D^1,D^2}}\ket{D^1}_{D^1}\ket{D^2}_{D^2}\ket{D^3}_{D^3} \ket{k_1}_{K_1} \ket{k_2}_{K_2} \ket{k_3}_{K_3},\\
    \dm\ket{\psi} =\frac{1}{\sqrt{\kappa_{D^1,D^2}}}\sum_{k_2\in \kappa_{D^1,D^2}}\ket{(D^2)^{-1}}_{D^1}\ket{(D^1)^{-1}}_{D^2}\ket{(D^3)^{-1}}_{D^3} \ket{k_3}_{K_1} \ket{k_2}_{K_2} \ket{k_1}_{K_3}.
\end{gather}
 The operator $\dm$ does not change the size of any database or introduces a collision by swapping and flipping, and only permutes the computational basis states spanning the subspaces $\Pi^{t}$ and $\Pi^{inj}$, hence $[\dm,\Pi^{t}]=[\dm, \Pi^{inj}]=0$. 
Now observe that considering only injective databases, otherwise if only one of $D^1$ or $D^2$ flips, then the admissible key set changes:  \begin{align}
  \ket{\psi}\in \Pi^g &\iff \im{D^1}\cap \im{D^2}\oplus k_2=\emptyset \text{ for all } k_2\in \kappa_{D^1,D^2}\\
  & \iff \im{(D^2)^{-1}} \cap \dom{(D^1)^{-1}} \oplus k_2=\emptyset \text{ for all } k_2\in \kappa_{D^1,D^2}\\
  & \iff \dm \ket{\psi}\in \Pi^g,
 \end{align} i.e., $\dm$ preserves both the good database subspace and its orthogonal complement, hence $[\dm, \Pi^{g}]\Pi^{inj}=0$.

 Moreover, the projector onto the uniform state of $K_1$ and $K_3$ is invariant under swapping the two registers, and hence commutes with $\dm$, i.e., $[\dm, \Pi^{k_1,k_3}]=0$.

\end{proof}
\subsection{Proof of \Cref{lem:v-dm-comm}}
\begin{proof} 
\begin{align}
&\norm{[\dm, V] \Pi^g \Pi^{inj} \Pi^t}  \\
    &= \norm{(\dm\vr\fp_{D^2}\vf\fp_{D^2} \pm \dm \fp_{D^2}\vf\fp_{D^2} -\vr\vr\fp_{D^2}\vf\fp_{D^2}\dm)\Pi^g \Pi^{inj} \Pi^t}\\
    &\leq \norm{[\dm,\fp_{D^2}\vf\fp_{D^2} ]\Pi^g\Pi^{inj} \Pi^t} + \norm{[\vr,\dm]\fp_{D^2}\vf\fp_{D^2}\Pi^g\Pi^{inj} \Pi^t} \label{eq:v-ssf-split}
\end{align} where triangle inequality and unitarity of the norm is used.

Intuitively, since $V$ is defined symmetrically for forward and inverse databases, $\dm$ commutes with $V$ as long as it acts identically before and after $V$. The only factor of $\dm$ sensitive to injectivity is the flip operator; since $\vf$ preserves injectivity, it commutes with $\dm$, and since $\vr$ introduces no new $D^3$ entry on good databases, it commutes with $\dm$ as well.

We bound each term separately. 
\begin{enumerate}[leftmargin=0pt, labelindent=0pt, itemindent=*, label=(\arabic*)]
    \item We bound the first term in \eqref{eq:v-ssf-split}. Consider the computational basis state \begin{align}
\ket{\psi}\coloneqq \frac{1}{\sqrt{|\kappa_{D^1,D^2}|}}\sum_{k_2\in\kappa_{D^1,D^2}}\ket{D^1}_{D^1}\ket{D^2}_{D^2}\ket{D^3}_{D^3} \ket{k_1}_{K_1} \ket{k_2}_{K_2} \ket{k_3}_{K_3} \text{ s.t. }\Pi^g\Pi^{inj}\Pi^t\ket{\psi}=\ket{\psi}.\label{eq:some-psi}
    \end{align} Since $k_1,k_3,D^1,D^2,D^3$ are fixed, we can fix the set $S(K_1,K_3, D^1, D^2,D^3)$ in \Cref{def:set-S}, and call it $S$ shortly, and say $\abs{S}=s$. Furthermore, since $K_2$ is untouched and stays in the product with the rest of the registers, we write the difference for a fixed $k_2$ branch in the rest and the result follows by orthogonality of $k_2$'s, while also noting that $\kappa_{D^1,D^2}$ is preserved by $\dm$ on injective databases. For fixed $k_2$ together with $D^1,D^2$, we can fix the set $J(K_2, D^1, D^2)$ in \Cref{def:J}, and call it $J_{k_2}$ shortly. For the state $\dm\ket{\psi}$, the same definitions are not meaningful but they are related as given in \Cref{rem:S-and-J-wrt-invdb}. We will use this relation and write the states by omitting the untouched registers.
    
First we write,
    \begin{align}
        &\dm \fp_{D^2}\vf\fp_{D^2}\ket{\psi}\\
        &= \dm \sum_{\substack{u_x\in J_{k_2}\text{ for } x\in S:\\\forall a\neq b\text{ }u_a\neq u_b }} \frac{1}{\sqrt{\abs{J_{k_2}}\dots (\abs{J_{k_2}}-s+1)}} \bigotimes_{x\in S} \ket{u_x}_{D^1_x}\ket{y_x \oplus k_3}_{D^2_{u_x\oplus k_2}}\ket{y_x}_{D^3_{x\oplus k_1}}\\
        &= \sum_{\substack{u_x\in J_{k_2}\text{ for } x\in S:\\\forall a\neq b\text{ }u_a\neq u_b }} \frac{1}{\sqrt{\abs{J_{k_2}}\dots (\abs{J_{k_2}}-s+1)}} \bigotimes_{x\in S}\ket{u_x\oplus k_2}_{D^1_{y_x\oplus k_3}} \ket{x}_{D^2_{u_x}}\ket{x\oplus k_1}_{D^3_{y_x}},
    \end{align} and next write \begin{align}
& \fp_{D^2}\vf\fp_{D^2}\dm\ket{\psi}\\
&=\fp_{D^2}\vf\fp_{D^2}\ket{(D^2)^{-1}}_{D^1}\ket{(D^1)^{-1}}_{D^2}\ket{(D^3)^{-1}}_{D^3} \ket{k_3}_{K_1} \ket{k_2}_{K_2} \ket{k_1}_{K_3}\\
&=\sum_{\substack{u_x \oplus k_2\in J_{k_2}'\text{ for } y_x \oplus k_3 \in S':\\\forall a\neq b\text{ }u_a\neq u_b }} \frac{1}{\sqrt{\abs{J_{k_2}}\dots (\abs{J_{k_2}}-s+1)}} \bigotimes_{y_x \oplus k_3 \in S'} \ket{u_x \oplus k_2}_{D^1_{y_x \oplus k_3}}\\&~~~~~~~~~~~\otimes\ket{(u_x\oplus k_2) \oplus k_2}_{D^2_{(x\oplus k_1) \oplus k_1}}\ket{x\oplus k_1}_{D^3_{(y_x\oplus k_3) \oplus k_3}}\\
&=\sum_{\substack{u_x \oplus k_2\in J_{k_2}'\text{ for } y_x \oplus k_3 \in S':\\\forall a\neq b\text{ }u_a\neq u_b }} \frac{1}{\sqrt{\abs{J_{k_2}}\dots (\abs{J_{k_2}}-s+1)}} \bigotimes_{y_x \oplus k_3 \in S'} \ket{u_x \oplus k_2}_{D^1_{y_x \oplus k_3}}\ket{u_x\oplus}_{D^2_{x}}\ket{x\oplus k_1}_{D^3_{y_x}}.
    \end{align} We have shown that $[\dm,\fp_{D^2}\vf\fp_{D^2}]\ket{\psi}$ vanishes for any basis state $\ket{\psi}$ in $\Pi^g \Pi^j \Pi^t$. Hence, the commutator is 0 on this subspace.
    \item We bound the second term in \eqref{eq:v-ssf-split}. Again consider $\ket{\psi}$ in \eqref{eq:some-psi}, since it is a good database state, we know that there is no $u\in\im{D^1}$ s.t. $u\oplus k_2 \in \dom{D^2}$. This means any $\vr\ppar{\cdot}$ acting nontrivially on the state $\fp_{D^2}\vf\fp_{D^2}\ket{\psi}$ does uncomputation on $D^3$, and does not introduce anything new to $D^3$. Hence, afterwards $\dm$ acts as intended (i.e. flip does not fail). Moreover, since $\dm$ preserves the difference between the elements of $\im{D^1}$ and $\dom{D^2}$ (while flipping and swapping the both), the same elements on $D^3$ are uncomputed on the state $\dm\fp_{D^2}\vf\fp_{D^2}\ket{\psi}$ as well. Hence, the commutator becomes 0.
\end{enumerate}
\end{proof}

\subsection{Proof of \Cref{thm:inv-to-fwd-hybrid}}
\begin{proof}

\begin{align}
  &(\dm  \cco^{R,fwd} \dm V - V \Pi^{g}  \Pi^{inj}  \dm \cco^{I,fwd}  \dm) \Pi^{g}  \Pi^{inj}  \Pi^t \Pi^{k_1,k_3}\\
  &=  (\dm  \cco^{R,fwd} \dm V - V  \dm \Pi^{g}  \Pi^{inj}  \cco^{I,fwd}  \dm) \Pi^{g}  \Pi^{inj}  \Pi^t \Pi^{k_1,k_3}  \quad \text{(By \Cref{lem:comm-dm-and-proj})}\\
  & =  (\dm  \cco^{R,fwd} V \dm  - \dm V \Pi^{g}  \Pi^{inj}  \cco^{I,fwd}  \dm) \Pi^{g}  \Pi^{inj}  \Pi^t \Pi^{k_1,k_3} \quad \text{(By \Cref{lem:v-dm-comm})}\label{eq:t+1-hidden}\\
   &=  \dm ( \cco^{R,fwd} V   -  V \Pi^{g}  \Pi^{inj}  \cco^{I,fwd}  )   \Pi^{g}  \Pi^{inj}  \Pi^t \Pi^{k_1,k_3}\dm \\
  &=  \dm ( \cco^{R,fwd} V   -  V \Pi^{g}  \Pi^{inj}  \cco^{I,fwd}  )  \Pi^{g}  \Pi^{inj}  \Pi^t \Pi^{k_1,k_3}\dm
  \quad \text{(By \Cref{lem:comm-dm-and-proj})} \label{eq:no-norm}
  \end{align} where in \eqref{eq:t+1-hidden} we implicitly used the fact that single query increases the database size by at most 1, i.e., $\cco^{I,fwd} \Pi^{g}  \Pi^{inj}  \Pi^t=\Pi^{t+1}\cco^{I,fwd} \Pi^{g}  \Pi^{inj}  \Pi^t$ by \Cref{lem:pi-t-violation}.
Hence, using the unitary invariance and submultiplicativity, the norm of the term in \eqref{eq:no-norm} is equal to \begin{align}
     \lVert( \cco^{R,fwd} V   -  V \Pi^{g}  \Pi^{inj}  \cco^{I,fwd}  )  \Pi^{g}  \Pi^{inj}  \Pi^t \Pi^{k_1,k_3} \rVert.
\end{align}
 \end{proof}
\section{$P_1$-query: Deferred Proofs}\label{app:p1}

 \subsection{Proof of \Cref{lem:p1-psi-y-d}}
 \begin{proof} 
\begin{enumerate}[leftmargin=0pt, labelindent=0pt, itemindent=*, label=(\arabic*)]
    \item Observe that $D^3_{x\oplus k_1}=\bot$ means $\vf\ppar{x}$ is not active, and since $D$ with $k_2$ is a good database by fulfilling \Cref{item:rem:fill-id-restricted-D1} of \Cref{rem:fill-or-identity}, we conclude that every factor of $\vf$ performs the fill
$\ket{\bot}\ket{\bot} \mapsto \ket{\psi_{k_2}}$ or acts as the identity. As a result, the image of the computational basis under $\fp_{D^2}\vf\fp_{D^2}$ has only positive coefficient terms. Combining this with the fact that an isometry preserves orthogonality of the inputs and that $\fp_{D^2}\vf\fp_{D^2}$ preserves $D^1_x$, we see that $\ket{\psi^y_D}\perp \ket{\psi^y_{D'}}$ for $\ket{D}\perp\ket{D'}$, and $\ket{\psi^y_D}$ and $\ket{\psi^y_{D'}}$ have pairwise orthogonal terms.

In addition, observing that $\ket{\psi^y_D}$ for all $y\in[N]\cup\{\bot\}$ consists of the terms in $\ket{\psi^\bot_D}$ and $\ket{\psi^\bot_D}\perp \ket{\psi^\bot_{D'}} $, we can conclude $$( \ket{\psi^y_D}- \ket{\psi^{y'}_D})\perp (\ket{\psi^y_{D'}}- \ket{\psi^{y'}_{D'}})$$ for $\ket{D}\perp\ket{D'}$ and $y,y'\in[N]\cup\set{\bot}$.

\item Note that by triangle inequality and the symmetry of the ensemble $(\ket{\psi^y_D})_{y\in[N]}$, we can write\begin{align}
    \lVert \ket{\psi^y_D} - \ket{\psi^{y'}_D} \rVert \leq 2\lVert \ket{\psi^y_D} - \ket{\psi^{\bot}_D} \rVert.
\end{align} For the state in \eqref{eq:lem-p1-psi}, we fix the set in \Cref{def:set-S}, and call it $S$ while noting that $S$ is independent of what $D^1_x$ holds since already $x\not\in S$. For the same state with $y=\bot$, we fix the set in \Cref{def:J} and call it $J$. Then, we can write $\ket{\psi^\bot_D}$, omitting the irrelevant and untouched registers, as \begin{align}
\ket{\psi^\bot_D}=
    \sum_{\substack{j\in \mathsf{Tup}(J,S):\\\forall a\neq b\text{ }j_a\neq j_b }} \frac{1}{\sqrt{\abs{J}\dots (\abs{J}-\abs{S}+1)}} \bigotimes_{i\in S}\ket{j_i}_{D^1_i}\ket{y_i\oplus k_3}_{D^2_{j_i\oplus k_2}}\ket{y_i}_{D^3_{i\oplus k_1}}.
\end{align} Similarly, for the state in \eqref{eq:lem-p1-psi} with $y\in[N]$, we fix the set in \Cref{def:J} and call it $J^y$, then do case analysis. If $y\not\in J$ (which happens when $y\oplus k_2\in\dom{D^2}$), then $J^y=J$ and $\ket{\psi^\bot_D}=\ket{\psi^y_D}$. Consequently,\begin{align}
    \lVert \ket{\psi^y_D} - \ket{\psi^{\bot}_D} \rVert=0.
\end{align} If $y\in J$, then $J^y=J\setminus \{ y\}$ and \begin{align}
 \ket{\psi^y_D}=
    \sum_{\substack{j\in \mathsf{Tup}(J^y,S):\\\forall a\neq b\text{ }j_a\neq j_b }} \frac{1}{\sqrt{\abs{J^y}\dots (\abs{J^y}-\abs{S}+1)}} \bigotimes_{i\in S}\ket{j_i}_{D^1_i}\ket{y_i\oplus k_3}_{D^2_{j_i\oplus k_2}}\ket{y_i}_{D^3_{i\oplus k_1}}.   
\end{align} The states $\ket{\psi^\bot_D}$ and $\ket{\psi^y_D}$ differ in two ways. First,
$\ket{\psi^\bot_D}$ contains terms with $y \in \im{D^1}$, which are absent
from $\ket{\psi^y_D}$, since filling $D^1_x$ with $y$ removes $y$ from the available set. Second, the terms common to both appear with different coefficients, as the normalization changes from $\abs{J}$ to $\abs{J}-1$. Hence, we write \begin{align}
    &\lVert \ket{\psi^y_D} - \ket{\psi^{\bot}_D} \rVert \\
    &\leq \norm{\sum_{\substack{j\in \mathsf{Tup}(J,S):\\\forall a\neq b\text{ }j_a\neq j_b, \exists i\, j_i=y }} \frac{1}{\sqrt{\abs{J}\dots (\abs{J}-\abs{S}+1)}} \bigotimes_{i\in S}\ket{j_i}_{D^1_i}\ket{y_i\oplus k_3}_{D^2_{j_i\oplus k_2}}\ket{y_i}_{D^3_{i\oplus k_1}}}\nonumber\\
    & \quad + \norm{ \sum_{\substack{j\in \mathsf{Tup}(J^y,S):\\\forall a\neq b\text{ }j_a\neq j_b }} \Bigg(\frac{1}{\sqrt{\abs{J^y}\dots (\abs{J^y}-\abs{S}+1)}} -\frac{1}{\sqrt{\abs{J}\dots (\abs{J}-\abs{S}+1)}}\Biggl)\bigotimes_{i\in S}\ket{j_i}_{D^1_i}\ket{y_i\oplus k_3}_{D^2_{j_i\oplus k_2}}\ket{y_i}_{D^3_{i\oplus k_1}} }\\
    &= \sqrt{\frac{(\abs{J}-1)\dots (\abs{J}-\abs{S}+1)\abs{S}}{\abs{J}\dots (\abs{J}-\abs{S}+1)}} \nonumber\\
    &\quad+ \sqrt{\Bigg(\frac{1}{\sqrt{\abs{J^y}\dots (\abs{J^y}-\abs{S}+1)}} -\frac{1}{\sqrt{\abs{J}\dots (\abs{J}-\abs{S}+1)}}\Biggl)^2 \abs{J^y}\dots (\abs{J^y}-\abs{S}+1)}\\
    &\leq O(\sqrt{t/(N-2t)}) + O(\sqrt{t/(N-3t)})\label{eq:boring-counting}
\end{align} using $\abs{S}\leq t$ and $\abs{J}\geq N-2t$.

\end{enumerate}
 
\end{proof}
\section{$P_2$-query Bound Deferred Proofs}\label{app:p2-query}

Using \Cref{thm:inv-to-fwd-hybrid}, it is enough to consider the difference between the hybrids for forward queries, i.e., consider the real and ideal world oracle queries for $b=0$ and $i=2$. 
\begin{lemma}\label{lem:p2-psi-y-d}
    Let $\ket{\psi}$ be a database state in the computational basis
\begin{align}
  \ket{\psi} \coloneqq \ket{y}_{D^2_x}\ket{D^1}_{D^1}\ket{D}_{(D^2_x)^c}\ket{D^3}_{D^3}
\ket{k_1}_{K_1}\ket{k_2}_{K_2}\ket{k_3}_{K_3},\label{eq:p2-psi-with-y}
\end{align} 
where $y \in [N]\setminus \im{D}\cup(\im{D^3}\oplus k_3), x\not\in\im{D^1}\oplus k_2$, and $D^1,D$ with $k_2$ is a good database, and each database $D^1,D,D^3$ is of size at most $t$ and is injective. Since $\fp_{D^2}\vf\fp_{D^2}$ acts as the identity
on the remaining registers, we may define $\ket{\psi^y_{D^1,D}}$ at register $D^1(D^2_x)^c$ by
\begin{align}
  \fp_{D^2}\vf\fp_{D^2}\ket{\psi} = \ket{y}_{D^2_x}\ket{\psi^y_{D^1,D}}\ket{D^3}_{D^3}\ket{k_1}_{K_1}\ket{k_2}_{K_2}\ket{k_3}_{K_3}. 
\end{align}Similarly, when the register $D^2_x$ is set to $\bot$, i.e.,
\begin{align}
  \ket{\psi} \coloneqq \ket{\bot}_{D^2_x}\ket{D^1}_{D^1}\ket{D}_{(D^2_x)^c}\ket{D^3}_{D^3}
\ket{k_1}_{K_1}\ket{k_2}_{K_2}\ket{k_3}_{K_3}, \label{eq:p2-psi-with-bot}
\end{align}we define $\ket{\psi^\bot_{D^1,D}}$ and $\ket{\not\bot}_{D^2_x}\ket{\psi^{\not\bot}_{D^1,D}}$ by \begin{align}
  \fp_{D^2}\vf\fp_{D^2}\ket{\psi} = (\ket{\bot}_{D^2_x}\ket{\psi^\bot_{D^1,D}}+ \ket{\not\bot}_{D^2_x}\ket{\psi^{\not\bot}_{D^1,D}}
)\otimes\ket{D^3}_{D^3}\ket{k_1}_{K_1}\ket{k_2}_{K_2}\ket{k_3}_{K_3} \label{eq:p2-psi-bot-d}
\end{align} where $\ket{\not\bot}_{D^2_x}\ket{\psi^{\not\bot}_D}_{D}$ is all the remaining vectors where $D^2_x$ does not hold $\bot$. Then \begin{enumerate}[label=(\arabic*)]
        \item \label{item:p2-y-y'-orthogonal}For databases such that $\ket{D^1}_{D^1}\ket{D}_{(D^2_x)^c}\perp\ket{B^1}_{D^1}\ket{B'}_{(D^2_x)^c}$, we have \begin{align}
        \ket{\psi^y_{D^1,D}}\perp \ket{\psi^y_{B^1,B'}} \text{ and }      \ket{\not\bot}_{D^2_x}\ket{\psi^{\not\bot}_{D^1,D}}\perp \ket{\not\bot}_{D^2_x}\ket{\psi^{\not\bot}_{B^1,B}},
        \end{align}
        as a consequence of the former one, we get 
        \begin{align}
            ( \ket{\psi^y_{D^1,D}}- \ket{\psi^{y'}_{D^1,D}})\perp (\ket{\psi^y_{B^1,B}}- \ket{\psi^{y'}_{B^1,B}}),
        \end{align}  
        \item We get $\lVert \ket{\psi^y_{D^1,D}} - \ket{\psi^{y'}_{D^1,D}} \rVert \leq O(\sqrt{t/N})$,\label{item:p2-y-y'-bound}
        \item We get $\lVert\ket{\not\bot}_{D^2_x}\ket{\psi^{\not\bot}_{D^1,D}}\rVert\leq O(\sqrt{t/N})$\label{item:p2-non-bots}
    \end{enumerate}for all $y, y' \in [N] \cup \set{\bot}$ and $D^1,D,B^1, B\in\mathbf{D}_t$ for which the relative terms are defined. For the bounds, we use the assumption $3q < N/2$ which is trivially satisfied by the regime for $q$.
\end{lemma}

\begin{proof}
\begin{enumerate}[leftmargin=0pt, labelindent=0pt, itemindent=*, label=(\arabic*)]
\item Observe that $D^1,D$ with $k_2$ is a good database and $x\not\in\im{D^1}\oplus k_2$ implies that if we extend the database $D$ at $(D^2_x)^c$ to include the register $D^2_x$, then the full databases at registers $D^1$ and $D^2$ with $k_2$ still form a good database. Hence, independent from what $D^2_x$ holds, by \Cref{rem:fill-or-identity}, we conclude that every factor of $\vf$ performs the fill
$\ket{\bot}\ket{\bot} \mapsto \ket{\psi_{k_2}}$ or acts as the identity. As a result, on such computational basis states, the image of $\fp_{D^2}\vf\fp_{D^2}$ has only positive coefficient terms. Combining this with the fact that $\fp_{D^2}\vf\fp_{D^2}$ is an isometry, we see that terms in its images are pairwise orthogonal on orthogonal inputs. Since $\ket{\psi^y_{D^1,D}}$ and $\ket{\psi^y_{B^1,B}}$ for $\ket{D^1}_{D^1}\ket{D}_D\perp\ket{B^1}_{D^1}\ket{B}_{(D^2_x)^c}$ are derived from images of $\fp_{D^2}\vf\fp_{D^2}$ where $D^2_x$ register is preserved, it is immediately implied that \begin{align}
    \ket{\psi^y_{D^1,D}}\perp \ket{\psi^y_{B^1,B'}},
\end{align} moreover, for the terms in the image where $D^2_x$ state is not preserved, we can write \begin{align}
\ket{\not\bot}_{D^2_x}\ket{\psi^{\not\bot}_{D^1,D}}\perp \ket{\not\bot}_{D^2_x}\ket{\psi^{\not\bot}_{B^1,B}}. 
\end{align}

In addition, observing that $\ket{\psi^y_{D^1,D}}$ for all $y\in[N]\cup\{\bot\}$ consists of the terms in $\ket{\psi^\bot_{D^1,D}}$, we can conclude $$( \ket{\psi^y_{D^1,D}}- \ket{\psi^{y'}_{D^1,D}})\perp (\ket{\psi^y_{B^1,B}}- \ket{\psi^{y'}_{B^1,B}})$$ for all $y, y' \in [N] \cup \set{\bot}$ and and $D^1,D,B^1, B\in\mathbf{D}_t$ for which the relative terms are defined.

\item Note that by triangle inequality, we can write 
\begin{align}
    \lVert \ket{\psi^y_{D^1,D}}- \ket{\psi^{y'}_{D^1,D}} \rVert \leq 2\max\set{\lVert \ket{\psi^y_{D^1,D}}- \ket{\psi^{\bot}_{D^1,D}} \rVert, \lVert \ket{\psi^{y'}_{D^1,D}}- \ket{\psi^{\bot}_{D^1,D}} \rVert},
\end{align}without loss of generality, we compute the norm for $y$, independent from $y$. For the state in \eqref{eq:p2-psi-with-bot}, we fix the sets in \Cref{def:set-S} and in \Cref{def:J}, and call it $S$ and $J$, respectively. Note that the set $J^\bot$ contains $x\oplus k_2$ since in the input state to $\fp_{D^2}\vf\fp_{D^2}$ both $D^2_x=\bot$ and $x\not\in\im{D^1}\oplus k_2$. However, the terms with $x\oplus k_2 \in\im{D^1}$ under $\fp_{D^2}\vf\fp_{D^2}$ does not appear in $\ket{\psi^\bot_{D^1,D}}$ since it consists of the terms whose state at $D^2_x$ is preserved as $\bot$. With this observation, we can write $\ket{\psi^\bot_{D^1,D}}$, omitting the irrelevant and untouched registers, as \begin{align}
\ket{\psi^\bot_{D^1,D}}=
    \sum_{\substack{j\in \mathsf{Tup}( J\setminus \set{x\oplus k_2} ,S):\\\forall a\neq b\text{ }j_a\neq j_b }} \frac{1}{\sqrt{\abs{J}\dots (\abs{J}-\abs{S}+1)}} \bigotimes_{i\in S}\ket{j_i}_{D^1_i}\ket{y_i\oplus k_3}_{D^2_{j_i\oplus k_2}}\ket{y_i}_{D^3_{i\oplus k_1}}.
\end{align} Similarly, for the state in \eqref{eq:p2-psi-with-y} with $y\in[N]$, we fix the set in \Cref{def:set-S} and in \Cref{def:J}, and call it $S^y$ and $J^y$, respectively. We note that the conditions on \eqref{eq:p2-psi-with-y} prevents $y$ being in $\im{D^3}\oplus k_3$. For this reason, the set $S^y=S$, otherwise $S^y$ would depend on what other related register holds. Furthermore, the set $J^y$ does not contain $x\oplus k_2$ since $D^2_x\neq\bot$. Hence, $J^y=J\setminus \set{x\oplus k_2}$. With these observations we can write $\ket{\psi^y_{D^1,D}}$, omitting the irrelevant and untouched registers, as \begin{align}
\ket{\psi^y_{D^1,D}}=
    \sum_{\substack{j\in \mathsf{Tup}(J^y,S):\\\forall a\neq b\text{ }j_a\neq j_b }} \frac{1}{\sqrt{\abs{J^y}\dots (\abs{J^y}-\abs{S}+1)}} \bigotimes_{i\in S}\ket{j_i}_{D^1_i}\ket{y_i\oplus k_3}_{D^2_{j_i\oplus k_2}}\ket{y_i}_{D^3_{i\oplus k_1}}.
\end{align} Hence, we see that the terms in $\ket{\psi^y_{D^1,D}}$ and $\ket{\psi^\bot_{D^1,D}}$ only differ in coefficients, we can bound the difference as \begin{align}
&\norm{\ket{\psi^y_{D^1,D}}-\ket{\psi^\bot_{D^1,D}}}\\
   &=\norm{ \sum_{\substack{j\in \mathsf{Tup}(J^y,S):\\\forall a\neq b\text{ }j_a\neq j_b }} \Bigg(\frac{1}{\sqrt{\abs{J^y}\dots (\abs{J^y}-\abs{S}+1)}} -\frac{1}{\sqrt{\abs{J}\dots (\abs{J}-\abs{S}+1)}}\Biggl)\bigotimes_{i\in S}\ket{j_i}_{D^1_i}\ket{y_i\oplus k_3}_{D^2_{j_i\oplus k_2}}\ket{y_i}_{D^3_{i\oplus k_1}} } \\
&\leq \sqrt{\Bigg(\frac{1}{\sqrt{\abs{J^y}\dots (\abs{J^y}-\abs{S}+1)}} -\frac{1}{\sqrt{\abs{J}\dots (\abs{J}-\abs{S}+1)}}\Biggl)^2 (\abs{J}-1)\dots (\abs{J}-\abs{S})}\\
    &\leq O(\sqrt{t/(N-3t)})
\end{align} using $\abs{S}\leq t$ and $\abs{J}\geq N-2t$.

\item By the analysis in the previous case, the terms whose $D^2_x$ is not preserved as $\bot$ under $\fp_{D^2}\vf\fp_{D^2}$ are exactly those for which $\vf$ fills $D^1_i$ with $x \oplus k_2$ for some $i \in S$, that is,  \begin{align}
    \ket{\not\bot}_{D^2_x}\ket{\psi^{\not\bot}_{D^1,D}}=\sum_{\substack{j\in \mathsf{Tup}(J,S) :\\\forall a\neq b\text{ }j_a\neq j_b, \exists i\, j_i=x\oplus k_2 }} \frac{1}{\sqrt{\abs{J}\dots (\abs{J}-\abs{S}+1)}} \bigotimes_{i\in S}\ket{j_i}_{D^1_i}\ket{y_i\oplus k_3}_{D^2_{j_i\oplus k_2}}\ket{y_i}_{D^3_{i\oplus k_1}},
\end{align}
since this is a sum of orthogonal terms, where $x \oplus k_2$ is fixed at each $D^1_i$ for $i \in S$ and the remaining indices in $S \setminus \{i\}$ range over all pairwise distinct values in $J$, its norm is 
\begin{align}
    \sqrt{\frac{(\abs{J}-1)\dots \abs{(\abs{J}-\abs{S}+1)}\abs{S}}{\abs{J}\dots (\abs{J}-\abs{S}+1)}} \leq O(\sqrt{t/(N-2t)})
\end{align}using $\abs{S}\leq t$ and $\abs{J}\geq N-2t$.

\end{enumerate}
 
\end{proof}

\subsection{Proof of \Cref{thm:p2-bound}}
\begin{proof}
Note that oracle query to $P_2$ is handled identically in the real and ideal worlds, i.e., $\cco^{R} = \cco^{I}$. Accordingly, we simply write $\cco$ for this operator. We use that the uniformity of $k_3$ and almost uniformity of $k_2$ in $\ket{\psi^{t}_{I,\Pi^\star }}$,\begin{align}
     & \lVert  ( \cco V   -  V \Pi^\star  \cco  )  \Pi^\star  \Pi^t \Pi^{k_1,k_3}\rVert \\
     &\quad \leq  \lVert  ( \cco V   -  V \Pi^\star  \cco  )  \Gamma_1 \Pi^\star  \Pi^t \Pi^{k_1,k_3}\rVert + O(\sqrt{t/N}) \quad \text{(By \Cref{lem:gamma1-perp-bound})}\\
     &\quad \leq  \lVert  ( \cco V   -  V  \cco  )  \Gamma_1 \Pi^\star  \Pi^t \Pi^{k_1,k_3} \rVert  + O(\sqrt{t/N}) \label{eq:some-term-p2}
\end{align} where in the last line we use the identity $\Pi^\star=\Id - (\Pi^\star)^\perp$ with triangle inequality and \Cref{lem:pi-star-violation}. For the remainder of the proof, it suffices to bound the first term of \Cref{eq:some-term-p2} on states in the image of $\Gamma_1 \Pi^\star  \Pi^t$. 
We continue to simplify by decomposing $\cco$ as \begin{align}
    &\lVert  ( \cco V   -  V  \cco  )  \Gamma_1 \Pi^\star  \Pi^t \Pi^{k_1,k_3}\rVert \\
    & \leq \norm{[V,\fc]\Gamma_1 \Pi^\star  \Pi^t} + \norm{[V,\pu] \fc\Gamma_1  \Pi^\star  \Pi^t \Pi^{k_1,k_3}} + \norm{[V,\fc] \pu \fc \Gamma_1  \Pi^\star  \Pi^t\Pi^{k_1,k_3}}
\end{align} by the submultiplicativity of the norm on the first and second terms, then we continue by simplifying the second and third term as 
 \begin{align}
   &\norm{[V,\pu] \fc \Gamma_1  \Pi^\star  \Pi^t\Pi^{k_1,k_3}} \\
   &= \norm{[V,\pu] \Gamma_1 \Pi^\star \Pi^{t+1} \fc \Gamma_1 \Pi^\star  \Pi^t\Pi^{k_1,k_3}} + O(\sqrt{t/N})\quad\text{(By \Cref{lem:pi-t-violation}, \ref{lem:pi-star-violation}, \ref{lem:gamma1-perp-bound})}\\
   &\leq\norm{[V,\pu] \Gamma_1 \Pi^\star \Pi^{t+1}  } + O(\sqrt{t/N})  \quad\text{(By submultiplicativity of the norm)}
\end{align}and similarly, additionally using the commutativity of all the projectors with $\pu$, we get 
\begin{align}
    \norm{[V,\fc] \pu \fc \Gamma_1  \Pi^\star  \Pi^{t}} \leq \norm{[V,\fc] \Gamma_1  \Pi^\star \Pi^{t+1}}+ O(\sqrt{t/N})  .
\end{align} Since big-$O$ is indifferent to constants, we treat $t+1$ and $t$ the same, and we bound the remaining terms by decomposing $V$ as \begin{align}
   \norm{[V,\pu]\Gamma_1 \Pi^\star  \Pi^t}&\leq \norm{[\fp_{D^2}\vf\fp_{D^2},\pu]\Gamma_1 \Pi^\star  \Pi^t} + \norm{[\vr,\pu]\fp_{D^2}\vf\fp_{D^2}\Gamma_1 \Pi^\star  \Pi^t} \\
   &=\norm{[\fp_{D^2}\vf\fp_{D^2},\pu]\Gamma_1 \Pi^\star  \Pi^t} \label{eq:p2-comm-pu}
\end{align}
since both $\vr$ and $\pu$ are controlled on $D^2_x$ and otherwise act nontrivially on distinct registers, and
\begin{align}
          \norm{[V,\fc]\Gamma_1 \Pi^\star  \Pi^t}\leq \norm{[\fp_{D^2}\vf\fp_{D^2},\fc]\Gamma_1 \Pi^\star  \Pi^t} + \norm{[\vr,\fc]\fp_{D^2}\vf\fp_{D^2}\Gamma_1 \Pi^\star  \Pi^t} \label{eq:p2-comm-vf-fc}
      \end{align}

In the rest, we fix the registers $X,K_2$ in the computational basis $x,k_2$, respectively, such that they are allowed by $\Gamma_1$, i.e., $x\not\in \im{D^1}\oplus k_2$.

\begin{enumerate}[leftmargin=0pt, labelindent=0pt, itemindent=*, label=(\arabic*)]
    \item Consider $D^2_x=\bot$.  

    \begin{enumerate}
        \item We first bound $\norm{[\fp_{D^2}\vf\fp_{D^2},\pu]\Gamma_1 \Pi^\star  \Pi^t}$. Observe that on good databases $\fp_{D^2}\vf\fp_{D^2}$ produces a superposition of database states whose terms are pairwise orthogonal across all databases (see \Cref{rem:fill-or-identity}), and $\pu$ preserves computational basis states. So, independent from the order of the two operators, the images of orthogonal databases stay orthogonal, hence we can fix a general database state in the computational basis as 
            \begin{align}
                  \ket{\psi}\coloneqq  \ket{D^1}_{D^1}\ket{\bot}_{D^2_x} \ket{(D^2_x)^c}_{(D^2_x)^c} \ket{D^3}_{D^3} \ket{k_1}_{K_1} \ket{k_2}_{K_2}\ket{k_3}_{K_3}
              \end{align} where $D^1,(D^2_x)^c$ with $k_2$ is good, and each database is of size at most $t$ and injective, and $x\not \in \im{D^1} \oplus k_2$ (due to $\Gamma_1$).

    Since $\pu$ acts as identity when $D^2_x=\bot$, the source of difference is exactly the terms where $D^2_x$ is not in $\bot$ after the operator $\fp_{D^2}\vf\fp_{D^2}$ applied. This is exactly the term bounded by $O(\sqrt{t/N})$ in \Cref{item:p2-non-bots} of \Cref{lem:p2-psi-y-d}.
  
        \item Next, we bound $\norm{[\fp_{D^2}\vf\fp_{D^2},\fc]\Gamma_1 \Pi^\star  \Pi^t}$. Observe that $\fc$ acts only on $D^2_x$ and $\fp_{D^2}\vf\fp_{D^2}$ preserves $D^2_x$ in the computational basis and mostly preserves in $\bot$, in either way its action on $D^1$ and $(D^2_x)^c$ depends on what $D^2_x$ holds in order to preserve injectivity of the databases. Consequently, orthogonality of all the registers other $D^1$ and $(D^2_x)^c$ is preserved by the commutator, and we fix all other registers in the computational basis and write a general state \begin{align}
            \ket{\psi} \coloneqq \sum_{(D^1, D) \in \mathbf{D}_{\mathrm{adm}}}
\gamma_{D^1,D}\ket{\bot}_{D^2_x}\ket{D^1}_{D^1}\ket{D}_{(D^2_x)^c}\ket{D^3}_{D^3}\ket{k_1}_{K_1}\ket{k_2}_{K_2}\ket{k_3}_{K_3}
        \end{align} where $D^3\in\mathbf{D}_{inj,t}$ and
        \begin{align}
            \mathbf{D}_{\mathrm{adm}} \coloneqq \set{ (D^1, D) \in
  \mathbf{D}_{t,inj} \times \mathbf{D}_{t,inj} :
  D^1, D \text{ with } k_2 \text{ is good}, \; x \notin \im{D^1} \oplus k_2 }.
        \end{align} The reason for splitting the databases in $D^1$ and $(D^2_x)^c$ registers in comparison to $P_1$-query case is that flip operator on $D^2$ fails if action of $\fc$ on $D^2_x$ violates the injectivity of the databases. We project out the vectors in both terms of the commutator where injectivity of $D^2$ is violated, and bound the difference of the remaining.
       
Then we write the terms of the commutator, while omitting the registers that are untouched,\begin{align}
    &\fc \fp_{D^2}\vf\fp_{D^2} \ket{\psi} \\
    &= \fc \Bigg(\sum_{(D^1, D) \in \mathbf{D}_{\mathrm{adm}}}\gamma_{D^1,D}(\ket{\bot}_{D^2_x}\ket{\psi^\bot_{D^1,D}}+  \ket{\not\bot}_{D^2_x}\ket{\psi^{\not\bot}_{D^1,D}})\Biggl)\\
    &= \frac{1}{\sqrt{N}}\sum_{y\in[N],(D^1, D) \in \mathbf{D}_{\mathrm{adm}}}\gamma_{D^1,D}\ket{y}_{D^2_x}\ket{\psi^\bot_{D^1,D}}+ \fc \Bigg(\sum_{(D^1, D) \in \mathbf{D}_{\mathrm{adm}}}\gamma_{D^1,D}\ket{\not\bot}_{D^2_x}\ket{\psi^{\not\bot}_D}\Biggl)
\end{align} 
and 
\begin{align}
 \fp_{D^2}\vf\fp_{D^2} \fc\ket{\psi} &=  \frac{1}{\sqrt{N}}\sum_{y\in[N],(D^1, D) \in \mathbf{D}_{\mathrm{adm}}}\gamma_{D^1,D}\ket{y}_{D^2_x} \ket{D^1}_{D^1}\ket{D}_{(D^2_x)^c} \\
 &=\frac{1}{\sqrt{N}}\sum_{\substack{y\in[N],(D^1, D) \in \mathbf{D}_{\mathrm{adm}}:\\y\not\in\im{D}}}\gamma_{D^1,D}\ket{y}_{D^2_x} \ket{\psi^y_{D^1,D}}\\
 &\quad+\fp_{D^2}\vf\fp_{D^2}\Bigg(\frac{1}{\sqrt{N}}\sum_{\substack{y\in[N],(D^1, D) \in \mathbf{D}_{\mathrm{adm}}:\\y\in\im{D}\cup \im{D^3}\oplus k_3}}\gamma_{D^1,D}\ket{y}_{D^2_x} \ket{D^1}_{D^1}\ket{D}_{(D^2_x)^c}\Bigg)
\end{align} 
where $\ket{\psi^\bot_{D^1,D}},\ket{\psi^y_{D^1,D}},\ket{\not\bot}_{D^2_x}\ket{\psi^{\not\bot}_{D^1,D}}$ are as defined in \Cref{lem:p2-psi-y-d}. Hence, by projecting out the terms that value at $D^2_x$ and $D$ is not injective, and $\vf$ does not preserve the state at $D^2_x$, we write the difference as \begin{align}
&\norm{[\fc,\fp_{D^2}\vf\fp_{D^2}]\ket{\psi}}\\
&\leq\norm{\frac{1}{\sqrt{N}}\sum_{\substack{y\in[N],(D^1, D) \in \mathbf{D}_{\mathrm{adm}}:\\y\not\in\im{D}}}\gamma_{D^1,D}\ket{y}_{D^2_x}\otimes\Big( \ket{\psi^\bot_{D^1,D}}-\ket{\psi^y_{D^1,D}}\Big)}\nonumber\\
&\quad+ \norm{\fc \Bigg(\sum_{(D^1, D) \in \mathbf{D}_{\mathrm{adm}}}\gamma_{D^1,D}\ket{\not\bot}_{D^2_x}\ket{\psi^{\not\bot}_{D^1,D}}\Biggl)} \nonumber\\
&\quad+ \norm{\fp_{D^2}\vf\fp_{D^2}\Bigg(\frac{1}{\sqrt{N}}\sum_{\substack{y\in[N],(D^1, D) \in \mathbf{D}_{\mathrm{adm}}:\\y\in\im{D}\cup (\im{D^3}\oplus k_3)}}\gamma_{D^1,D}\ket{y}_{D^2_x} \ket{D^1}_{D^1}\ket{D}_{(D^2_x)^c}\Bigg)}\nonumber \\
&\quad+ \norm{\frac{1}{\sqrt{N}}\sum_{\substack{y\in[N],(D^1, D) \in \mathbf{D}_{\mathrm{adm}}:\\y\in\im{D}\cup (\im{D^3}\oplus k_3)}}\gamma_{D^1,D}\ket{y}_{D^2_x} \ket{\psi^\bot_{D^1,D}}} \\
&\leq \sqrt{\frac{1}{N}\sum_{y\in[N]}\sum_{(D^1, D) \in \mathbf{D}_{\mathrm{adm}}:y\not\in\im{D}}\abs{\gamma_{D^1,D}}^2\lVert \ket{\psi^\bot_{D^1,D}}-\ket{\psi^y_{D^1,D}} \rVert^2} \nonumber\\
&\qquad \text{(By the orthogonality of $\ket{y}$'s and $\ket{\psi^\bot_{D^1,D}}-\ket{\psi^y_{D^1,D}}$'s as given in \Cref{lem:p2-psi-y-d})}\nonumber\\
&\quad +\sqrt{\sum_{(D^1, D) \in \mathbf{D}_{\mathrm{adm}}}\abs{\gamma_{D^1,D}}^2 \lVert\ket{\not\bot}_{D^2_x}\ket{\psi^{\not\bot}_{D^1,D}} \rVert^2} \nonumber\\
&\qquad \text{(By the orthogonality of $\ket{\not\bot}_{D^2_x}\ket{\psi^{\not\bot}_{D^1,D}}$'s as given in \Cref{lem:p2-psi-y-d})}\nonumber\\
& \quad+ \sqrt{\frac{1}{N}\sum_{y\in[N],(D^1, D)\in\mathbf{D}_{\mathrm{adm}}:y\in\im{D}\cup (\im{D^3}\oplus k_3) }\abs{\gamma_{D^1,D}}^2 }\nonumber\\
 &\qquad \text{(By the orthogonality of computational basis states)}\nonumber\\
&\quad + \sqrt{\frac{1}{N}
  \sum_{y\in[N],(D^1,D)\in\mathbf{D}_{\mathrm{adm}}:y\in\im{D}\cup (\im{D^3}\oplus k_3)}
  \abs{\gamma_{D^1,D}}^2 \lVert\ket{\psi^\bot_{D^1,D}}\rVert^2}\\
  &\qquad \text{(By the orthogonality of $\ket{\not\bot}_{D^2_x}\ket{\psi^{\not\bot}_{D^1,D}}$'s as given in \Cref{lem:p2-psi-y-d})}\nonumber\\
 &\leq \sqrt{\frac{1}{N}\sum_{y\in[N]}\sum_{(D^1, D) \in \mathbf{D}_{\mathrm{adm}}:y\not\in\im{D}}\abs{\gamma_{D^1,D}}^2\frac{t}{N}}  \quad\text{(By \Cref{item:p2-y-y'-bound} of \Cref{lem:p2-psi-y-d})}\nonumber\\
&\quad +\sqrt{\sum_{(D^1, D) \in \mathbf{D}_{\mathrm{adm}}}\abs{\gamma_{D^1,D}}^2 \frac{t}{N} }\quad\text{(By \Cref{item:p2-non-bots} of \Cref{lem:p2-psi-y-d})}\nonumber\\
& \quad+ \sqrt{\frac{1}{N}\sum_{(D^1, D) \in \mathbf{D}_{\mathrm{adm}}}\abs{\gamma_{D^1,D}}^2 t}\quad\text{(By $\abs{\im{D}\cup (\im{D^3}\oplus k_3)}\leq t $)}\nonumber\\
&\quad + \sqrt{\frac{1}{N}
  \sum_{(D^1,D)\in\mathbf{D}_{\mathrm{adm}}}
  \abs{\gamma_{D^1,D}}^2 t} \quad \text{(By $\abs{\im{D}\cup (\im{D^3}\oplus k_3)}\leq t $, $\lVert\ket{\psi^\bot_{D^1,D}}\rVert\leq1$)}\\
  &\leq O(\sqrt{t/N}).
\end{align}

 \item  Lastly, bound $\norm{[\vr,\fc]\fp_{D^2}\vf\fp_{D^2}\Gamma_1 \Pi^\star  \Pi^t}$. We do this by observing that images of $\fp_{D^2}\vf\fp_{D^2}$ mostly preserved by $\Gamma_1$, and write \begin{align}
   &\norm{[\vr,\fc]\fp_{D^2}\vf\fp_{D^2}\Gamma_1 \Pi^\star  \Pi^t}\\
   &\leq   \norm{[\vr,\fc]\Gamma_1\fp_{D^2}\vf\fp_{D^2}\Gamma_1 \Pi^\star  \Pi^t} + \norm{[\vr,\fc](\Gamma_1)^\perp\fp_{D^2}\vf\fp_{D^2}\Gamma_1 \Pi^\star  \Pi^t}\\
   &\leq   \norm{[\vr,\fc]\Gamma_1} + 2\norm{(\Gamma_1)^\perp\fp_{D^2}\vf\fp_{D^2}\Gamma_1 \Pi^\star  \Pi^t} \label{eq:p2-1c}
 \end{align} by the submultiplicativity of the norm and the triangle inequality, then bound the leftovers separately. The first term in \eqref{eq:p2-1c} is trivially zero because by $\Gamma_1$, we have $x\not\in\im{D^1}\oplus k_2$ for any computational basis state and hence $\vr$ does not act on $D^2_x$ at all. 
 
 For the second term in \eqref{eq:p2-1c}, $\fp_{D^2}\vf\fp_{D^2}$ again
produces only nonnegative coefficients on good basis states, so terms from
distinct inputs are distinct, and it suffices to bound a fixed basis state
$\ket{\psi}$ with $\Gamma_1 \Pi^\star \Pi^t \ket{\psi} = \ket{\psi}$, where
$(D^2_x)^c$ holds $D$ and all other registers take their standard values, then
 \begin{align}
\norm{\Gamma_1^\perp\fp_{D^2}\vf\fp_{D^2}\Gamma_1 \Pi^\star  \Pi^t}&\leq \norm{\Gamma_1^\perp\fp_{D^2}\vf\fp_{D^2}\ket{\psi}}\\
&=\norm{ \Gamma_1^\perp (\ket{\bot}_{D^2_x}\ket{\psi^\bot_{D^1,D}}+\ket{\not\bot}_{D^2_x}\ket{\psi^{\not\bot}_{D^1,D}})}\\
&=\norm{\ket{\not\bot}_{D^2_x}\ket{\psi^{\not\bot}_{D^1,D}}} \quad \text{(By definitions in \Cref{lem:p2-psi-y-d})}\\
&\leq O(\sqrt{t/N}) \quad \text{(By \Cref{item:p2-non-bots} of \Cref{lem:p2-psi-y-d})}
 \end{align} where in the last equality we used the fact that when $x$ is initially not in $\im{D^1}\oplus k_2$ and $D^2_x$ is preserved by  $\fp_{D^2}\vf\fp_{D^2}$ then $x$ is still not in $\im{D^1}\oplus k_2$.

    \end{enumerate}

    \item Consider $D^2_x\neq \bot$. 
    \begin{enumerate}
        \item We first bound $=\norm{[\fp_{D^2}\vf\fp_{D^2},\pu]\Gamma_1 \Pi^\star  \Pi^t}$. Observe that when $D^2_x$ is non-$\bot$ and $x\not\in \im{D^1}\oplus k_2$ due to $\Pi^g$, hence the operator $\fp_{D^2}\vf\fp_{D^2}$ preserves the computational basis states at $D^2_x$. The other operator $\pu$ is only controlled on $D^2_x$. Hence, the commutator is 0.

        \item Next, we bound $\norm{[\fp_{D^2}\vf\fp_{D^2},\fc] \Gamma_1 \Pi^\star\Pi^t}$. This is essentially similar to \Cref{it:d1x-non-bot-vf-fc} of \Cref{thm:p1-bound}, except that we additionally project out the terms in $D^2_x$ that violates the injectivity of the database $D^2$ after (de)compression operator $\fc$. Observe that $\fc$ acts only on $D^2_x$, and $\fp_{D^2}\vf\fp_{D^2}$ preserves $D^2_x$ in the computational basis, while its action on $(D^2_x)^c$ and $D^2$ depends on what $D^2_x$ register holds in order to preserve injectivity of the databases. Since orthogonality of all the registers other than $D^1$ and $(D^2_x)^c$ is preserved by this commutator, we fix all other registers in the computational basis and write a general state\begin{align}
\ket{\psi}\coloneqq \sum_{y\in[N],(D^1,D)\in\mathbf{D}_\mathrm{y,adm}}\gamma_{y,D^1,D}\ket{y}_{D^2_x}\ket{D^1}_{D^1}\ket{D}_{(D^2_x)^c}
\ket{D^3}_{D^3}\ket{k_1}_{K_1}\ket{k_2}_{K_2}\ket{k_3}_{K_3} 
    \end{align} where $D^3\in\mathbf{D}_{inj,t}$ and\begin{align}
\mathbf{D}_\mathrm{y,adm}\coloneqq\set{(D^1,D)\in \mathbf{D}_{inj,t} \times \mathbf{D}_{inj,t-1}: (D^1_x)^c, D^2 \text{ with } k_2 \text{ is good},\, y\not\in \im{D}, \im{D^3}\oplus k_3}.
    \end{align} 
Then we write the terms of the commutator, while omitting the registers that are untouched, \begin{align}
    \fc \fp_{D^2}\vf\fp_{D^2} \ket{\psi} &= \fc\Biggl( \sum_{y\in [N], (D^1,D)\in \mathbf{D}_\mathrm{y,adm}} \gamma_{y,D^1,D}\ket{y}_{D^2_x} \ket{\psi^y_{D^1,D}}\Biggr)\\
&= \sum_{y\in [N], D\in \mathbf{D}_\mathrm{y,adm} } \gamma_{y,D^1,D}\Biggl(\ket{y}-\frac{\ket{+^n}}{\sqrt{N}}+\frac{\ket{\bot}}{\sqrt{N}}\Biggr)_{D^2_x} \ket{\psi^y_{D^1,D}}
\end{align} and \begin{align}
 &\fp_{D^2}\vf\fp_{D^2}\fc\ket{\psi} \\
 &= \fp_{D^2}\vf\fp_{D^2} \Biggl(\sum_{y\in [N], (D^1,D)\in \mathbf{D}_\mathrm{y,adm}} \gamma_{y,D^1,D}\Biggl(\ket{y}-\frac{\ket{+^n}}{\sqrt{N}}+\frac{\ket{\bot}}{\sqrt{N}}\Biggr)_{D^2_x}\ket{D^1}_{D^1}\ket{D}_{(D^2_x)^c}\Biggr)\\
 &= \sum_{y\in [N], (D^1,D)\in \mathbf{D}_\mathrm{y,adm}} \gamma_{y,D^1,D}\Biggl(\ket{y}_{D^2_x}\ket{\psi^y_{D^1,D}} -\sum_{y'\in[N]:y'\not\in\im{D}\cap(\im{D^3}\oplus k_3)}\frac{\ket{y'}_{D^2_x}}{N}\ket{\psi^{y'}_{D^1,D}}\nonumber\\
 &\qquad+ \frac{\ket{\bot}_{D^2_x}}{\sqrt{N}}\ket{\psi^{\bot}_{D^1,D}} + \frac{1}{\sqrt{N}}\ket{\not\bot}_{D^2_x}\ket{\psi^{\not\bot}_{D^1,D}}\Biggr)\nonumber\\
 &\quad + \fp_{D^2}\vf\fp_{D^2}\Biggl( \sum_{y\in [N], (D^1,D)\in \mathbf{D}_\mathrm{y,adm}} \gamma_{y,D^1,D}\sum_{y'\in[N]:y'\in\im{D}\cup(\im{D^3}\oplus k_3)}-\frac{\ket{y'}_{D^2_x}}{N}\ket{D^1}_{D^1}\ket{D}_{(D^2_x)^c}\Biggr)
\end{align} where $\ket{\psi^\bot_{D^1,D}}, \ket{\psi^y_{D^1,D}}$ and $\ket{\not\bot}_{D^2_x}\ket{\psi^{\not\bot}_{D^1,D}}$ are as defined in \Cref{lem:p2-psi-y-d}. We project out the terms such that the state at $D^2_x$ and $D$ is not injective, the state at $D^2_x$ is in $\im{D^3}\oplus k_3$, and $\vf$ does not preserve the state at $D^2_x$, then we write the difference as
 \begin{align}
    &\norm{(\fc \fp_{D^2}\vf\fp_{D^2} - \fp_{D^2}\vf\fp_{D^2}\fc)\ket{\psi}}\\
    &\leq \norm{\frac{1}{N}\sum_{y\in [N], (D^1,D)\in \mathbf{D}_\mathrm{y,adm}} \gamma_{y,D^1,D}\sum_{y'\in[N]:y'\not\in\im{D}\cap(\im{D^3}\oplus k_3)}\ket{y'}_{D^2_x}\otimes(\ket{\psi^{y}_{D^1,D}}-\ket{\psi^{y'}_{D^1,D}})}\nonumber\\
    &\quad + \norm{\frac{1}{N}\sum_{y\in [N], (D^1,D)\in \mathbf{D}_\mathrm{y,adm}} \gamma_{y,D^1,D}\sum_{y'\in[N]:y'\in\im{D}\cup(\im{D^3}\oplus k_3)}\ket{y'}_{D^2_x}\ket{\psi^{y}_{D^1,D}} }\nonumber\\
    &\quad + \norm{\fp_{D^2}\vf\fp_{D^2}\Biggl( \frac{1}{N}\sum_{y\in [N], (D^1,D)\in \mathbf{D}_\mathrm{y,adm}} \gamma_{y,D^1,D}\sum_{y'\in[N]:y'\in\im{D}\cup(\im{D^3}\oplus k_3)}\ket{y'}_{D^2_x}\ket{D^1}_{D^1}\ket{D}_{(D^2_x)^c}\Biggr)}\nonumber\\
    &\quad +\norm{\frac{1}{\sqrt{N}}\sum_{y\in [N], (D^1,D)\in \mathbf{D}_\mathrm{y,adm}} \gamma_{y,D^1,D} \ket{\bot}_{D^2_x}\otimes(\ket{\psi^{y}_{D^1,D}}-\ket{\psi^{\bot}_{D^1,D}})}\nonumber\\
    &\quad + \norm{\frac{1}{\sqrt{N}}\sum_{y\in [N], (D^1,D)\in \mathbf{D}_\mathrm{y,adm}} \gamma_{y,D^1,D} \ket{\not\bot}_{D^2_x}\ket{\psi^{\not\bot}_{D^1,D}}}.\label{eq:p2-split-5-terms}\end{align} Next, we bound each term separately by applying the triangle inequality to the sum over $y$, then using orthogonality arguments in \Cref{lem:p2-psi-y-d}, and then applying Cauchy-Schwarz. The orthogonality arguments rest on the observation that, on good databases, $\fp_{D^2}\vf\fp_{D^2}$ produces only nonnegative coefficients; since it also preserves orthogonality of the inputs, the resulting terms are distinct across distinct inputs. Since the terms are bounded by essentially the same argument, we give the first in detail and treat the rest briefly.

    First term in \eqref{eq:p2-split-5-terms} is bounded as \begin{align}
        &\leq \frac{1}{N}\sum_{y\in [N]}\norm{\sum_{ (D^1,D)\in \mathbf{D}_\mathrm{y,adm}} \gamma_{y,D^1,D}\sum_{y'\in[N]:y'\not\in\im{D}\cap(\im{D^3}\oplus k_3)}\ket{y'}_{D^2_x}\otimes(\ket{\psi^{y}_{D^1,D}}-\ket{\psi^{y'}_{D^1,D}})}\\
        &= \frac{1}{N}\sum_{y\in [N]} \sqrt{\sum_{ (D^1,D)\in \mathbf{D}_\mathrm{y,adm}} \abs{\gamma_{y,D^1,D}}^2 \sum_{y'\in[N]:y'\not\in\im{D}\cap(\im{D^3}\oplus k_3)} \lVert \ket{\psi^{y}_{D^1,D}}-\ket{\psi^{y'}_{D^1,D}}\rVert^2}\\
        &\qquad \text{(By the orthogonality of $\ket{y'}$'s and $\ket{\psi^y_{D^1,D}}-\ket{\psi^{y'}_{D^1,D}}$'s as given in \Cref{lem:p2-psi-y-d})}\nonumber\\
        &\leq \frac{1}{N}\sum_{y\in [N]} \sqrt{\sum_{ (D^1,D)\in \mathbf{D}_\mathrm{y,adm}} \abs{\gamma_{y,D^1,D}}^2 \sum_{y'\in[N]:y'\not\in\im{D}\cap(\im{D^3}\oplus k_3)} \frac{t}{N}} \quad \text{(By \Cref{lem:p2-psi-y-d})}\\
        &\leq O\Big(\frac{\sqrt{t}}{N}\Big)\sum_{y\in [N]} \sqrt{\sum_{ (D^1,D)\in \mathbf{D}_\mathrm{y,adm}} \abs{\gamma_{y,D^1,D}}^2 }\\
        &\leq O\Big(\frac{\sqrt{t}}{N}\Big) \sqrt{\sum_{y\in [N]} 1^2} \sqrt{\sum_{y\in [N]}\sum_{ (D^1,D)\in \mathbf{D}_\mathrm{y,adm}}\abs{\gamma_{y,D^1,D}}^2 } \quad \text{(By Cauchy-Schwarz)}\label{eq:emsal-CS}\\
        & \leq O(\sqrt{t/N}).
    \end{align}
     Second term in \eqref{eq:p2-split-5-terms} is bounded as \begin{align}
        &\leq \frac{1}{N}\sum_{y\in [N]}\norm{\sum_{(D^1,D)\in \mathbf{D}_\mathrm{y,adm}} \gamma_{y,D^1,D}\sum_{y'\in[N]:y'\in\im{D}\cup(\im{D^3}\oplus k_3)}\ket{y'}_{D^2_x}\ket{\psi^{y}_{D^1,D}} }\\
        &\leq \frac{1}{N}\sum_{y\in [N]} \sqrt{\sum_{ (D^1,D)\in \mathbf{D}_\mathrm{y,adm}} \abs{\gamma_{y,D^1,D}}^2 \sum_{y'\in[N]:y'\in\im{D}\cup(\im{D^3}\oplus k_3)} \lVert \ket{y'}_{D^2_x}\ket{\psi^{y}_{D^1,D}}\rVert^2} \\
        &\leq \frac{1}{N}\sum_{y\in [N]} \sqrt{\sum_{ (D^1,D)\in \mathbf{D}_\mathrm{y,adm}} \abs{\gamma_{y,D^1,D}}^2 2t} \quad \text{(By $\lVert \ket{y'}_{D^2_x}\ket{\psi^{y}_{D^1,D}}\rVert\leq 1$, $\abs{\im{D}\cup(\im{D^3}\oplus k_3)}\leq 2t$)}\\
        &\leq \frac{\sqrt{2t}}{N}\sum_{y\in [N]} \sqrt{\sum_{ (D^1,D)\in \mathbf{D}_\mathrm{y,adm}} \abs{\gamma_{y,D^1,D}}^2 }\\
        & \leq O(\sqrt{t/N}). \quad \text{(By Cauchy-Schwarz as in \eqref{eq:emsal-CS})}
    \end{align}
    Third term in \eqref{eq:p2-split-5-terms} is bounded as \begin{align}
        &\leq \frac{1}{N}\sum_{y\in [N]} \norm{ \sum_{(D^1,D)\in \mathbf{D}_\mathrm{y,adm}} \gamma_{y,D^1,D}\sum_{y'\in[N]:y'\in\im{D}\cup(\im{D^3}\oplus k_3)}\ket{y'}_{D^2_x}\ket{D^1}_{D^1}\ket{D}_{(D^2_x)^c}}\\
        &\leq \frac{1}{N}\sum_{y\in [N]} \sqrt{\sum_{ (D^1,D)\in \mathbf{D}_\mathrm{y,adm}} \abs{\gamma_{y,D^1,D}}^2 \sum_{y'\in[N]:y'\in\im{D}\cup(\im{D^3}\oplus k_3)} 1} \\
        &\leq \frac{\sqrt{t}}{N}\sum_{y\in [N]} \sqrt{\sum_{ (D^1,D)\in \mathbf{D}_\mathrm{y,adm}} \abs{\gamma_{y,D^1,D}}^2} \quad \text{(By  $\abs{\im{D}\cup(\im{D^3}\oplus k_3)}\leq 2t$)}\\
        & \leq O(\sqrt{t/N}). \quad \text{(By Cauchy-Schwarz as in \eqref{eq:emsal-CS})}
    \end{align} 
    Fourth term in \eqref{eq:p2-split-5-terms} is bounded as \begin{align}
        &\leq \frac{1}{\sqrt{N}}\sum_{y\in [N]}\norm{\sum_{ (D^1,D)\in \mathbf{D}_\mathrm{y,adm}} \gamma_{y,D^1,D} \ket{\bot}_{D^2_x}\otimes(\ket{\psi^{y}_{D^1,D}}-\ket{\psi^{\bot}_{D^1,D}})}\\
         &= \frac{1}{\sqrt{N}}\sum_{y\in [N]} \sqrt{\sum_{ (D^1,D)\in \mathbf{D}_\mathrm{y,adm}} \abs{\gamma_{y,D^1,D}}^2 \lVert \ket{\psi^{y}_{D^1,D}}-\ket{\psi^{\bot}_{D^1,D}}\rVert^2}\\
        &\qquad \text{(By the orthogonality of $\ket{y'}$'s and $\ket{\psi^y_{D^1,D}}-\ket{\psi^{\bot}_{D^1,D}}$'s as given in \Cref{lem:p2-psi-y-d})}\nonumber\\
        &\leq \frac{1}{\sqrt{N}}\sum_{y\in [N]} \sqrt{\sum_{ (D^1,D)\in \mathbf{D}_\mathrm{y,adm}} \abs{\gamma_{y,D^1,D}}^2 \frac{t}{N}} \quad \text{(By \Cref{lem:p2-psi-y-d})}\\
        &\leq O\Big(\frac{\sqrt{t}}{N}\Big)\sum_{y\in [N]} \sqrt{\sum_{ (D^1,D)\in \mathbf{D}_\mathrm{y,adm}} \abs{\gamma_{y,D^1,D}}^2 }\\
        & \leq O(\sqrt{t/N}). \quad \text{(By Cauchy-Schwarz as in \eqref{eq:emsal-CS})}\\
    \end{align}
    Fifth term in \eqref{eq:p2-split-5-terms} is bounded as \begin{align}
        &\leq \frac{1}{\sqrt{N}}\sum_{y\in [N]}\norm{\sum_{ (D^1,D)\in \mathbf{D}_\mathrm{y,adm}} \gamma_{y,D^1,D} \ket{\not\bot}_{D^2_x}\ket{\psi^{\not\bot}_{D^1,D}}}\\
        &\leq \frac{1}{\sqrt{N}}\sum_{y\in [N]} \sqrt{\sum_{ (D^1,D)\in \mathbf{D}_\mathrm{y,adm}} \abs{\gamma_{y,D^1,D}}^2 \lVert \ket{\not\bot}_{D^2_x}\ket{\psi^{\not\bot}_{D^1,D}} \rVert^2} \\
        &\leq \frac{1}{\sqrt{N}}\sum_{y\in [N]} \sqrt{\sum_{ (D^1,D)\in \mathbf{D}_\mathrm{y,adm}} \abs{\gamma_{y,D^1,D}}^2 \frac{t}{N}} \quad \text{(By \Cref{lem:p2-psi-y-d})}\\
        &\leq \frac{\sqrt{t}}{N}\sum_{y\in [N]} \sqrt{\sum_{ (D^1,D)\in \mathbf{D}_\mathrm{y,adm}} \abs{\gamma_{y,D^1,D}}^2}\\
        & \leq O(\sqrt{t/N}). \quad \text{(By Cauchy-Schwarz as in \eqref{eq:emsal-CS})}
    \end{align}

\item Lastly, bound the second term in \eqref{eq:p2-comm-vf-fc}. Observe that for any state preserved by $\Pi^\star$, we know $\Pi^g$ gives $x$ is not in $\im{D^1}\oplus k_2$. By definition, in order to preserve injectivity $\vf$ never puts $x\oplus k_2$ in $\im{D^1}$. Hence, $\vr$ never visits $D^2_x$ as a control register whether applied before or after $\fc$. As a result, what $\fc$ does at $D^2_x$ is irrelevant, and the difference of the commutator is 0.
        
    \end{enumerate} 
\end{enumerate}
      
\end{proof}
\section{$P_3$-query Bound Deferred Proofs}\label{app:p3-query-lemma}

Using \Cref{thm:inv-to-fwd-hybrid}, it is enough to consider the difference between the hybrids for forward queries, i.e., consider the real and ideal world oracle queries for $b=0$ and $i=3$.

In this query the oracles in the real and ideal worlds differ. To return the same answer to the adversary's output register, the isometry must perform the fill and the removal on the queried register of the database, namely $D^3_x$. This requires the factor $\vf\ppar{x \oplus k_1}$ to be active, that is, to perform the fill $\ket{\bot}\ket{\bot} \mapsto \ket{\psi_{k_2}}$, which in turn requires that both of its target registers be $\bot$; $\Gamma_1$ and $\Gamma_2$ enforce this for the relevant indices of $D^1$ and $D^2$ registers, respectively. Once $\vf\ppar{x \oplus k_1}$ is active, $\vr\ppar{x \oplus k_1}$ automatically does the rest of the work of uncomputing $D^3_x$.

The remaining factors of the isometry do not play a role here. We therefore split the isometry into a product controlled on the query register as in \Cref{def:vp-vpc}. After bounding the difference between $\vpc\vp$ and $V$ in \Cref{lem:p-vpvpc-V}, we discard the irrelevant parts, namely $\vpc$, by bounding its commutator with $\cco^R$ and using the unitarity of the norm in \Cref{lem:comm-real-vpc}. Hence, we have removed the order dependence of the factors in $V$. 

 \begin{definition} \label{def:vp-vpc}
    \begin{align}
        \vp&\coloneqq \sum_{x,k_1\in[N]} \proj{x}_X \otimes \proj{k_1}_{K_1} \otimes \vr\ppar{x\oplus k_1}\fp_{D_2} \vf\ppar{x\oplus k_1} \fp_{D_2} \label{eq:def-vp}\\
        \vpc&\coloneqq \sum_{x,k_1\in[N]} \proj{x}_X \otimes \proj{k_1}_{K_1} \otimes \prod_{\substack{x'\in[N]:\\x'\neq x\oplus k_1}}\vr\ppar{x'}\fp_{D_2}  \prod_{\substack{x'\in[N]:\\x'\neq x\oplus k_1}}\vf\ppar{x'} \fp_{D_2}
    \end{align}
    For fixed $x, k_1\in[N]$ at registers $X,K_1$, respectively, we use the notation
    \begin{align}
      \vfpc\coloneqq  \prod_{x'\in[N]:x'\neq x\oplus k_1}\vf\ppar{x'} \text{ and } \vrpc\coloneqq  \prod_{x'\in[N]:x'\neq x\oplus k_1}\vr\ppar{x'}. 
    \end{align}
\end{definition}

\begin{lemma}\label{lem:p-vpvpc-V}
\begin{align}
    \norm{(V-\vpc\vp)\Pi^\star\Pi^t}=0
\end{align}
\end{lemma}
\begin{proof}
    Fix $X$ and $K_1$ registers to $x$ and $k_1$, respectively, in computational basis since they are control registers. Then write \begin{align}
         &\norm{(V-\vpc\vp)\Pi^\star\Pi^t}\\
         &= \lVert(V\pm \vrpc\vr\ppar{x\oplus k_1}\fp_{D^2}\vfpc\fp_{D^2}\fp_{D^2}\vf\ppar{x\oplus k_1}\fp_{D^2} \nonumber\\
         &\quad-\vrpc\fp_{D^2}\vfpc\fp_{D^2}\vr\ppar{x\oplus k_1}\fp_{D^2}\vf\ppar{x\oplus k_1}\fp_{D^2})\Pi^\star\Pi^t\rVert\\
         &\leq \norm{(V -  \vrpc\vr\ppar{x\oplus k_1}\fp_{D^2}\vfpc\fp_{D^2}\fp_{D^2}\vf\ppar{x\oplus k_1}\fp_{D^2})\Pi^\star\Pi^t} \nonumber\\
         &\quad + \norm{\vrpc\Big[\vr\ppar{x\oplus k_1},\fp_{D^2}\vfpc\fp_{D^2}\Big] \fp_{D^2}\vf\ppar{x\oplus k_1}\fp_{D^2}\Pi^\star\Pi^t}.\label{eq:v-vpvpc-split}
    \end{align} 
    
    Observe that the first term in \eqref{eq:v-vpvpc-split} vanishes as \begin{align}
         &\norm{(V -  \vrpc\vr\ppar{x\oplus k_1}\fp_{D^2}\vfpc\fp_{D^2}\fp_{D^2}\vf\ppar{x\oplus k_1}\fp_{D^2})\Pi^\star\Pi^t}\\
         &=\norm{(V -  \vr\fp_{D^2}\vfpc\vf\ppar{x\oplus k_1}\fp_{D^2})\Pi^\star\Pi^t}\\
         &=\norm{(V -  \vr\fp_{D^2}\vf\fp_{D^2})\Pi^\star\Pi^t}
    \end{align} where the first equality is because $\vr\ppar{\cdot}$'s commute pairwise as they have disjoint targets and $\fp_{D^2}$ is self-inverse, and the second equality follows from \Cref{rem:comm-vf-factors} asserting on injective databases $\vf\ppar{\cdot}$'s have disjoint targets and on good databases the set of values to fill is fixed for all.

    The second term in \eqref{eq:v-vpvpc-split} also vanishes. Notice that operators of the commutator could only have a conflicting register on $D^2$ since they act on disjoint indices of $D^1$ and $D^3$. Both operators preserve $D^1_{x\oplus k_1}$, it is enough to consider the cases that it is $\bot$ or not. If it is $\bot$, $\vr\ppar{x\oplus k_1}$ acts as identity anyway and the commutator vanishes. If it is non-$\bot$, start by fixing $J$ for $\vpc$ as defined in \Cref{def:J} and note that $\vr\ppar{x\oplus k_1}$ preserves $J$ since it does not act on $D^1$ or $D^2$. Because $D^1_{x\oplus k_1}$ is non-$\bot$, $u\not\in J$. Hence, $D^2_{D^1(x)\oplus k_2}$ is not a target for $\fp_{D^2}\vfpc\fp_{D^2}$, it is preserved, and  is the only register $\vr\ppar{x\oplus k_1}$ acts on $D^2$.
\end{proof}

\begin{lemma}\label{lem:comm-real-vpc}
\begin{align}
   \norm{[\vpc,\cco^R]\vp\Gamma_2\Gamma_1\Pi^\star\Pi^t\Pi^{k_1} \Pi^{k_3}} \leq O(\sqrt{t/N}) .
\end{align}
\end{lemma}

\begin{proof}
 Throughout, we fix the $X$ and the key registers in the computational basis in a standard way since they are controlled on or trivially acted on by all operators. First, factor $\cco^R$ as $\fc\ppar{1}\cdot \w\cdot \fc\ppar{1}$ where $\w\coloneqq \fc\ppar{2}\cdot\pu^{em}\cdot\fc\ppar{2}$ and denote $\Pi\coloneqq \Gamma_2\Gamma_1\Pi^\star\Pi^t\Pi^{k_3}$, then we reduce the commutator to sum of simpler commutators as \begin{align}
      & \lVert[\vpc,\cco^R]\vp\Pi\rVert \\
       &\leq  \lVert [\fc\ppar{1},\vpc]\w \fc\ppar{1} \vp\Pi\rVert+ \lVert\fc\ppar{1}[\w,\vpc]\fc\ppar{1}\vp\Pi \rVert \nonumber+\lVert \fc\ppar{1} \w[\fc\ppar{1},\vpc]\vp\Pi\rVert\nonumber\\
       &\leq  \lVert[\w,\vpc]\fc\ppar{1}\vp\Pi \rVert + \lVert [\fc\ppar{1},\vpc]\w \fc\ppar{1} \vp\Pi\rVert+\lVert [\fc\ppar{1},\vpc]\vp\Pi \rVert
       \label{eq:splitted-vpc-cor}
    \end{align} where we used the unitarity of the norm on the last line. We bound each term separately.
    \begin{enumerate}[leftmargin=0pt, labelindent=0pt, itemindent=*, label=(\arabic*)]
        \item  We start with the first term in \eqref{eq:splitted-vpc-cor}. Use the triangle inequality to simplify it: \begin{align}
         \lVert[\w,\vpc]\fc\ppar{1}\vp\Pi \rVert &\leq \lVert\vrpc[\w,\vfpc]\fc\ppar{1}\vp\Pi \rVert +    \lVert[\w,\vrpc]\vfpc\fc\ppar{1}\vp\Pi \rVert\\
         &\leq \lVert[\w,\vfpc]\fc\ppar{1}\vp\Pi \rVert +    \lVert[\w,\vrpc]\vfpc\fc\ppar{1}\vp\Pi \rVert.
        \end{align} We will classify with respect to the state of $D^3_x$ being $\bot$ or not.
        \begin{enumerate} [leftmargin=0pt, labelindent=0pt, itemindent=*, label=(\alph*)]
            \item If $D^3_x=\bot$, then $\vp$ acts as identity by definition and the fact that $\Gamma_1$ sets $D^1_{x\oplus k_1}=\bot$. Define a partial decompression operator $\widetilde{\fc\ppar{1}}$ controlled on $\im{D^1}$ and $\dom{D^2}\oplus k_2$ so that they are not attained as image of $D^1_{x\oplus k_1}$, i.e.,
            \begin{align}\label{eq:fcppar1-tilde}
              \widetilde{\fc\ppar{1}}\coloneqq \sum_{x, k_1, k_2,k_3\in [N], (D^3_x)^c,D^2\in\mathbf{D}_{t}}
  \proj{x}_X \otimes \proj{k_1}_{K_1} \otimes \proj{k_2}_{K_2} \otimes\\ \otimes \proj{(D^1_x)^c}_{(D^1_x)^c}
  \otimes \proj{D^2}_{D^2} \otimes \tilde{\fc}_{D^1_{x \oplus k_1}}
\end{align}
where $\widetilde{\fc}_x\coloneqq \operatorname{Exc}\big(\ket{\bot}, \ket{[N]\setminus \im{(D^1_x)^c} \cup (\dom{D^2}\oplus k_2)}\big)$. Note that by \Cref{lem:diff-full-and-partial-decomp}, $\lVert (\fc\ppar{1}- \widetilde{\fc\ppar{1}})\Pi^t\rVert\leq O(\sqrt{t/N})$. Observe that using $\widetilde{\fc\ppar{1}}$, we make sure that $\w$ acts on the registers at $D^2$ that holds $\bot$ only. Similarly, replace the decompression operator $\fc\ppar{2}$ in $\w$ by $\widetilde{\w}$ such that it is controlled on $\im{(D^3_x)^c}\oplus k_3$,\begin{align}
                \widetilde{\w} \coloneqq  \widetilde{\fc\ppar{2}}\cdot\pu^{em}\cdot\widetilde{\fc\ppar{2}}  \label{eq:w-tilde}
            \end{align}where \begin{align}
    \widetilde{\fc\ppar{2}} &\coloneqq \sum_{\substack{x, y_1, k_1, k_2, k_3 \in [N] \\ (D^3_x)^c \in \mathbf{D}_t}}
    \proj{x}_X \otimes \proj{k_1}_{K_1} \otimes \proj{k_2}_{K_2} \otimes \proj{k_3}_{K_3} \nonumber\\
    &\otimes \proj{(D^3_x)^c}_{(D^3_x)^c} \otimes \proj{y_1}_{D^1_{x \oplus k_1}} \otimes \tilde{\fc}_{D^2_{y_1 \oplus k_2}}
\end{align}
where $\widetilde{\fc}_x\coloneqq \operatorname{Exc}\big(\ket{\bot}, \ket{[N]\setminus im{(D^3_x)^c}\oplus k_3 }\big)$, hence it does not introduce any value in $(D^3_x)^c \oplus k_3$ to $D^2$ after $\widetilde{\w}$ is applied, and again by \Cref{lem:diff-full-and-partial-decomp}, $\lVert (\w- \widetilde{\w})\Pi^t\rVert\leq O(\sqrt{t/N})$.\footnote{The lemma is stated for a single (de)compression; the case of an oracle query reduces to a difference of (de)compressions by the triangle inequality, since big-$O$ notation is insensitive to a constant number of additional terms.} Combining this with the preserved injectivity of $D^1$ by $\widetilde{\fc\ppar{1}}$, we see that $\widetilde{\w}$ and $\vfpc$ act on disjoint registers. Another interaction between the two is through the set $J$ as in \Cref{def:J} for $\vfpc$, which is preserved by $\widetilde{\w}$, since the relation between the two is through the domain of $D^2$, and the indices on which $\widetilde{\w}$ acts are already outside $J$, as they lie in $\im{D^1} \oplus k_2$. The only thing between the two is $\widetilde{\w}$ may result in failure of the flip, once we project out these, the difference vanishes. As a result,
            \begin{align}
                 &\lVert[\w,\vfpc]\fc\ppar{1}\vp\Pi \rVert \\
                 &\leq \lVert[\w,\vfpc]\widetilde{\fc\ppar{1}}\vp\Pi \rVert + 2\lVert \fc\ppar{1}- \widetilde{\fc\ppar{1}}\rVert\\
                 &\leq \lVert[\widetilde{\w},\vfpc]\widetilde{\fc\ppar{1}}\vp\Pi \rVert + \lVert  (\w- \widetilde{\w} ) \vfpc\widetilde{\fc\ppar{1}}\vp\Pi \rVert+ \lVert  (\w- \widetilde{\w} ) \widetilde{\fc\ppar{1}}\vp\Pi \rVert + O(\sqrt{t/N})\\
                 &\leq \lVert[\widetilde{\w},\vfpc]\widetilde{\fc\ppar{1}}\vp\Pi \rVert + \lVert  (\w- \widetilde{\w} ) \Pi^{2t} \rVert+ \lVert  (\w- \widetilde{\w} ) \Pi^{t+1} \rVert + O(\sqrt{t/N})\\
                 &\leq \lVert[\widetilde{\w},\vfpc]\widetilde{\fc\ppar{1}}\vp\Pi \rVert + O(\sqrt{t/N})\\
                 &\leq  \lVert(\Pi^{inj}_{D^2}\widetilde{\w}\vfpc- \vfpc\Pi^{inj}_{D^2}\widetilde{\w})\widetilde{\fc\ppar{1}}\vp\Pi \rVert + \lVert (\Pi^{inj}_{D^2})^\perp \widetilde{\w}\vfpc \widetilde{\fc\ppar{1}}\vp\Pi \rVert \nonumber\\
                 &\quad+ \lVert (\Pi^{inj}_{D^2})^\perp \widetilde{\w}\widetilde{\fc\ppar{1}}\vp\Pi \rVert  + O(\sqrt{t/N})\\
                 & \leq  \lVert(\Pi^{inj}_{D^2}\widetilde{\w}\vfpc- \vfpc\Pi^{inj}_{D^2}\widetilde{\w})\widetilde{\fc\ppar{1}}\vp\Pi \rVert + O(\sqrt{t/N})\\
                 &= 0 + O(\sqrt{t/N})
            \end{align} where the line before the last follows from \begin{align}
                &\lVert (\Pi^{inj}_{D^2})^\perp \widetilde{\w}\vfpc \widetilde{\fc\ppar{1}}\vp\Pi \rVert+ \lVert (\Pi^{inj}_{D^2})^\perp \widetilde{\w}\widetilde{\fc\ppar{1}}\vp\Pi \rVert\\
                &\leq \lVert (\Pi^{inj}_{D^2})^\perp \widetilde{\w}\Pi^{inj}_{D^2}\Pi^{2t} \rVert+ \lVert (\Pi^{inj}_{D^2})^\perp \widetilde{\w}\Pi^{inj}_{D^2}\Pi^{t+1} \rVert \\
                & \leq O(\sqrt{t/N-2t})
            \end{align} using the fact that all of $\vp, \vfpc, \widetilde{\fc\ppar{1}}$  preserve the injectivity of $D^2$ by definition and increase the size of the database by at most a constant factor, and $\widetilde{\w}$ is a decompression, using \Cref{lem:pi-star-violation}.

            Furthermore, by only using the injective preserving property of $\widetilde{\fc\ppar{1}}$ and the fact that $\fp_{D^2}\vfpc\fp_{D^2}$ preserves the injectivity by definition, we see that $D^1(x\oplus k_1)\oplus$ is not a control register for $\vrpc$, hence $\w$ and $\vrpc$ acts on disjoint registers and commute perfectly, i.e., \begin{align}
                &\lVert[\w,\vrpc]\vfpc\fc\ppar{1}\vp\Pi \rVert \\
                 &\leq \lVert[\w,\vrpc]\vfpc\widetilde{\fc\ppar{1}}\vp\Pi \rVert + 2\lVert \fc\ppar{1}- \widetilde{\fc\ppar{1}}\rVert\\
                 &\leq 0 + O(\sqrt{t/N}).
            \end{align}
        \end{enumerate}
\begin{enumerate}[leftmargin=0pt, labelindent=0pt, itemindent=*, label=(\alph*)]
    \item If $D^3_x\neq\bot$, then image of a computational basis database state $\ket{\psi}\in\Pi$ under $\vp$ is of the form, omitting untouched registers, \begin{align}
        \vp\ket{\psi}=\frac{1}{\sqrt{J}}\sum_{u\in J} \ket{u}_{D^1_{x\oplus k_1}}\ket{D^3(x\oplus k_1)\oplus k_3}_{D^2(u\oplus k_2)}.
    \end{align} Replacing $\fc\ppar{1}$ with $\widetilde{\fc\ppar{1}}$ as defined in \eqref{eq:fcppar1-tilde}, since it preserves $\dom{D^2 \oplus k_2}$ values, it acts as the identity on $\vp\ket{\psi}$ for every branch $u$. Replacing $\w$ with $\widetilde{\w}$ as defined in \Cref{eq:w-tilde}, the decompression neither introduces values that are targets of $\vfpc$ nor acts on them, by the injectivity preserving property of $\vfpc$ and the fact that the targets of $\widetilde{\w}$ already lie in $\im{D^1} \oplus k_2$. Therefore $\widetilde{\w}$ preserves $J$ for $\vfpc$. Consequently, removing the injectivity violation on $D^2$ introduced by $\w$, the remaining difference vanishes, and the exact same derivation as in the case $D^3_x = \bot$ concludes the proof giving \begin{align}
\lVert[\w,\vfpc]\fc\ppar{1}\vp\Pi \rVert \leq O(\sqrt{t/N}),        
    \end{align} as well as, \begin{align}
        \lVert[\w,\vrpc]\vfpc\fc\ppar{1}\vp\Pi \rVert\leq O(\sqrt{t/N}).
    \end{align}    
\end{enumerate}

\item  We bound $\lVert [\fc\ppar{1},\vpc]\vp\Pi \rVert$. We simplify the commutator to argue more explicitly \begin{align}
 &\norm{[\fc\ppar{1},\vpc]\vp\Pi } \leq  \norm{[\fc\ppar{1}, \vrpc] \fp_{D^2}\vfpc\fp_{D^2}} + \norm{[\fc\ppar{1}, \fp_{D^2}\vfpc\fp_{D^2}] \vp\Pi }\label{eq:comm-fc-vpc-split}
\end{align} We bound this by observing its similarity to
$\norm{[V,\fc]\Gamma_1\Pi^\star\Pi^t}$ in \eqref{eq:comm-v-fc} for the
$P_1$-query. The argument there rests on the following: before the
commutator, the database satisfies one of the criteria of
\Cref{rem:fill-or-identity}, so that every active $\vf\ppar{\cdot}$ either
performs the fill $\ket{\bot}\ket{\bot}\mapsto\ket{\psi_{k_2}}$ or acts as
the identity; consequently, on such computational basis states, the first
part of the isometry $\fp_{D^2}\vf\fp_{D^2}$ creates only terms with
positive coefficients. Here the situation is even simpler, since $\vpc$ does
not contain the factor $\vr\ppar{x\oplus k_1}$, so the action of
$\fc\ppar{1}$ makes no difference for the second part of the isometry. We
therefore do not rederive the bound, and only explain why the state
satisfies a criterion of \Cref{rem:fill-or-identity}.

Consider a state $\ket{\psi}$ in the image of $\Pi$; in particular
$\Pi^g\ket{\psi}=\ket{\psi}$. The state $\vp\ket{\psi}$ is no longer in the
image of $\Pi^g$. If, however, we restrict the index of the register $D^1$
to $[N]\setminus\set{x\oplus k_1}$, then the restricted database does
satisfy \Cref{def:pi-g}. Indeed, $\vp$ either acts as the identity (when
$D^3_x=\bot$) or performs the fill
$\ket{\bot}\ket{\bot}\mapsto\ket{\psi_{k_2}}$ (when $D^3_x\neq\bot$), and by
definition the indices of the registers filled at $D^2$ do not lie in
$\im{D^1}\oplus k_2$ for the initial database $D^1$. Hence the restricted
database after $\vp$ is good. Since every $x'$ for which $\vf\ppar{x'}$ is
active in $\vpc$ already lies in $[N]\setminus\set{x\oplus k_1}$,
\Cref{item:rem:fill-id-restricted-D1} of \Cref{rem:fill-or-identity} applies:
$\vpc$ performs only the fill $\ket{\bot}\ket{\bot}\mapsto\ket{\psi_{k_2}}$
or acts as the identity, so we may fix a set $J$ as in \Cref{def:J} for any
computational basis state before applying the commutator.

For the first term in \eqref{eq:comm-fc-vpc-split}, the only conflicting
register is $D^1_{x\oplus k_1}$: it is a target of $\fc\ppar{1}$ and a
control of $\vfpc$. The dependence is that $\vf\ppar{x'}$, for
$x'\neq x\oplus k_1$, no longer fills with the value held by
$D^1_{x\oplus k_1}$ if that value lies in $J$. Thus for the term
$\fc\ppar{1}\vfpc$ each $\vf\ppar{x'}$ fills from $J$, whereas for the term
$\vfpc\,\fc\ppar{1}$ it fills from a set that may differ from $J$ by one
missing element. The former therefore carries additional terms, which when
projected out are bounded by $O(\sqrt{t/N})$. The remaining difference
consists of the same terms, differing only in their coefficients, and is
bounded by the same quantity as in \eqref{eq:boring-counting}, used for
\eqref{eq:comm-v-fc}.

For the second term in \eqref{eq:comm-fc-vpc-split}, the two operators act
on disjoint registers, since $\vrpc$ neither acts on nor is controlled by
$D^1_{x\oplus k_1}$. Hence the commutator vanishes.

\item We bound $\lVert [\fc\ppar{1},\vpc]\w \fc\ppar{1} \vp\Pi\rVert$. We will give an explanation for why before the commutator, the database satisfies one of the criteria of \Cref{rem:fill-or-identity} and restore injectivity of $D^2$ for the success of the flip operator since $\w$ acts on $D^2$, then the rest follows identically.

Again observe that $\vp$ preserves $\Pi^g$ restricted to the index of the register $D^1$
to $[N]\setminus\set{x\oplus k_1}$. Because of the restricted indices $\fc\ppar{1}$ preserves goodness. If we just restore injectivity after $\fc\ppar{1}$ at a cost $O(\sqrt{t/N})$ by \Cref{lem:pi-star-violation}, then since the operator $\w$ acts on $D^2_{D^1(x\oplus k_1)\oplus k_2}$ on injective $D^1$ database, it preserves goodness. Lastly after $\w$ we restore injectivity of $D^2$ at a cost $O(\sqrt{t/N})$ by \Cref{lem:pi-star-violation} again.\footnote{Note that \Cref{lem:pi-star-violation} is for a standard query operator controlled on $X$, while $\w$ is controlled both on $X$ and $K_1$. However, the same argument applies since $\w$ is block diagonal on these registers.} Then the rest of the argument for the commutator exactly follows as in the previous case.
     \end{enumerate}
    
\end{proof}

The next lemma quantifies how far the decompression operators move a computational basis database state.
\begin{lemma}\label{lem:p3-decomp-divergence} 
Let $\ket{\psi}$ be a state in which the registers $X, K_1, K_2$ hold fixed values in $[N]$ and the register $(D^3_x)^c$ holds a fixed database in $\mathbf{D}_t$, omitting these fixed registers, we write
\begin{align}
  \ket{\psi} = \ket{+^n}_{K_3}\otimes
    \sum_{\substack{y\in I_1,\ y'\in I_2\\ D\in\mathbf{D}_{inj,t}}}
    \gamma_{y,y',D}\,\ket{\bot}_{D^1_{x\oplus k_1}}
    \ket{\hat{\ell}_y}_Y\,\ket{\hat{\ell}_{y'}}_{D^3_x}\,
    \ket{D}_{D^2}\,\ket{\psi_{y,y',D}},
\end{align}
where $I_1, I_2 \subseteq [N]\setminus\set{0}$; the joint state of $Y$ and $D^3_x$ is, in the Hadamard basis, pairwise orthogonal across distinct pairs $(y,y')\in I_1\times I_2$ and the $\ket{\psi_{y,y',D}}$ are subnormalized. Then
    \begin{align}
        \lVert(\Id-\fc\ppar{1})\vp\Gamma_2\ket{\psi}\rVert \,,\, \lVert(\Id-\fc\ppar{2})\vp\Gamma_2\ket{\psi}\rVert\leq  O(\sqrt{t/N}),\label{eq:p3-id-fc}
    \end{align} and consequently, $\lVert(\Id-\fc\ppar{2}\fc\ppar{1})\vp\Gamma_2\ket{\psi}\rVert \leq O(\sqrt{t/N})$. Moreover, if $\ell_{y'}=0$, then the bound for the first term in $\eqref{eq:p3-id-fc}$ still holds.

\end{lemma}

\begin{proof}
    First, observe that \begin{align}
        &\lVert(\Id \pm \fc\ppar{2} -\fc\ppar{2} \fc\ppar{1})\vp\Gamma_2\ket{\psi}\rVert\\
        &\leq  \lVert\fc\ppar{2}(\Id-\fc\ppar{1})\vp\Gamma_2\ket{\psi}\rVert + \lVert(\Id-\fc\ppar{2})\vp\Gamma_2\ket{\psi}\rVert\\
        &=  \lVert(\Id-\fc\ppar{1})\vp\Gamma_2\ket{\psi}\rVert + \lVert(\Id-\fc\ppar{2})\vp\Gamma_2\ket{\psi}\rVert
    \end{align} by the triangle inequality on the first line and the unitarity of the norm on the second line. Recalling that $\vp=\vr\ppar{x\oplus k_1}\fp_{D^2}\vf\ppar{x\oplus k_1}\fp_{D^2}$ and $\ket{\hat{y}}=\sum(-1)^{y.z}\ket{z}$, $z\in[N]$, where $y\cdot z$ denotes the inner product of the $n$-bit strings $y$ and $z$ modulo $2$, and denoting $J_D$ in \Cref{def:J} for each $D$, we write \begin{align}
&\vp \Gamma_2 \ket{\psi}\\
  &= \vp \Biggl(\sum_{y\in I_1,y'\in I_2} \ket{\hat{\ell_y}}_Y 
  \frac{1}{\sqrt{N}}\sum_{z\in [N]} (-1)^{\ell_{y'}.z} \ket{z}_{D^3_x}\ket{\bot}_{D^1_{x\oplus k_1}}\sum_{\substack{D\in \mathbf D_{inj, t}}} \gamma_{y,y',D} \ket{D}_{D^2}\ket{\psi_{y,y',D}} \frac{1}{\sqrt{N}}\sum_{\substack{k_3\in[N]:\\
  z\oplus k_3 \not \in \im{D}}} \ket{k_3}_{K_3}\Biggr)\\
  &= \vr\ppar{x\oplus k_1} \fp_{D^2} \Biggl( \sum_{y\in I_1,y'\in I_2} \ket{\hat{\ell_y}}_Y 
  \frac{1}{\sqrt{N}}\sum_{z\in [N]} (-1)^{\ell_{y'}.z} \ket{z}_{D^3_x}\sum_{\substack{D\in \mathbf D_{inj, t}}} \gamma_{y,y',D} \ket{\psi_{y,y',D}} \frac{1}{\sqrt{N}} \nonumber\\
  &\quad \cdot\sum_{\substack{k_3\in[N]:\\
  z\oplus k_3 \not \in \im{D}}}\ket{k_3}_{K_3}\ket{((D^2
  _{z\oplus k_3})^c)^{-1}}_{(D^2_{z\oplus k_3})^c} \frac{1}{\sqrt{\abs{J_D}}} \sum_{u\in J_D} \ket{u}_{D^1_{x\oplus k_1}} \ket{u \oplus k_2}_{D^2_{z\oplus k_3}}\Biggr)\\
  &=\sum_{y\in I_1,y'\in I_2} \ket{\hat{\ell_y}}_Y \ket{\bot}_{D^3_x}
  \frac{1}{\sqrt{N}}\sum_{z\in [N]}  \sum_{\substack{D\in \mathbf D_{inj, t}}} \gamma_{y,y',D} \ket{\psi_{y,y',D}}\frac{1}{\sqrt{N}} \sum_{\substack{k_3\in[N]:\\
  z\oplus k_3 \not \in \im{D}}}\ket{k_3}_{K_3}\nonumber\\
  &\quad \cdot \frac{1}{\sqrt{\abs{J_D}}} \sum_{u\in J_D} (-1)^{\ell_{y'}.z} \ket{u}_{D^1_{x\oplus k_1}}   \ket{z\oplus k_3}_{D^2_{u \oplus k_2}}\ket{(D_{u \oplus k_2})^c}_{(D^2_{u \oplus k_2})^c}.
  \end{align}
Before applying $(\Id-\fc\ppar{1})$, we define a partial (de)compression operator as 
\begin{align}
    \widetilde{\fc\ppar{1}}\coloneqq \sum_{x, k_1, k_2\in [N], D^2\in\mathbf{D}_t}
  \proj{x}_X \otimes \proj{k_1}_{K_1} \otimes \proj{k_2}_{K_2}
  \otimes \proj{D^2}_{D^2} \otimes \tilde{\fc}_{D^1_{x \oplus k_1}}
\end{align}
where $\widetilde{\fc}_x\coloneqq \operatorname{Exc}(\ket{\bot}, \ket{[N]\setminus \dom{D^2}\oplus k_2})$, then using \Cref{lem:diff-full-and-partial-decomp}, we can conclude $\lVert\widetilde{\fc\ppar{1}}- \fc\ppar{1}\rVert\leq O(\sqrt{t/N})$. We introduce this because it acts as the identity whenever the target register holds a value outside $\dom{D^2}\oplus k_2$. We integrate this as\begin{align}
\lVert(\Id-\fc\ppar{1})\vp\Gamma_2\ket{\psi}\rVert 
&=\lVert(\Id\pm \widetilde{\fc\ppar{1}} -\fc\ppar{1})\vp\Gamma_2\ket{\psi}\rVert\\
&\leq \lVert(\Id- \widetilde{\fc\ppar{1}})\vp\Gamma_2\ket{\psi}\rVert + \lVert( \widetilde{\fc\ppar{1}} -\fc\ppar{1})\vp\Gamma_2\ket{\psi}\rVert\\
&\leq 0 + \sqrt{t/N}.
\end{align} 
Next, rearrange $\vp \Gamma_2 \ket{\psi}$ by introducing the factor $1=(-1)^{\ell_{y'}.k_3}(-1)^{\ell_{y'}.k_3}$ as \begin{align}
&\vp \Gamma_2 \ket{\psi}\\
&=\sum_{y\in I_1,y'\in I_2} \ket{\hat{\ell_y}}_Y \ket{\bot}_{D^3_x}
  \frac{1}{\sqrt{N}}\sum_{z\in [N]}  \sum_{\substack{D\in \mathbf D_{inj, t}}} \gamma_{y,y',D} \ket{\psi_{y,y',D}}\frac{1}{\sqrt{N}} \sum_{\substack{k_3\in[N]:\\
  z\oplus k_3 \not \in \im{D}}}(-1)^{\ell_{y'}.k_3}\ket{k_3}_{K_3}\nonumber\\
&\quad \cdot \frac{1}{\sqrt{\abs{J_D}}} \sum_{u\in J_D} (-1)^{\ell_{y'}.(z\oplus k_3)} \ket{u}_{D^1_{x\oplus k_1}}   \ket{z\oplus k_3}_{D^2_{u \oplus k_2}}\ket{(D_{u \oplus k_2})^c}_{(D^2_{u \oplus k_2})^c}\\
&=\sum_{y\in I_1,y'\in I_2} \ket{\hat{\ell_y}}_Y \ket{\bot}_{D^3_x}
  \frac{1}{\sqrt{N}}\sum_{z\in [N]}  \sum_{\substack{D\in \mathbf D_{inj, t}}} \gamma_{y,y',D} \ket{\psi_{y,y',D}}\frac{1}{\sqrt{N}} \sum_{\substack{k_3\in[N]:\\
  z\oplus k_3 \not \in \im{D}}}(-1)^{\ell_{y'}.k_3}\ket{k_3}_{K_3}\nonumber\\
&\quad \cdot \frac{1}{\sqrt{\abs{J_D}}} \sum_{u\in J_D} (-1)^{\ell_{y'}.(z\oplus k_3)} \ket{u}_{D^1_{x\oplus k_1}}   \ket{z\oplus k_3}_{D^2_{u \oplus k_2}}\ket{(D_{u \oplus k_2})^c}_{(D^2_{u \oplus k_2})^c} \pm \ket{\phi\ppar{k_3}}\\
&=\sum_{y\in I_1,y'\in I_2} \ket{\hat{\ell_y}}_Y \ket{\bot}_{D^3_x}
  \frac{1}{\sqrt{N}}  \sum_{\substack{D\in \mathbf D_{inj, t}}} \gamma_{y,y',D} \ket{\psi_{y,y',D}}\frac{1}{\sqrt{N}} \sum_{k_3\in[N]}(-1)^{\ell_{y'}.k_3}\ket{k_3}_{K_3}\nonumber\\
&\quad \cdot \frac{1}{\sqrt{\abs{J_D}}} \sum_{u\in J_D} \ket{u}_{D^1_{x\oplus k_1}}   \ket{\hat{\ell_{y'}}}_{D^2_{u \oplus k_2}}\ket{(D_{u \oplus k_2})^c}_{(D^2_{u \oplus k_2})^c} - \ket{\phi\ppar{k_3}}
\end{align} where \begin{align}
  \ket{\phi\ppar{k_3}}&\coloneqq\sum_{y\in I_1,y'\in I_2} \ket{\hat{\ell_y}}_Y \ket{\bot}_{D^3_x}
  \frac{1}{\sqrt{N}}\sum_{z\in [N]}  \sum_{\substack{D\in \mathbf D_{inj, t}}} \gamma_{y,y',D} \ket{\psi_{y,y',D}}\frac{1}{\sqrt{N}} \sum_{\substack{k_3\in[N]:\\
  z\oplus k_3 \in \im{D}}}(-1)^{\ell_{y'}.k_3}\ket{k_3}_{K_3}\nonumber\\
&\quad \cdot \frac{1}{\sqrt{\abs{J_D}}} \sum_{u\in J_D} (-1)^{\ell_{y'}.(z\oplus k_3)} \ket{u}_{D^1_{x\oplus k_1}}   \ket{z\oplus k_3}_{D^2_{u \oplus k_2}}\ket{(D_{u \oplus k_2})^c}_{(D^2_{u \oplus k_2})^c}.  
\end{align} Using the fact that $\fc\ppar{2}$ acts as the identity on non-zero Hadamard basis states, we conclude \begin{align}
    \lVert(\Id-\fc\ppar{2})\vp\Gamma_2\ket{\psi}\rVert&\leq \lVert(\Id-\fc\ppar{2})(\vp\Gamma_2\ket{\psi}\pm \ket{\phi\ppar{k_3}}) \rVert\\
    &\leq \lVert(\Id-\fc\ppar{2})(\vp\Gamma_2\ket{\psi}+ \ket{\phi\ppar{k_3}}) \rVert +  2 \lVert\ket{\phi\ppar{k_3}}) \rVert\\
    &\leq 0 + O(\sqrt{t/N})
\end{align} where the second term follows by the change of variable $w\coloneqq z\oplus k_3$ in $\ket{\phi\ppar{k_3}}$ as \begin{align}
\ket{\phi\ppar{k_3}}&\sum_{y\in I_1,y'\in I_2} \ket{\hat{\ell_y}}_Y \ket{\bot}_{D^3_x}
  \frac{1}{\sqrt{N}}  \sum_{\substack{D\in \mathbf D_{inj, t}}} \gamma_{y,y',D} \ket{\psi_{y,y',D}}\frac{1}{\sqrt{N}} \sum_{k_3\in[N]}(-1)^{\ell_{y'}.k_3}\ket{k_3}_{K_3}\nonumber\\
&\quad \otimes \sum_{w\in [N]: w\in \im{D}}\frac{1}{\sqrt{\abs{J_D}}} \sum_{u\in J_D} (-1)^{\ell_{y'}.w} \ket{u}_{D^1_{x\oplus k_1}}   \ket{w}_{D^2_{u \oplus k_2}}\ket{(D_{u \oplus k_2})^c}_{(D^2_{u \oplus k_2})^c}\\
&=\sum_{y\in I_1,y'\in I_2} \ket{\hat{\ell_y}}_Y \ket{\bot}_{D^3_x}
  \frac{1}{\sqrt{N}}  \sum_{\substack{D\in \mathbf D_{inj, t}}} \gamma_{y,y',D} \ket{\psi_{y,y',D}}\ket{\hat{\ell_{y'}}}_{K_3}\nonumber\\
&\quad \otimes \sum_{w\in [N]: w\in \im{D}}\frac{1}{\sqrt{\abs{J_D}}} \sum_{u\in J_D} (-1)^{\ell_{y'}.w} \ket{u}_{D^1_{x\oplus k_1}}   \ket{w}_{D^2_{u \oplus k_2}}\ket{(D_{u \oplus k_2})^c}_{(D^2_{u \oplus k_2})^c}
\end{align} whose norm is $\sqrt{\sum_{y\in I_1,y'\in I_2}  \sum_{\substack{D\in \mathbf D_{inj, t}}} \abs{\gamma_{y,y',D}}^2 \sum_{w\in [N]: w\in \im{D}} \frac{1}{N^2\abs{J_D}} \lVert (-1)^{\ell_{y'}.w} \ket{\psi_{y,y',D} } \rVert^2 }=O(\sqrt{t/N})$.
\end{proof}

\subsection{Proof of \Cref{thm:p3-bound}}
\begin{theorem}[$P_3$-query]
\begin{align}
     \lVert (\cco^R V - V \Pi^\star \cco^I) \Pi^\star \Pi^t \Pi^{k_1} \Pi^{k_3} \rVert \leq O(\sqrt{t/N}).
    \end{align} 
\end{theorem}

\begin{proof}
We insert two additional projectors separately, which ensures that the
isometry transfers the ideal-world query answer from the database $D^3$ to
$D^1$ and $D^2$, 
\begin{align}
     &\lVert (\cco^R V - V \Pi^\star \cco^I) \Pi^\star \Pi^t \Pi^{k_1} \Pi^{k_3} \rVert \\
     &\leq \lVert (\cco^R V \Gamma_2 - V \Gamma_2 \Pi^\star \cco^I)\Gamma_1\Pi^\star \Pi^t \Pi^{k_1}  \Pi^{k_3} \rVert + O(\sqrt{t/N})\\
     &\leq \lVert (\cco^R V \Gamma_2 - V \Gamma_2 \cco^I)\Gamma_1\Pi^\star \Pi^t \Pi^{k_1}  \Pi^{k_3} \rVert + O(\sqrt{t/N}) 
\end{align} where the $\Gamma_1$ violation is bounded by \Cref{lem:gamma1-perp-bound} and the $\Gamma_2$ violation by \Cref{lem:gamma2-perp-bound}, each contributing $O(\sqrt{t/N})$. And in the third line we remove $\Pi^\star$ using $\Pi^\star = \Id - (\Pi^\star)^\perp$ and \Cref{lem:pi-star-violation} together with triangle inequality.

We continue by reducing $V$ to $\vp$
 \begin{align}
   &\lVert (\cco^R V \Gamma_2 - V \Gamma_2 \cco^I)\Gamma_1 \Pi^\star \Pi^t \Pi^{k_1}  \Pi^{k_3} \rVert \\
   &\leq \lVert (\cco^R \vpc\vp \Gamma_2 - \vpc\vp \Gamma_2 \cco^I)\Gamma_1\Pi^\star \Pi^t \Pi^{k_1}  \Pi^{k_3} \rVert \quad \text{(By \Cref{lem:p-vpvpc-V})}\\
   &\leq \lVert (\vpc \cco^R \vp \Gamma_2 - \vpc\vp \Gamma_2 \cco^I)\Gamma_1\Pi^\star \Pi^t \Pi^{k_1}  \Pi^{k_3} \rVert + O(\sqrt{t/N}) \quad \text{(By \Cref{lem:comm-real-vpc})}\\
   & \leq \lVert (\cco^R \vp \Gamma_2 - \vp \Gamma_2 \cco^I)\Gamma_1\Pi^\star \Pi^t \Pi^{k_1}  \Pi^{k_3} \rVert + O(\sqrt{t/N})
\end{align} where the last line follows from unitarity of the norm. We note that the bound on $\Gamma_2^{\perp}$ is the same for any state whose $K_3$ register holds a uniform superposition and whose $D^2$ register holds
a database of size at most $t$. In particular, it does not matter whether $\cco^{I}$ has been applied, so in the rest we do not pay for $\Gamma_2^{\perp}$ at each occurrence.

In what follows we fix the registers $X$ and $K_1$ in the computational
basis to $x$ and $k_1$, since all remaining operators are either controlled
on them or act trivially on them; by $\Gamma_1$ the register
$D^1_{x\oplus k_1}$ is then in $\bot$. We may therefore drop $\Pi^{k_1}$, as
the uniformity of $k_1$ was needed only for \Cref{lem:gamma1-perp-bound} and
plays no role below. Moreover, $\vp$ now has the single fill factor
$\vf\ppar{x\oplus k_1}$, whose target registers are already set to $\bot$ by
$\Gamma_1$ and $\Gamma_2$; we may therefore also drop the good and
injective database projectors, keeping only $\Pi^{inj}_{D^2}$ for the
flip, and still fix the set of admissible elements $J$ of \Cref{def:J} for
each computational basis state. Hence the hybrid difference to be bounded
becomes
\begin{align}
    \norm{(\cco^R \vp \Gamma_2 - \vp \Gamma_2 \cco^I) \Gamma_1\Pi^{inj}_{D^2} \Pi^{k_3}\Pi^t} .\label{eq:p-hybrids-with-vp} 
\end{align} 

We fix the registers $K_2$ to $k_2\in[N]$ and registers $(D^1_{x\oplus k_1})^c, (D^3_x)^c$ to a database of size at most $t$ in the computational basis (while noting that $\abs{\dom{(D^3_x)^c}}\leq t-1 $ if $ D^3_x\neq \bot$) since the operators preserve their orthogonality. We analyze \eqref{eq:p-hybrids-with-vp} with the registers $D^3_x$ and $Y$ in Hadamard basis states, where the effect of $\cco^{I}$, whether it defines or uncomputes a database entry, is easier to classify. Take the Hadamard basis $\set{\ket{\hat{y}}}_{y\in [N]}$ where $\ket{\hat{y}}= H^{\otimes n} \ket{y}$ for $y\in \set{0,1}^n$.\footnote{When $\ket{\hat{\bot}}$ appears, it is interpreted as $\ket{\bot}$, by extending $H^{\otimes n }$ to act as the identity on $\ket{\bot}$.} We write the general state as, omitting the fixed and untouched registers,
\begin{align}
    \ket{\psi}\coloneqq \ket{+^n}_{K_3}\otimes \sum_{y\in [N],y'\in [N]\cup \set{\bot},D\in \mathbf D_{inj, t}} \gamma_{y,y',D} \ket{\bot}_{D^1_{x\oplus k_1}} \ket{\hat{y}}_Y \ket{\hat{y'}}_{D^3_x}\ket{D}_{D^2}\ket{\psi_{y,y',D}},
\end{align}
where $\ket{\psi_{y,y',D}}$ are subnormalized, then write $\ket{\psi} = \sum \ket{\psi_i}$ for $i\in\set{1,2,3,4,5}$, partitioning the cases with respect to the effect of the ideal world query on the database register $D^{3}_x$ as
\begin{enumerate}[label={\ensuremath{i=\arabic*}:}]
  \item $\bot \mapsto \bot$, which occurs when $Y$ holds $\ket{\hat 0}$; 
    bounded in \eqref{eq:p-psi1};
  \item non-$\bot$ and non-$\ket{\hat 0}$ $\mapsto \bot$, which occurs when
    $Y$ and $D^3_x$ hold the same Hadamard basis state; bounded in
    \eqref{eq:p-psi2};
  \item $\bot \mapsto$ non-$\bot$, which occurs when $Y$ holds a Hadamard
    basis state other than $\ket{\hat 0}$; bounded in \eqref{eq:p-psi3};
  \item non-$\bot$ $\mapsto$ non-$\bot$, which occurs when $Y$ and $D^3_x$
    hold distinct Hadamard basis states; bounded in \eqref{eq:p-psi4};
  \item $\ket{\hat 0} \mapsto \ket{\hat 0}$, the case in which $D^3_x$ is
    in $\ket{\hat 0}$ throughout; bounded in \eqref{eq:p-psi5}.
\end{enumerate}
We go through each case separately and conclude the result using triangle inequality.
\begin{enumerate}[leftmargin=0pt, labelindent=0pt, itemindent=*, label={\ensuremath{(i=\arabic*)}:}]
    \item After being queried, $D^3_x$ changes from $\ket{\bot}$ to $\ket{\bot}$ if and only if $Y$ register is in $\ket{\hat{0}}$ and $D^3_x$ register is in $\ket{\bot}$. Define \begin{align}\label{eq:p-psi1}
    \ket{\psi_1}\coloneqq \ket{+^n}_{K_3}\otimes \sum_{D\in \mathbf D_{inj, t}} \gamma_{0,\bot,D} \ket{\bot}_{D^1_{x\oplus k_1}} \ket{\hat{0}}_Y \ket{\bot}_{D^3_x}\ket{D}_{D^2}\ket{\psi_{0,\bot,D}},
    \end{align} and bound the difference between the hybrids in \eqref{eq:p-hybrids-with-vp} when applied to $\ket{\psi_1}$. To this end, observe that when $Y$ register is in $\ket{\hat{0}}$ the query operators $\cco^{R}$ and $\cco^{I}$ act as identity on any state and also $\vp$ does not act on $Y$, so we can write
    \begin{align}
        \norm{(\cco^{R} \vp\Gamma_2- \vp\Gamma_2 \cco^{I})\ket{\psi_1}} &\leq \norm{(\cco^{R}\otimes \ket{\hat{0}}_Y) \vp\Gamma_2-\vp\Gamma_2 (\cco^{I}\otimes \ket{\hat{0}}_Y)}\\
        &\leq \norm{(\ket{\hat{0}}_Y \otimes\vp-\ket{\hat{0}}_Y \otimes\vp)\Gamma_2}\\
        &=0.
    \end{align}
    \item After being queried, $D^3_x$ changes from non-$\ket{\bot}$ (excluding $\ket{\hat{0}}$) to $\ket{\bot}$ if and only if both $Y$ and $D^3_x$ register is in $\ket{\hat{y}}$. Define
    \begin{align}\label{eq:p-psi2}
        \ket{\psi_2}\coloneqq \ket{+^n}_{K_3}\otimes \sum_{y\in [N]\setminus\set{0},D\in \mathbf D_{inj, t}} \gamma_{y,y,D} \ket{\bot}_{D^1_{x\oplus k_1}} \ket{\hat{y}}_Y \ket{\hat{y}}_{D^3_x}\ket{D}_{D^2}\ket{\psi_{y,y,D}},
    \end{align} and bound the difference between the hybrids in \eqref{eq:p-hybrids-with-vp} when applied to $\ket{\psi_2}$. For each $D$, we denote the associated set of admissable values in \Cref{def:J} by $J_D$. We modify and bound the difference of the hybrids as follows
    \begin{align}
        &\norm{(\cco^{R} \vp\Gamma_2- \vp\Gamma_2 \cco^{I})\ket{\psi_2}}\\
        &\leq \norm{(\fc\ppar{1}\fc\ppar{2}\pu^{em}\vp\Gamma_2- \vp\Gamma_2 \cco^{I})\ket{\psi_2} } + \norm{{(\Id - \fc^{(2)}\fc\ppar{1}) \vp \Gamma_2 \ket{\psi_2}}}\\
        &\leq \norm{(\pu^{em}\vp\Gamma_2- \fc\ppar{2}\fc\ppar{1}\vp\Gamma_2 \cco^{I})\ket{\psi_2} } + O(\sqrt{t/N}). \quad \text{(By \Cref{lem:p3-decomp-divergence})} \label{eq:modified-bound-on-hybrid-psi2}
    \end{align}
Recalling that $\vp=\vr\ppar{x\oplus k_1}\fp_{D^2}\vf\ppar{x\oplus k_1}\fp_{D^2}$ and $\ket{\hat{y}}=\sum(-1)^{y.z}\ket{z}$, $z\in[N]$, where $y\cdot z$ denotes the inner product of the $n$-bit strings $y$ and $z$ modulo $2$, we bound the first term in \eqref{eq:modified-bound-on-hybrid-psi2} by writing each term separately; the first one is \begin{align}
&\pu^{em}\vp \Gamma_2 \ket{\psi_2}\\
  &= \pu^{em}\vp \Biggl(\sum_{y\in [N]\setminus\set{0}} \ket{\hat{y}}_Y 
  \frac{1}{\sqrt{N}}\sum_{z\in [N]} (-1)^{y.z} \ket{z}_{D^3_x}\ket{\bot}_{D^1_{x\oplus k_1}}\sum_{\substack{D\in \mathbf D_{inj, t}}} \gamma_{y,y,D} \ket{D}_{D^2}\ket{\psi_{y,y,D}} \frac{1}{\sqrt{N}}\sum_{\substack{k_3\in[N]:\\
  z\oplus k_3 \not \in \im{D}}} \ket{k_3}_{K_3}\Biggr)\\
  &= \pu^{em}\vr\ppar{x\oplus k_1} \fp_{D^2} \Biggl( \sum_{y\in [N]\setminus\set{0}} \ket{\hat{y}}_Y 
  \frac{1}{\sqrt{N}}\sum_{z\in [N]} (-1)^{y.z} \ket{z}_{D^3_x}\sum_{\substack{D\in \mathbf D_{inj, t}}} \gamma_{y,y,D} \ket{\psi_{y,y,D}} \frac{1}{\sqrt{N}} \nonumber\\
  &\quad \cdot\sum_{\substack{k_3\in[N]:\\
  z\oplus k_3 \not \in \im{D}}}\ket{k_3}_{K_3}\ket{((D^2
  _{z\oplus k_3})^c)^{-1}}_{(D^2_{z\oplus k_3})^c} \frac{1}{\sqrt{\abs{J_D}}} \sum_{u\in J_D} \ket{u}_{D^1_{x\oplus k_1}} \ket{u \oplus k_2}_{D^2_{z\oplus k_3}}\Biggr)\\
  &=\pu^{em}\Biggl(\sum_{y\in [N]\setminus\set{0}} \ket{\hat{y}}_Y \ket{\bot}_{D^3_x}
  \frac{1}{\sqrt{N}}\sum_{z\in [N]}  \sum_{\substack{D\in \mathbf D_{inj, t}}} \gamma_{y,y,D} \ket{\psi_{y,y,D}}\frac{1}{\sqrt{N}} \sum_{\substack{k_3\in[N]:\\
  z\oplus k_3 \not \in \im{D}}}\ket{k_3}_{K_3}\nonumber\\
  &\quad \cdot \frac{1}{\sqrt{\abs{J_D}}} \sum_{u\in J_D} (-1)^{y.z} \ket{u}_{D^1_{x\oplus k_1}}   \ket{z\oplus k_3}_{D^2_{u \oplus k_2}}\ket{(D_{u \oplus k_2})^c}_{(D^2_{u \oplus k_2})^c}\Biggr) \label{eq:vp-gamma2-psi2-reference}\\
  &=\sum_{y\in [N]\setminus\set{0}}  \ket{\bot}_{D^3_x}
  \frac{1}{\sqrt{N}}\sum_{z\in [N]}  \sum_{\substack{D\in \mathbf D_{inj, t}}} \gamma_{y,y,D} \ket{\psi_{y,y,D}} \frac{1}{\sqrt{N}} \sum_{\substack{k_3\in[N]:\\
  z\oplus k_3 \not \in \im{D}}}\ket{k_3}_{K_3} \frac{1}{\sqrt{\abs{J_D}}} \sum_{u\in J_D} \ket{u}_{D^1_{x\oplus k_1}} \nonumber\\
  &\quad \otimes\ket{(D_{u \oplus k_2})^c}_{(D^2_{u \oplus k_2})^c}  \frac{1}{\sqrt{N}}\sum_{v\in[N]}(-1)^{y.v}(-1)^{y.z} \ket{v\oplus(z\oplus k_3)\oplus k_3}_Y \ket{z\oplus k_3}_{D^2_{u \oplus k_2}}\\
  &=\sum_{y\in [N]\setminus\set{0}}  \ket{\hat{y}}_Y \ket{\bot}_{D^3_x}
  \frac{1}{\sqrt{N}}\sum_{z\in [N]}  \sum_{\substack{D\in \mathbf D_{inj, t}}} \gamma_{y,y,D} \ket{\psi_{y,y,D}} \frac{1}{\sqrt{N}} \sum_{\substack{k_3\in[N]:\\
  z\oplus k_3 \not \in \im{D}}}\ket{k_3}_{K_3} \frac{1}{\sqrt{\abs{J_D}}} \sum_{u\in J_D} \ket{u}_{D^1_{x\oplus k_1}} \nonumber\\
  &\quad \otimes  \ket{z\oplus k_3}_{D^2_{u \oplus k_2}} \ket{(D_{u \oplus k_2})^c}_{(D^2_{u \oplus k_2})^c}, \label{eq:psi2-real-term}
    \end{align}
    and the second one is \begin{align}
     &\fc\ppar{2}\fc\ppar{1}\vp\Gamma_2 \cco^{I}\ket{\psi_2}\\
     &=\fc\ppar{2}\fc\ppar{1} \vp\Gamma_2 \Biggl(\sum_{y\in [N]\setminus\set{0}}  \ket{\hat{y}}_Y \ket{\bot}_{D^3_x}\ket{\bot}_{D^1_{x\oplus k_1}}\ket{+^n}_{K_3}\sum_{D\in \mathbf D_{inj, t}} \gamma_{y,y,D} \ket{D}_{D^2} \ket{\psi_{y,y,D}} \Biggr)\\
     &= \fc\ppar{2}\fc\ppar{1}\Biggl( \sum_{y\in [N]\setminus\set{0}}  \ket{\hat{y}}_Y \ket{\bot}_{D^3_x}\ket{\bot}_{D^1_{x\oplus k_1}}\ket{+^n}_{K_3}\sum_{D\in \mathbf D_{inj, t}} \gamma_{y,y,D} \ket{D}_{D^2} \ket{\psi_{y,y,D}} \Biggr)\\
     &=\fc\ppar{2}\Biggl( \sum_{y\in [N]\setminus\set{0}}  \ket{\hat{y}}_Y \ket{\bot}_{D^3_x}\ket{+^n}_{K_3}\sum_{D\in \mathbf D_{inj, t}} \gamma_{y,y,D} \ket{D}_{D^2} \ket{\psi_{y,y,D}}  \frac{1}{\sqrt{N}}\sum_{u\in[N]}\ket{u}_{D^1_{x\oplus k_1}}\Biggr)\\
     &=\fc\ppar{2}\Biggl( \sum_{y\in [N]\setminus\set{0}}  \ket{\hat{y}}_Y \ket{\bot}_{D^3_x}\ket{+^n}_{K_3}\sum_{D\in \mathbf D_{inj, t}} \gamma_{y,y,D} \ket{D}_{D^2} \ket{\psi_{y,y,D}}  \frac{1}{\sqrt{N}}\sum_{u\in J_D}\ket{u}_{D^1_{x\oplus k_1}}\Biggr) + E\\
     &=\sum_{y\in [N]\setminus\set{0}}  \ket{\hat{y}}_Y \ket{\bot}_{D^3_x}\ket{+^n}_{K_3}\sum_{D\in \mathbf D_{inj, t}} \gamma_{y,y,D}\ket{\psi_{y,y,D}} \frac{1}{\sqrt{N}}\sum_{u\in J_D}\ket{u}_{D^1_{x\oplus k_1}} \frac{1}{\sqrt{N}}\sum_{z\in[N]}\ket{z}_{D^2_{u\oplus k_2}} \nonumber\\
     &\quad\otimes\ket{(D_{u \oplus k_2})^c}_{(D^2_{u \oplus k_2})^c}+E \\
     &=\sum_{y\in [N]\setminus\set{0}}  \ket{\hat{y}}_Y \ket{\bot}_{D^3_x}\frac{1}{\sqrt{N}}\sum_{z\in[N]}\sum_{D\in \mathbf D_{inj, t}} \gamma_{y,y,D}\ket{\psi_{y,y,D}} \frac{1}{\sqrt{N}}\sum_{k_3\in[N]}\ket{k_3}_{K_3} \frac{1}{\sqrt{N}}\sum_{u\in J_D}\ket{u}_{D^1_{x\oplus k_1}}  \nonumber\\
     &\quad \otimes\ket{z\oplus k_3}_{D^2_{u\oplus k_2}}\ket{(D_{u \oplus k_2})^c}_{(D^2_{u \oplus k_2})^c}+E\\
        &\eqqcolon \ket{\phi} + E
    \end{align} where \begin{align}
        E\coloneqq \fc\ppar{2} \sum_{y\in [N]\setminus\set{0}}  \ket{\hat{y}}_Y \ket{\bot}_{D^3_x}\ket{+^n}_{K_3}\sum_{D\in \mathbf D_{inj, t}} \gamma_{y,y,D} \ket{D}_{D^2} \ket{\psi_{y,y,D}}  \frac{1}{\sqrt{N}}\sum_{u\not\in J_D}\ket{u}_{D^1_{x\oplus k_1}}.
    \end{align}
    Hence, by defining and using the following term as a stepping stone \begin{align}
       \ket{\phi\ppar{k_3}}&\coloneqq \sum_{y\in [N]\setminus\set{0}}  \ket{\hat{y}}_Y \ket{\bot}_{D^3_x}\frac{1}{\sqrt{N}}\sum_{z\in[N]}\sum_{D\in \mathbf D_{inj, t}} \gamma_{y,y,D}\ket{\psi_{y,y,D}} \frac{1}{\sqrt{N}}\sum_{\substack{k_3\in[N]:\\
       z\oplus k_3 \not\in \im{D}}}\ket{k_3}_{K_3} \frac{1}{\sqrt{N}}\nonumber\\
     &\quad \cdot\sum_{u\in J_D}\ket{u}_{D^1_{x\oplus k_1}}  \ket{z\oplus k_3}_{D^2_{u\oplus k_2}}\ket{(D_{u \oplus k_2})^c}_{(D^2_{u \oplus k_2})^c}.
    \end{align}In words, $\ket{\phi\ppar{k_3}}$ is the same as $\ket{\phi}$ except that the index for $k_3$ is the complement set in order to restore the independence between $z$ and $k_3$. We can express the difference as \begin{align}
       &\norm{(\pu^{em}\vp\Gamma_2- \fc\ppar{2}\fc\ppar{1}\vp\Gamma_2 \cco^{I})\ket{\psi_2} }\\
       &\leq \norm{\ket{\phi}-\ket{\phi\ppar{k_3}}} +\norm{\ket{\phi\ppar{k_3}}-\eqref{eq:psi2-real-term}}+ \norm{E} \\
       &=\Big\lVert \sum_{y\in [N]\setminus\set{0}}  \ket{\hat{y}}_Y \ket{\bot}_{D^3_x}\frac{1}{\sqrt{N}}\sum_{z\in[N]}\sum_{D\in \mathbf D_{inj, t}} \gamma_{y,y,D}\ket{\psi_{y,y,D}} \frac{1}{\sqrt{N}}\sum_{\substack{k_3\in[N]:\\
       z\oplus k_3 \in \im{D}}}\ket{k_3}_{K_3} \frac{1}{\sqrt{N}}\nonumber\\
     &\quad \cdot\sum_{u\in J_D}\ket{u}_{D^1_{x\oplus k_1}}  \ket{z\oplus k_3}_{D^2_{u\oplus k_2}}\ket{(D_{u \oplus k_2})^c}_{(D^2_{u \oplus k_2})^c} \Big\rVert + \Big\lVert \sum_{y\in [N]\setminus\set{0}}  \ket{\hat{y}}_Y \ket{\bot}_{D^3_x}\frac{1}{\sqrt{N}}\sum_{z\in[N]}\nonumber\\
     &\quad\sum_{D\in \mathbf D_{inj, t}} \gamma_{y,y,D}\ket{\psi_{y,y,D}} \frac{1}{\sqrt{N}}\sum_{\substack{k_3\in[N]:\\
       z\oplus k_3 \not\in \im{D}}}\ket{k_3}_{K_3} \Bigl(\frac{1}{\sqrt{\abs{J_D}}} - \frac{1}{\sqrt{N}}\Bigr) \sum_{u\in J_D}\ket{u}_{D^1_{x\oplus k_1}}  \ket{z\oplus k_3}_{D^2_{u\oplus k_2}}\nonumber\\
     &\quad\ket{(D_{u \oplus k_2})^c}_{(D^2_{u \oplus k_2})^c} \Big\rVert + \norm{E}\\
       &\leq \sqrt{ \sum_{y\in [N]\setminus\set{0}} \frac{1}{N} \sum_{z\in[N]}\sum_{D\in \mathbf D_{inj, t}} \abs{\gamma_{y,y,D}}^2 \frac{1}{{N}}\sum_{\substack{k_3\in[N]:\\
       z\oplus k_3 \in \im{D}}  }\frac{1}{N}\sum_{u\in J_D}\norm{\ket{\psi_{y,y,D}}^2}} \nonumber\\
       &\quad + \sqrt{ \sum_{y\in [N]\setminus\set{0}} \frac{1}{N} \sum_{z\in[N]}\sum_{D\in \mathbf D_{inj, t}} \abs{\gamma_{y,y,D}}^2 \frac{1}{{N}}\sum_{\substack{k_3\in[N]:\\
       z\oplus k_3 \not\in \im{D}} }\sum_{u\in J_D} \Big(\frac{1}{\sqrt{\abs{J_D}}}-\frac{1}{\sqrt{N}}\Big)^2\norm{\ket{\psi_{y,y,D}}^2}} \nonumber\\
       &\quad+\sqrt{ \sum_{y\in [N]\setminus\set{0}}  \frac{1}{N} \sum_{z\in[N]}\sum_{D\in \mathbf D_{inj, t}} \abs{\gamma_{y,y,D}}^2 \sum_{u\not\in J_D} \norm{\ket{\psi_{y,y,D}}^2}}\\
       &\leq O(\sqrt{t/N}) \label{eq:p-psi2-hybrids-bound-one}
    \end{align} using that the orthogonality of both the $y$'s and the $D$'s is preserved, that the $\ket{\psi_{y,y,D}}$ are subnormalized, and that $\abs{J_D}\geq N-2t$ in all terms. For the first term, orthogonal $D$ and $D'$ yield orthogonal images on the joint state of $D^1_{x\oplus k_1}$ and $D^2$. To see this, suppose the two branches hold the same value $u$ at $D^1_{x\oplus k_1}$; then $D$ and $D'$ both had $\bot$ at the index $u\oplus k_2$, and after the fill both hold a non-$\bot$ value there. Since $D$ and $D'$ were orthogonal to begin with, they must differ at some other index, and they remain orthogonal there. If instead the two branches hold different values at $D^1_{x\oplus k_1}$, they are already orthogonal on that register. The same argument applies to the second term. For the third term, namely $E$, one can remove $\fc\ppar{2}$ by unitarity of the norm, and then orthogonality of both the $y$'s and the $D$'s is preserved clearly.

\item After being queried, $D^3_x$ changes from $\ket{\bot}$ to non-$\ket{\bot}$ if and only if $Y$ register is in $\ket{\hat{y}}$ s.t. $\ket{\hat{y}}\neq \ket{\hat{0}}$ and $D^3_x$ register is in $\ket{\bot}$. Define
    \begin{align}\label{eq:p-psi3}
        \ket{\psi_3}\coloneqq \ket{+^n}_{K_3} \otimes \sum_{y\in [N]\setminus\set{0}} \ket{\hat{y}}_Y \ket{\bot}_{D^3_x}\ket{\bot}_{D^1_{x\oplus k_1}}\sum_{D\in \mathbf D_{inj, t}} \gamma_{y,\bot,D} \ket{D}_{D^2} \ket{\psi_{y,\bot,D}}
    \end{align} and bound the difference between the hybrids in \eqref{eq:p-hybrids-with-vp} when applied to $\ket{\psi_3}$.

We modify and bound the difference of the hybrids as follows
    \begin{align}
        &\norm{(\cco^{R} \vp\Gamma_2- \vp\Gamma_2 \cco^{I})\ket{\psi_3}}\\
        &=\norm{(\pu^{em}\fc\ppar{2}\fc\ppar{1}\vp\Gamma_2- \fc\ppar{2}\fc\ppar{1}\vp\Gamma_2 \cco^{I})\ket{\psi_3} } \\
&\leq\norm{(\pu^{em}\fc\ppar{2}\fc\ppar{1}\vp\Gamma_2- \vp\Gamma_2 \cco^{I})\ket{\psi_3} }+ \norm{{(\Id - \fc^{(2)}\fc\ppar{1}) \vp \Gamma_2 \cco^{I}\ket{\psi_3}}}\\
&=\norm{(\fc\ppar{2}\fc\ppar{1}\vp\Gamma_2-\pu^{em} \vp\Gamma_2 \cco^{I})\ket{\psi_3} }+ \norm{{(\Id - \fc^{(2)}\fc\ppar{1}) \vp \Gamma_2 \cco^{I}\ket{\psi_3}}}\\
&=\norm{\fc\ppar{2}\fc\ppar{1}\ket{\psi_3}-\pu^{em} \vp\Gamma_2 \cco^{I}\ket{\psi_3} }+ \norm{{(\Id - \fc^{(2)}\fc\ppar{1}) \vp \Gamma_2 \cco^{I}\ket{\psi_3}}}\label{eq:modified-bound-on-hybrid-psi3}
\end{align}

Observe that setting $\gamma_{y,y,D}$ to $\gamma_{y,\bot,D}$ in \eqref{eq:p-psi2}, we get $\cco^{I}\ket{\psi_3}=\ket{\psi_2}$, equivalently, $\ket{\psi_3}=\cco^{I}\ket{\psi_2}$ (meaning double query on the same input uncomputes the database entry). Replacing this in \eqref{eq:modified-bound-on-hybrid-psi3} gives \begin{align}
    &\norm{(\cco^{R} \vp\Gamma_2- \vp\Gamma_2 \cco^{I})\ket{\psi_3}}\\
&\leq \norm{\fc\ppar{2}\fc\ppar{1}\cco^{I}\ket{\psi_2}-\pu^{em} \vp\Gamma_2 \ket{\psi_2} }+ \norm{{(\Id - \fc^{(2)}\fc\ppar{1}) \vp \Gamma_2 \ket{\psi_2}}}\\
&\leq O(\sqrt{t/N}) 
\end{align} where in the last line the first term follows from \eqref{eq:p-psi2-hybrids-bound-one} and the second follows from \Cref{lem:p3-decomp-divergence}. 
   
\item After being queried, $D^3_x$ changes from non-$\ket{\bot}$ (excluding $\ket{\hat{0}}$) to non-$\ket{\bot}$ if and only if $D^3_x$ register is in non-$\ket{\bot}$, say $\ket{\hat{y'}}$ and $Y$ register is in $\ket{\hat{y'}}$ s.t. $\ket{\hat{y'}}\neq \ket{\hat{y}}$. Define
    \begin{align}\label{eq:p-psi4}
        \ket{\psi_4}\coloneqq \sum_{y\in[N],y'\in [N]\setminus\set{0}:y\neq y'}  \ket{\hat{y}}_Y \ket{\hat{y'}}_{D^3_x}\ket{\bot}_{D^1_{x\oplus k_1}}\ket{+^n}_{K_3}\sum_{D\in \mathbf D_{inj, t}} \gamma_{y,y',D} \ket{D}_{D^2} \ket{\psi_{y,y',D}}
    \end{align} and bound the difference between the hybrids in \eqref{eq:p-hybrids-with-vp} when applied to $\ket{\psi_4}$.

 We modify and bound the difference of the hybrids as follows
\begin{align}
    &\norm{(\cco^{R} \vp\Gamma_2- \vp\Gamma_2 \cco^{I})\ket{\psi_4}}\\
    &= \norm{(\pu^{em}\fc\ppar{2}\fc\ppar{1} \vp\Gamma_2- \fc\ppar{2}\fc\ppar{1}\vp\Gamma_2 \cco^{I})\ket{\psi_4}}\\
    &= \norm{(\pu^{em}\fc\ppar{2}\fc\ppar{1} \vp\Gamma_2- \vp\Gamma_2 \cco^{I})\ket{\psi_4}} + \norm{(\Id-\fc\ppar{2}\fc\ppar{1})\vp\Gamma_2 \cco^{I}\ket{\psi_4}} \\
    &= \norm{(\pu^{em}\vp\Gamma_2- \vp\Gamma_2 \cco^{I})\ket{\psi_4}} + \norm{(\Id-\fc\ppar{2}\fc\ppar{1})\vp\Gamma_2 \cco^{I}\ket{\psi_4}}\nonumber\\
    &\quad + \norm{(\Id-\fc\ppar{2}\fc\ppar{1}) \vp\Gamma_2\ket{\psi_4}} \label{eq:p-hybrids-psi4}\\
    &\leq \norm{(\pu^{em}\vp\Gamma_2- \vp\Gamma_2 \cco^{I})\ket{\psi_4}} + O(\sqrt{t/N})
\end{align} where the last line follows by \Cref{lem:p3-decomp-divergence}. For the last term it is clear that the lemma applies. For the middle term, notice that \begin{align}
    \cco^{I}\ket{\psi_4}=\ket{+^n}_{K_3}\otimes\sum_{y\in[N],y'\in [N]\setminus\set{0}:y\neq y'}  \ket{\hat{y}}_Y \ket{\widehat{y\oplus y'}}_{D^3_x}\ket{\bot}_{D^1_{x\oplus k_1}}\sum_{D\in \mathbf D_{inj, t}} \gamma_{y,y',D} \ket{D}_{D^2} \ket{\psi_{y,y',D}} \label{eq:p3-co-r-psi4}
\end{align} so that the joint state of the registers $Y$ and $D^3_x$ remains orthogonal
across distinct pairs $(y,y')$. We bound the remaining term in \eqref{eq:p-hybrids-psi4} by writing the each term in the norm separately, the first one is written by adapting \eqref{eq:vp-gamma2-psi2-reference} for $\ket{\psi_4}$\footnote{The adaptation consists only of changing the coefficients set by the index of the first sum from $\gamma_{y,y,D}$ to $\gamma_{y,y',D}$, and denoting the queried database entry by $\hat{y'}$ rather than $\hat{y}$.}, \begin{align}
    &\pu^{em}\vp\Gamma_2 \ket{\psi_4}\\
&=\pu^{em}\sum_{\substack{y\in[N],y'\in [N]\setminus\set{0}:\\y\neq y'}}  \ket{\hat{y}}_Y \ket{\bot}_{D^3_x}
  \frac{1}{\sqrt{N}}\sum_{z\in [N]}  \sum_{\substack{D\in \mathbf D_{inj, t}}} \gamma_{y,y',D}\ket{\psi_{y\in[N],y',D}} \frac{1}{\sqrt{N}} \sum_{\substack{k_3\in[N]:\\
  z\oplus k_3 \not \in \im{D}}}\ket{k_3}_{K_3}\nonumber\\
  &\quad \cdot \frac{1}{\sqrt{\abs{J_D}}} \sum_{u\in J_D} (-1)^{y'.z}\ket{u}_{D^1_{x\oplus k_1}}   \ket{z\oplus k_3}_{D^2_{u \oplus k_2}}\ket{(D_{u \oplus k_2})^c}_{(D^2_{u \oplus k_2})^c}\\
  &=\sum_{\substack{y\in[N],y'\in [N]\setminus\set{0}:\\y\neq y'}}   \ket{\bot}_{D^3_x}
  \frac{1}{\sqrt{N}}\sum_{z\in [N]}  \sum_{\substack{D\in \mathbf D_{inj, t}}} \gamma_{y,y',D}\ket{\psi_{y,y',D}} \frac{1}{\sqrt{N}} \sum_{\substack{k_3\in[N]:\\
  z\oplus k_3 \not \in \im{D}}}\ket{k_3}_{K_3} \frac{1}{\sqrt{\abs{J_D}}} \sum_{u\in J_D} \ket{u}_{D^1_{x\oplus k_1}}\nonumber\\
  &\quad \otimes \ket{(D_{u \oplus k_2})^c}_{(D^2_{u \oplus k_2})^c} \frac{1}{\sqrt{N}}\sum_{v\in[N]}(-1)^{y.v} \ket{v\oplus(z\oplus k_3)\oplus k_3}_Y (-1)^{y'.z}\ket{z\oplus k_3}_{D^2_{u \oplus k_2}}\\
  &=\sum_{\substack{y\in[N],y'\in [N]\setminus\set{0}:\\y\neq y'}}   \ket{\bot}_{D^3_x}
  \frac{1}{\sqrt{N}}\sum_{z\in [N]}  \sum_{\substack{D\in \mathbf D_{inj, t}}} \gamma_{y,y',D} \ket{\psi_{y,y',D}}\frac{1}{\sqrt{N}} \sum_{\substack{k_3\in[N]:\\
  z\oplus k_3 \not \in \im{D}}}\ket{k_3}_{K_3} \frac{1}{\sqrt{\abs{J_D}}} \sum_{u\in J_D} \ket{u}_{D^1_{x\oplus k_1}} \nonumber\\
  &\quad \otimes\ket{(D_{u \oplus k_2})^c}_{(D^2_{u \oplus k_2})^c} \frac{1}{\sqrt{N}}\sum_{v\in[N]}(-1)^{y.(v\oplus z)} \ket{v\oplus z}_Y (-1)^{(y'\oplus y).z}\ket{z\oplus k_3}_{D^2_{u \oplus k_2}}\\
  &=\sum_{\substack{y\in[N],y'\in [N]\setminus\set{0}:\\y\neq y'}}   \ket{\hat{y}}_Y\ket{\bot}_{D^3_x}
  \frac{1}{\sqrt{N}}\sum_{z\in [N]}  \sum_{\substack{D\in \mathbf D_{inj, t}}} \gamma_{y,y',D}\ket{\psi_{y,y',D}} \frac{1}{\sqrt{N}} \sum_{\substack{k_3\in[N]:\\
  z\oplus k_3 \not \in \im{D}}}\ket{k_3}_{K_3}\nonumber\\
  &\quad \cdot \frac{1}{\sqrt{\abs{J_D}}} \sum_{u\in J_D}  (-1)^{(y'\oplus y).z}\ket{u}_{D^1_{x\oplus k_1}} \ket{z\oplus k_3}_{D^2_{u \oplus k_2}} \ket{(D_{u \oplus k_2})^c}_{(D^2_{u \oplus k_2})^c}, 
\end{align}and similarly, by adapting \eqref{eq:vp-gamma2-psi2-reference} for $\cco^I\ket{\psi_4}$ given in $\eqref{eq:p3-co-r-psi4}$, the second term is \begin{align}
 \vp \Gamma_2 \cco^{I}\ket{\psi_4} 
    &=\sum_{\substack{y\in[N],y'\in [N]\setminus\set{0}:\\y\neq y'}}  \ket{\hat{y}}_Y \ket{\bot}_{D^3_x}
  \frac{1}{\sqrt{N}}\sum_{z\in [N]}  \sum_{\substack{D\in \mathbf D_{inj, t}}} \gamma_{y,y',D} \frac{1}{\sqrt{N}} \sum_{\substack{k_3\in[N]:\\
  z\oplus k_3 \not \in \im{D}}}\ket{k_3}_{K_3}\nonumber\\
  &\quad \cdot \frac{1}{\sqrt{\abs{J_D}}} \sum_{u\in J_D} (-1)^{(y'\oplus y).z}\ket{u}_{D^1_{x\oplus k_1}}  \ket{z\oplus k_3}_{D^2_{u \oplus k_2}}\ket{(D_{u \oplus k_2})^c}_{(D^2_{u \oplus k_2})^c}.
\end{align}
Hence, their difference becomes zero, i.e., \begin{align}
    \norm{(\pu^{em}\vp\Gamma_2- \vp\Gamma_2 \cco^{I})\ket{\psi_4}}=0.
\end{align}
\item  Assume $D^3_x=\ket{\hat{0}}$, and define \begin{align}\label{eq:p-psi5}
     \ket{\psi_5}\coloneqq \sum_{y\in [N]}\ket{\hat{y}}_Y \ket{\hat{0}}_{D^3_x}\ket{\bot}_{D^1_{x\oplus k_1}}\ket{+^n}_{K_3}\sum_{D\in \mathbf D_{inj, t}} \gamma_{y,0,D} \ket{D}_{D^2} \ket{\psi_{y,0,D}}
    \end{align} and bound the difference between the hybrids in \eqref{eq:p-hybrids-with-vp} when applied to $\ket{\psi_5}$. We modify and bound the difference of the hybrids as follows
    \begin{align}
        &\norm{(\cco^{R}\vp\Gamma_2 - \vp\Gamma_2\cco^{I}) \ket{\psi_5}}\\
&=\lVert(\pu^{em}\fc\ppar{2}\fc\ppar{1}\vp\Gamma_2- \fc\ppar{2}\fc\ppar{1}\vp\Gamma_2\cco^{I})\ket{\psi_5}\rVert\\
&=\lVert(\pu^{em}-\Id)\fc\ppar{2}\fc\ppar{1}\vp\Gamma_2\ket{\psi_5}\rVert \quad\text{(Since $\cco^{I}\ket{\psi_5}=\ket{\psi_5}$)}\nonumber\\
&=\norm{(\pu^{em}-\Id)\fc\ppar{2}\fc\ppar{1}\vp\Gamma_2\ket{\psi_5}\pm (\pu^{em}-\Id)\fc\ppar{2}\vp\Gamma_2\ket{\psi_5}}\\
&\leq \norm{(\pu^{em}-\Id)\fc\ppar{2}\vp\Gamma_2\ket{\psi_5}} +\norm{ (\pu^{em}-\Id)\fc\ppar{2}(\fc\ppar{1}-\Id)\vp\Gamma_2\ket{\psi_5}}\\
&\leq \norm{(\pu^{em}-\Id)\fc\ppar{2}\vp\Gamma_2\ket{\psi_5}} +2 \norm{ (\fc\ppar{1}-\Id)\vp\Gamma_2\ket{\psi_5}}\\
& \leq \norm{(\pu^{em}-\Id)\fc\ppar{2}\vp\Gamma_2\ket{\psi_5}\pm (\pu^{em}-\Id)\fc\ppar{2}\ket{\phi^{k_3}}} + O(\sqrt{t/N})\quad \text{(By \Cref{lem:p3-decomp-divergence})}\\
& \leq \norm{(\pu^{em}-\Id)\fc\ppar{2}(\vp\Gamma_2\ket{\psi_5} + \ket{\phi^{k_3}})} + \norm{ (\pu^{em}-\Id)\fc\ppar{2}\ket{\phi^{k_3}} }+2 \norm{ (\fc\ppar{1}-\Id)\vp\Gamma_2\ket{\psi_5}}\\
& \leq \norm{(\pu^{em}-\Id)\fc\ppar{2}(\vp\Gamma_2\ket{\psi_5} + \ket{\phi^{k_3}})} + 2\norm{\ket{\phi^{k_3}} } +2 \norm{ (\fc\ppar{1}-\Id)\vp\Gamma_2\ket{\psi_5}}
\label{eq:p-psi5-split}
    \end{align} where \begin{align}
        \ket{\phi^{k_3}}\coloneqq &\sum_{y\in [N]} \ket{\hat{y}}_Y \ket{\bot}_{D^3_x}
  \frac{1}{\sqrt{N}}\sum_{z\in [N]}  \sum_{\substack{D\in \mathbf D_{inj, t}}} \gamma_{y,0,D} \ket{\psi_{y,0,D}}\frac{1}{\sqrt{N}} \sum_{\substack{k_3\in[N]:\\
  z\oplus k_3 \in \im{D}}}\ket{k_3}_{K_3}\nonumber\\
  &\quad \cdot \frac{1}{\sqrt{\abs{J_D}}} \sum_{u\in J_D} \ket{u}_{D^1_{x\oplus k_1}} \ket{z\oplus k_3}_{D^2_{u \oplus k_2}}\ket{(D_{u \oplus k_2})^c}_{(D^2_{u \oplus k_2})^c} .
    \end{align} In words, $\ket{\phi^{k_3}}$ is the same as $\vp\Gamma_2\ket{\psi_5}$ except that the index for $k_3$ is the complement set in order to restore the independence between $z$ and $k_3$. We start by writing state of the first term in \Cref{eq:p-psi5-split}. By adapting \eqref{eq:vp-gamma2-psi2-reference} for $\ket{\psi_5}$,\footnote{The adaptation consists only of extending the index $y$ of the first summation to include $0$, and changing $y,y$ to $y,0$, and hence denoting the queried database entry by $\hat{0}$ rather than $\hat{y}$.}\begin{align}
\vp\Gamma_2\ket{\psi_5}
    &=\sum_{y\in [N]}  \ket{\hat{y}}_Y \ket{\bot}_{D^3_x}
  \frac{1}{\sqrt{N}}\sum_{z\in [N]}  \sum_{\substack{D\in \mathbf D_{inj, t}}} \gamma_{y,0,D} \ket{\psi_{y,0,D}}\frac{1}{\sqrt{N}} \sum_{\substack{k_3\in[N]:\\
  z\oplus k_3 \not \in \im{D}}}\ket{k_3}_{K_3}\nonumber\\
  &\quad \cdot \frac{1}{\sqrt{\abs{J_D}}} \sum_{u\in J_D} \ket{u}_{D^1_{x\oplus k_1}} \ket{z\oplus k_3}_{D^2_{u \oplus k_2}}\ket{(D_{u \oplus k_2})^c}_{(D^2_{u \oplus k_2})^c}, 
\end{align} then \begin{align}
    \vp\Gamma_2\ket{\psi_5} + \ket{\phi^{k_3}}&=\sum_{y\in [N]}  \ket{\hat{y}}_Y \ket{\bot}_{D^3_x}
  \frac{1}{\sqrt{N}}\sum_{z\in [N]}  \sum_{\substack{D\in \mathbf D_{inj, t}}} \gamma_{y,0,D} \ket{\psi_{y,0,D}}\frac{1}{\sqrt{N}} \sum_{k_3\in[N]}\ket{k_3}_{K_3}\nonumber\\
  &\quad \cdot \frac{1}{\sqrt{\abs{J_D}}} \sum_{u\in J_D} \ket{u}_{D^1_{x\oplus k_1}}  \ket{z\oplus k_3}_{D^2_{u \oplus k_2}}\ket{(D_{u \oplus k_2})^c}_{(D^2_{u \oplus k_2})^c}\\
  &=\sum_{y\in [N]}  \ket{\hat{y}}_Y \ket{\bot}_{D^3_x}
    \sum_{\substack{D\in \mathbf D_{inj, t}}} \gamma_{y,0,D} \ket{\psi_{y,0,D}} \frac{1}{\sqrt{N}} \sum_{k_3\in[N]}\ket{k_3}_{K_3}\nonumber\\
  &\quad \cdot \frac{1}{\sqrt{\abs{J_D}}} \sum_{u\in J_D} \ket{u}_{D^1_{x\oplus k_1}}  \Bigl( \frac{1}{\sqrt{N}}\sum_{z\in [N]}\ket{z\oplus k_3}_{D^2_{u \oplus k_2}}\Bigr)\ket{(D_{u \oplus k_2})^c}_{(D^2_{u \oplus k_2})^c}\\
  &=\sum_{y\in [N]}  \ket{\hat{y}}_Y \ket{\bot}_{D^3_x}
    \sum_{\substack{D\in \mathbf D_{inj, t}}} \gamma_{y,0,D}\ket{\psi_{y,0,D}} \frac{1}{\sqrt{N}} \sum_{k_3\in[N]}\ket{k_3}_{K_3}\nonumber\\
  &\quad \cdot \frac{1}{\sqrt{\abs{J_D}}} \sum_{u\in J_D} \ket{u}_{D^1_{x\oplus k_1}}  \ket{+^n}_{D^2_{u \oplus k_2}}\ket{(D_{u \oplus k_2})^c}_{(D^2_{u \oplus k_2})^c},
\end{align} applying $\fc\ppar{2}$ to this gives \begin{align}
  &=\sum_{y\in [N]}  \ket{\hat{y}}_Y \ket{\bot}_{D^3_x}
    \sum_{\substack{D\in \mathbf D_{inj, t}}} \gamma_{y,0,D}\ket{\psi_{y,0,D}} \frac{1}{\sqrt{N}} \sum_{k_3\in[N]}\ket{k_3}_{K_3}\nonumber\\
  &\quad \cdot \frac{1}{\sqrt{\abs{J_D}}} \sum_{u\in J_D} \ket{u}_{D^1_{x\oplus k_1}} \ket{\bot}_{D^2_{u \oplus k_2}}\ket{(D_{u \oplus k_2})^c}_{(D^2_{u \oplus k_2})^c} .  
\end{align} This means the first term in \Cref{eq:p-psi5-split} becomes 0 since definition of $\pu^{em}$ is extended to act as identity on $\bot$ state.

Finally, we bound the second term in \Cref{eq:p-psi5-split} as \begin{align}
\norm{\ket{\phi^{k_3}}}&=\sqrt{\sum_{y\in [N]}\frac{1}{N}\sum_{z\in[N]} \sum_{\substack{D\in \mathbf D_{inj, t}}} \abs{\gamma_{y,0,D}}^2 \frac{1}{N} \sum_{\substack{k_3\in[N]:\\
  z\oplus k_3 \in \im{D}}} \frac{1}{\abs{J_D}} \sum_{u\in J_D}\norm{\ket{\psi_{y,0,D}}}^2}\\
  &\leq O(\sqrt{t/N}).
\end{align}

\end{enumerate}

\end{proof}

\end{document}